\documentclass[12pt]{article}

\usepackage{amsmath, amsthm, amssymb, amsfonts}
\DeclareMathOperator{\E}{\mathbb{E}}

\usepackage{graphicx}

\usepackage{appendix}
 
 \usepackage[dvipsnames]{xcolor}
 
\usepackage{placeins}

\usepackage{soul} 

\usepackage{adjustbox}

\makeatletter
\AtBeginDocument{%
  \expandafter\renewcommand\expandafter\subsection\expandafter
    {\expandafter\@fb@secFB\subsection}%
  \newcommand\@fb@secFB{\FloatBarrier
    \gdef\@fb@afterHHook{\@fb@topbarrier \gdef\@fb@afterHHook{}}}%
  \g@addto@macro\@afterheading{\@fb@afterHHook}%
  \gdef\@fb@afterHHook{}%
}
\makeatother

\makeatletter
\AtBeginDocument{%
  \expandafter\renewcommand\expandafter\subsection\expandafter{%
    \expandafter\@fb@secFB\subsection
  }%
}
\makeatother

\usepackage{color}

\usepackage[longnamesfirst]{natbib}	

\usepackage[top=2.5cm, left=2.5cm, bottom=2.5cm, right=2.5cm]{geometry}
\usepackage{setspace}	
\usepackage{array}
\usepackage{adjustbox}
\usepackage{multirow}

\usepackage{bm} 
\usepackage[hidelinks]{hyperref} 


\newtheorem{prop}{Proposition}

\begin{document}

\begin{onehalfspace} 

\title{\textbf{Dynamic Effects of Persistent Shocks}\\
}

\author{ \begin{tabular}{ccc}
Mario Alloza\footnotemark[1] & Jes\'{u}s Gonzalo\footnotemark[1] & Carlos Sanz\footnotemark[1]\thanks{We thank Fabio Canova, Jes\'us Fern\'andez-Villaverde, Alessandro Galesi, Gergely G\'anics, Juan F. Jimeno, Mikkel Plagborg-M{\o}ller, Juan Rubio-Ram\'irez, Enrique Sentana,  and seminar participants at the I Workshop of the Spanish Macroeconomics Network (Universidad P\'ublica de Navarra), Bank of Spain, CFE 2018 (University of Pisa), II Workshop in Structural VAR models (Queen Mary University of London), VII Workshop on Empirical Macroeconomics (Ghent University), 2019 American Meeting of the Econometric Society (University of Washington), and 2019 edition of the Padova Macro Talks for insightful comments.
Alloza: \href{mailto:m.alloza@bde.es}{m.alloza@bde.es}. Gonzalo: \href{mailto:jesus.gonzalo@uc3m.es}{jesus.gonzalo@uc3m.es}. Sanz: \href{mailto:carlossanz@bde.es}{carlossanz@bde.es}.}\\
\begin{footnotesize}Bank of Spain \end{footnotesize} & 
\begin{footnotesize}Universidad Carlos III de Madrid \end{footnotesize} &  \begin{footnotesize}Bank of Spain \end{footnotesize} \\ 
\end{tabular} }

\date{\vspace{1cm}  
 \today }

\maketitle

\begin{abstract}
We provide evidence that many narrative shocks used by prominent literature are persistent. We show
that the two leading methods to estimate impulse responses to an independently identified shock (local projections and distributed lag models) treat persistence differently, hence identifying different objects. We propose corrections to re-establish the equivalence between local projections and distributed lag models, providing applied researchers with methods and guidance to estimate their desired object of interest. We apply these methods to well-known empirical work and find that how persistence is treated has a sizable impact on the estimates of dynamic effects.
\end{abstract}

\thispagestyle{empty}


\noindent \textit{JEL} classification: C32, E32, E52, E62.

\end{onehalfspace}



\newpage
\section{Introduction}
Estimating the impact of economic shocks is a crucial aspect of macroeconomics. 
To identify economically meaningful shocks, the literature has traditionally relied on systems of equations coupled with restrictions implied by economic theory. 
Recently, researchers are increasingly using narrative identification, e.g., looking at written official documentation or newspapers and exploiting arguably exogenous variation in these series.\footnote{See \citet{romer2004measure}, \citet{romer2010macroeconomic}, or \citet{ramey2018government} for prominent examples of narrative identification.} While its focus on identifying exogenous variation is appealing, the lack of restrictions in narrative methods yields objects with less standard time series properties.

In this paper, we analyze how the presence of persistence in narrative shocks affects the identification and estimation of their dynamic effects, providing empirical researchers with methods and guidance to deal with this issue.\footnote{Throughout the paper we use the term \textit{persistence} as a phenomenon captured or reflected by \textit{serial correlation}, a testable condition. We use both terms interchangeably. }

We begin by showing that many narrative shocks used by prominent literature are serially correlated. In particular, we systematically test for serial correlation in eight shocks used in leading economics journals. We find evidence of serial correlation in seven of them. The presence of persistence in the shock does not necessarily preclude these variables from being categorized as ``shocks'' following standard definitions of aggregate shocks. More concretely, according to \citet{ramey2016macroeconomic}, a shock should represent unanticipated movements. What this condition implies is that shocks are unforecastable, i.e., they are forecast errors. In particular, when the forecasting loss function is not quadratic, for instance, the check function, the forecasting errors may not be a martingale difference sequence (m.d.s) and therefore could be serially correlated. However, serial correlation poses additional challenges for the identification of the macroeconomic experiment of interest.

When estimating the dynamic response of some variable to a serially correlated shock, some part of this persistence may be passed on to the impulse response function (IRF). Hence, a researcher may want to identify two objects of interest: the response \textit{as if} the shock were uncorrelated, i.e., to a counter-factual serially uncorrelated shock ($\mathcal{R}(h)^{*}$), or the response to the shock as it is, i.e.,  including the effect of persistence in the IRF ($\mathcal{R}(h)$). 
Deciding for one or the other depends on what specific question the researcher is trying to address. On the one hand, $\mathcal{R}(h)^{*}$ allows to compare effects with those obtained from a theoretical or empirical model, and facilitates comparisons across different types of shocks (e.g., monetary versus fiscal shocks) or across countries. On the other hand, $\mathcal{R}(h)$ is more appropriate if  the researcher is interested in evaluating the \textit{most likely dynamic response} of a variable to a shock based on historical data. Regardless of which object is preferred by the researcher, the difference between $\mathcal{R}(h)$  and $\mathcal{R}(h)^{*}$ is informative about how much of the dynamic transmission of a shock is due to the presence of persistence.

We consider the two most popular methods to estimate impulse responses when a shock has already been identified (e.g., using narrative methods). These are local projections (LPs) (\citet{jorda2005estimation}) and distributed lag models (DLMs).\footnote{By DLMs  we refer to single-equation regressions of an outcome variable against the contemporaneous value and lags of the shock with or without an autoregressive component. These methods are also known as truncated moving average regressions. These specifications are frequent in the applied literature---see, e.g., \citet{romer2004measure}, \citet{cerra2008myth}, \citet{romer2010macroeconomic}, \citet{alesina2015output}, \citet{arezki2017news}, and \citet{coibion2018cyclical}.} We show that, if there is no serial correlation, the two methods identify the same object. However, we demonstrate that this equivalence breaks down in the presence of serial correlation. In this case, LPs identify $\mathcal{R}(h)$ while DLMs regressions identify  $\mathcal{R}(h)^{*}$. The intuition is that LPs compute the response at horizon $h$ by regressing the outcome variable in $t+h$ against the shock in time $t$. Since the standard setting does not account for how the shock evolves between $t$ and $t+h$, the responses include two components: an economic effect (the economic impact of the shock on the endogenous variables) and an effect that exclusively depends on the degree of serial correlation of the shock. By contrast, DLMs  implicitly account for the evolution of the shock, hence identifying the effect as if the shock were not persistent. 

While this result might seem discouraging, we then show that it is possible to adjust both estimating methods to obtain the desired object of interest. Consider a researcher who wants to use LPs and is interested in identifying $\mathcal{R}(h)^{*}$. As mentioned, if she runs standard LPs with a persistent shock, she will identify $\mathcal{R}(h)$ instead. Perhaps surprisingly, the most obvious solution of including lags of the shock will not address this issue. However, we show that, by including \textit{leads} of the shock, she will recover $\mathcal{R}(h)^{*}$.  Likewise, we show how standard DLMs can be adapted so that they identify $\mathcal{R}(h)$.

To illustrate how our methods work, we consider an actual empirical application, which also serves to assess the quantitative relevance of persistence in a real case by comparing estimates of $\mathcal{R}(h)$ and $\mathcal{R}(h)^{*}$. In particular, we consider \citet{ramey2018government}'s LPs estimation of the dynamic effects to a shock constructed from news about future changes in defense spending. We find that, after two years, the responses that exclude the effect of persistence in the shock are about 40\% lower than the original Ramey and Zubairy (2018)'s estimates. The effect of serial correlation also seems to have an effect on the short-run response of fiscal multipliers during recessions.
In the appendix, we consider additional applications, based on \citet{guajardo2014expansionary}, \citet{romer2004measure}, \citet{gertler2015monetary}, and \citet{romer2010macroeconomic}. Overall, we find that how persistence is treated can have a sizable impact on the estimated effects.
 

The results of this paper generalize in at least three important aspects. First, the results of the (lack of) equivalence between LPs and DLMs when the shock is persistent carry over to multivariate settings popularly used in the empirical literature. Building on a result by \citet{plagborg2018local}, we show that the dynamic response from a VAR with the shock embedded as an endogenous variable is equivalent to that of a VAR with the shock included as an exogenous variable only when that shock has no serial correlation.\footnote{This result arises because a VAR with a shock as an exogenous variable (often known as VAR-X) can be seen as  multivariate generalization of a DLM (see  \citet{mertens2012empirical} or \citet{favero2012tax} for examples of VAR-X specifications). Furthermore, \citet{plagborg2018local} show that, under some assumptions, LPs are equivalent to a VAR when the shock is included as an endogenous variable (as in  \citet{bloom2009uncertainty} or \citet{ramey2011identifying}).} We believe this result has relevant practical implications for applied macroeconomic researchers. Second, we also show that our results generalize to specific contexts where a researcher employs an instrument in a LP setting (also known as LP-IV). 
Lastly, a researcher interested in using LPs to uncover  the dynamic relations of two variables may be interested in including leads of a third variable to construct counterfactual responses as if the behavior of that third variable had remained constant over the response horizon. This can be seen as the LP counterpart of constructing counterfactual responses in a VAR that allow to separate a direct effect of a regressor on a dependent variable from other indirect effects. This procedure has been frequently used in the empirical VAR literature.\footnote{See, for example, \citet{bernanke1997systematic}, \citet{sims2006does}, or \citet{bachman2012confidence}.  In recent research,  \citet{cloyne2020decomposing} propose an alternative method based on a Blinder-Oaxaca-type decomposition.}

Our paper makes four contributions to the literature. First, we formally and systematically test for the presence of serial correlation in shocks used by previous work. Although the issue of persistence in shocks has been noted before,\footnote{\citet{ramey2016macroeconomic} finds that the time aggregation required to convert the shock in \citet{gertler2015monetary} to monthly frequency, inserts serial correlation. \citet{miranda2018transmission} corroborate this finding, by regressing the shock on four lags and testing their joint significance. They also find that other measures of monetary shocks such as \citet{romer2004measure} exhibit serial correlation.} we believe we are the first to formally and systematically test for serial correlation in prominent narratively-identified shocks.

Our second contribution is to show that, while both LPs and DLMs  identify the same object if the shock is serially uncorrelated, this equivalence breaks down in the presence of persistence. \citet{plagborg2018local} prove that LPs and VAR methods identify the same impulse responses when both methods have an unrestricted lag structure. This result formalizes some of the examples provided in \citet{ramey2016macroeconomic}, which implies that different identification schemes in a VAR setting can be implemented in a LP context. Our result builds on a different premise: we consider the cases where the shock has already been identified using narrative measures and the researcher wants to use LPs or DLMs  to estimate dynamic effects. 

Our third contribution is to provide methods to re-establish the LP-DLM equivalence when there is persistence, providing applied researchers with a menu of options to identify their desired object of interest. In this regard, our method of adding leads to LPs is related to the tradition in factor analysis by \citet{geweke1981maximum} and on the DOLS estimation of cointegration vectors (\citet{stock1993simple}). \citet{dufour1998causality}  introduce leads in some of their IRFs to study causality at different horizons. 
 \citet{faust2011efficient} find that including ex-post forecast errors  results in an accuracy improvement when  forecasting excess bond and equity returns. More recently, \citet{teulings2014economic} find that  estimating dynamic effects of a dummy variable (e.g., banking crisis) in a panel data context with fixed effects and LPs suffers from a negative small-sample bias, since the estimation of the fixed effect picks up the value of future realization of the dummy variable. The authors show that this bias is attenuated either by increasing the sample size or by including future realizations of the dummy variable over the response horizon.\footnote{By contrast, the difference between LPs and DLMs  that we identify is not due to a bias in the estimates, but instead to differences in identification due to the persistence of the shock. Since our problem still persists asymptotically, increasing the sample does not reduce the LP-DLM difference. Additionally, this difference is not necessarily negative, but will depend on the nature of the data generating process that drives the persistence.}

Finally, we speak to some recent and well-known empirical work on the effects of monetary and fiscal policy (\citet{ramey2018government} \citet{guajardo2014expansionary}, \citet{romer2004measure}, and \citet{gertler2015monetary}). Our contribution is to apply our methods to these works and re-assess their empirical evidence. We do not claim that any of these papers is ``wrong''. Rather, what our results indicate is that the correct interpretation of their results depends on the desired object of interest and the employed estimating method.

The rest of the paper proceeds as follows. Section \ref{sec:evidence} provides evidence of serial correlation in shocks used by previous work. Section \ref{sec:econometric} describes that LPs and DLMs  treat persistence differently, and proposes a solution to re-establish the equivalence between them. It also provides simulations to help understand the results. Section~\ref{sec:discussion} discusses the previous findings and the options available to applied researchers working with a persistent shock. Section~\ref{sec:applications} lays out an application. Section~\ref{sec:conclusions} concludes. The online appendix contains proofs of the theoretical results and further material, including the generalization of the results to VAR and IV settings, additional robustness exercises, and other empirical applications.
\section{Evidence and implications of serial correlation in shocks}\label{sec:evidence}



When shocks are identified from within an empirical model, the researcher imposes  a set of restrictions to recover shocks that can be economically meaningful. 
Typically, this implies that the resulting shocks are well-behaved and display some statistical features that might be seen as desirable---in particular, no persistence. Alternatively, shocks may be identified without the use of a model, for example, by using \textit{narrative} methods. This alternative identification relies on the existence of historical sources, such as official documentation, periodicals, etc., from which a shock variable is constructed. In this section, we provide evidence that it is common that shocks identified this way are persistent. We then take stock on this finding in light of \citet{ramey2016macroeconomic}'s canonical definition of a shock.

We study eight aggregate shocks used by prominent literature on monetary and fiscal policy. Some of these shocks are identified using narrative methods, while some employ alternative strategies such as timing restrictions using high-frequency methods.\footnote{ In particular, \citet{romer2010macroeconomic} and \citet{cloyne2013tax} construct measures of exogenous tax changes for the US and the UK, respectively. The authors classify legislated tax measures according to the motivation, as reflected in official documentation, and consider those tax changes that are the result of causes non-related to the state of the economy. In a similar vein, \citet{ramey2018government} construct a measure of government spending shocks by looking at the announcements of future changes in defense spending. \citet{guajardo2014expansionary} 
construct a series of fiscal consolidations in OECD countries motivated by a desire to reduce the deficit (as opposed to motivated by current or prospective economic conditions). \citet{romer2004measure} and \citet{cloyne2016monetary} identify exogenous changes in monetary policy by looking at the minutes and discussion of the monetary policy committees of the Federal Reserve and Bank of England, respectively (they also orthogonalize the resulting series using forecastable information available at that time). Alternatively, \citet{gertler2015monetary} identify a proxy of monetary policy shocks using  high frequency surprises around policy announcements. Lastly, \citet{arezki2017news} construct a measure of news shocks based on the date and size of worldwide giant oil discoveries. While some of these papers employ auxiliary regressions to isolate forecastable information, all have in common that the shocks have not been exclusively identified from a time series model.}

To test for the presence of persistence we use a \emph{portmanteau}-type test following \citet{box1970distribution}.\footnote{We implement the small sample correction following \citet{ljung1978measure}. For the cases of  \citet{arezki2017news} and \citet{guajardo2014expansionary}, which refer to panel data, we test serial correlation using a generalized version of the autocorrelation test proposed by \citet{arellano1991tests} that specifies the null hypothesis of no autocorrelation at a given lag order.} The null hypothesis is that the data are not serially correlated. We test for the presence of autocorrelation in 40 periods, although results are robust to different horizons (see Table~\ref{tab:survey_rob}).


The results from these tests are displayed in Table~\ref{tab:survey}.  Out of the eight considered shocks, six show very large test statistics that result in rejections of the hypothesis of serial uncorrelation for any level of significance. One of them (\citet{romer2004measure}) displays some degree of serial correlation which leads to failure to reject the null hypothesis only for significance levels above 5\%.\footnote{The hypothesis of serial uncorrelation is rejected for significance levels below 5\% when considering fewer lags in the test or when considering a longer series (with updated data) from \citet{coibion2012monetary}. The presence of some degree of autocorrelation is shown in Panel E of Figure~\ref{fig:ac}.} As further evidence of the presence of serial correlation in the above series, Figure~\ref{fig:ac}  plots the associated correlograms. \citet{romer2010macroeconomic} constitutes the only considered shock for which we fail to detect the presence of persistence.\footnote{Persistence may have different origins. In some instances, it arises because of the method used to convert a nominal series into real terms. For example, \citet{cloyne2013tax} and \citet{arezki2017news} divide their series by lagged GDP, while  \citet{ramey2018government} use the GDP deflator and a measure of trend GDP.  In other instances, the serial correlation arises because of the mapping between different time frequencies. This is usually the case with the identification of monetary policy shocks, such as \citet{romer2004measure}, \citet{gertler2015monetary}, or \citet{cloyne2016monetary}, where daily monetary changes are converted into monthly series.  Finally, there are other shocks that are more likely to appear together, because of their multi-period nature (for example, episodes of fiscal consolidations, as identified by \citet{guajardo2014expansionary}, tend to be spread over the course a few years) or because  they cluster around events like wars (as in \citet{ramey2018government}). }

\begin{table}
\renewcommand{\arraystretch}{1.5} 
\caption{Persistence in macroeconomic shocks} 
\label{tab:survey}
\begin{center}
{\small
\begin{tabular}{ cccc } 
 \hline
  \\[-1.4em] 
paper&type of shock&Box-Pierce (40) test&p-value\\
  \\[-1.4em]
\hline
\citet{arezki2017news}&news about oil discoveries&177.903&0.000\\
\citet{cloyne2013tax}&tax  (UK)&98.751&0.000\\
\citet{cloyne2016monetary}&monetary policy  (UK)&84.422&0.000\\
\citet{gertler2015monetary}&monetary policy  (US)&124.568&0.000\\
\citet{guajardo2014expansionary}&fiscal consolidations&185.810&0.000\\
\citet{ramey2018government}&government spending&182.950&0.000\\
\citet{romer2004measure}&monetary policy  (US)&53.758&0.072\\
\citet{romer2010macroeconomic}&tax  (US)&19.023&0.998\\
 \hline
\end{tabular}
 }
\end{center}
\vspace{-0.2cm}
 {\footnotesize The third column implements the \citet{box1970distribution} test of serial correlation using the small sample correction following \citet{ljung1978measure}. The null hypothesis of this test assumes that the data are not serially correlated within 40 periods. For \citet{arezki2017news} and \citet{guajardo2014expansionary}, which refer to panel data, we use a generalized version of the autocorrelation test proposed by \citet{arellano1991tests}. The serial correlation test yields p-values smaller than 0.05 when testing the shocks of \citet{romer2004measure} with fewer lags or when using the updated data from \citet{coibion2012monetary} (p-value drops to 0.0041). \citet{ramey2018government} use extended data from \citet{ramey2011identifying}.}
\end{table}



According to the canonical definition (\citet{ramey2016macroeconomic}), empirical shocks should (i) be exogenous to current and lagged endogenous variables, (ii) be uncorrelated to other exogenous shocks, and (iii) represent unanticipated movements (or news about future shocks). While one might think that the presence of persistence violates the third condition, this is not necessarily the case. When the forecasting loss function is the quadratic one, it is well known that the forecasting errors must be a m.d.s with respect to some information set and therefore uncorrelated.\footnote{ See \citet{granger2006forecasting} and \citet{lee2008loss} for a description and analysis of loss functions.} This is the case when the shocks come directly from a conditional expectation model, like a VAR model. When the forecasting loss function is not quadratic, for instance, the check function (popular in quantile regressions), the forecasting errors are not a m.d.s and therefore they could be serially correlated. They still are forecasting errors (satisfy (iii)) but are serially correlated. 

This indicates that serially-correlated shocks can still be labeled ``shocks'' according to the previous definition. However, even if a researcher always operates under the quadratic loss function and considers that serially-correlated shocks should not be called ``shocks'', in the rest of the paper we show that such shocks can still provide valuable information for empirical analysis. 

\section{Theoretical framework} \label{sec:econometric}

We consider the following VAR as the data generating process: 
\begin{eqnarray}  \label{eq:DGP_big}
	\bm {y_t} &=& \sum_{\ell=1}^\infty \bm{A_\ell y_{t-\ell}} + \sum_{q=0}^\infty \bm{\delta_q} x_{t-q} + \bm{u_t} \nonumber \\ 
	x_t &=& \sum_{r=1}^\infty \gamma_r x_{t-r} + \varepsilon_t,
\end{eqnarray}

where $\bm{y_t}$ is a vector of endogenous time series, $x_t$ is a strictly exogenous variable such that $\E \left( \bm{u_t} | \bm {y_{t-s}}, x_{t-p}  \right) $ $\forall s>0$, $p \gtreqless 0$, and $\bm{u_t}$ and $\varepsilon_t$ are a vector and a scalar i.i.d. variables, with mean and variance given by $\bm{u_t} \sim (\bm{0},\,\bm{\Sigma_u}^{2})$ and $\varepsilon_t \sim (0,\sigma_\varepsilon^{2})$, respectively. Following the evidence discussed in the previous section, $x_t$ is considered to be a shock identified using narrative methods and is allowed to be persistent.

This general framework encompasses several empirical specification often found in the literature. For example, when ignoring the second equation, system~\eqref{eq:DGP_big}  becomes a VAR with an exogenous variable (or VAR-X).\footnote{See, for example, \citet{mertens2012empirical} or \citet{favero2012tax}, which assume $\ell$ and $q$ are finite numbers.} Additionally, when $\bm{y_t}$ is a scalar and $\bm{A_\ell} = \bm{0}$ $\forall \ell$, system ~\eqref{eq:DGP_big}  becomes a DLM.\footnote{As in \citet{romer2004measure} or \citet{romer2010macroeconomic}.} Alternatively, when $x_t$ is instead included in the vector of endogenous variables $\bm{y_t}$, system~\eqref{eq:DGP_big} becomes a standard VAR.\footnote{As in \citet{bloom2009uncertainty}  or \citet{ramey2011identifying}.} We explore the implications of this last representation in Appendix~\ref{app:VAR}.

Without loss of generality, we consider a simpler version of system~\eqref{eq:DGP_big} with $\bm A_\ell=\bm{0} $ $\forall \ell$, $\bm \delta_q=0 $ $\forall q>0$ and  $\gamma_r=0 $ $\forall r>1$: 
\begin{eqnarray} \label{eq:DGP}
	y_t &=& \delta x_{t} + u_t \nonumber \\ 
	x_t &=& \gamma x_{t-1} + \varepsilon_t,
\end{eqnarray}
where $y_t$ is now the economic outcome variable for interest (for example, GDP), $x_t$ is an economic shock (e.g., a fiscal or monetary policy shock) which is strictly exogenous $\E \left( {u_t} |  x_{t-p}  \right) $ $\forall p \gtreqless 0$, and $u_t$ and $\varepsilon_t$ are i.i.d variables with mean and variance given by $u_t \sim (0,\,\sigma_u^{2})$ and $\varepsilon_t \sim (0,\sigma_\varepsilon^{2})$, respectively.  $\delta$ measures the contemporaneous impact of variable $x_t$ on $y_t$ and is the main parameter of interest.

The data generating process described by system~\eqref{eq:DGP} is intentionally simple to illustrate how the dynamic relationship between the dependent variable $y_t$ and the shock $x_t$ depends on the persistence of the latter. Importantly,  the obtained results also arise in more complex settings when we incorporate more general characteristics as in system~\eqref{eq:DGP_big}.\footnote{For example, in Subsection 3.3, we consider  models that also include persistence in the dependent variable and lagged effects of the shock. Appendix \ref{sec:IV} proposes a DGP that calls for the use of instruments in LP regressions. Appendix~\ref{app:simulreal} provides an alternative specification where the degree of serial correlation in the shock $x_t$ is taken from the actual data, instead of following an autoregressive process.}

We are interested in recovering the response of our variable of interest $y_t$ when a shock $x_t$ hits the system in period $t$. We consider two different IRFs. The first one, denoted by $\mathcal{R}(h)$ for period $h$, is:

\begin{equation} \label{eq:IRF_per}
	\mathcal{R}(h) = \E \left[ y_{t+h} | x_t=1, \Omega_{t-1} \right]  - \E \left[ y_{t+h} | x_t=0, \Omega_{t-1} \right],
\end{equation}

where $\Omega_{t-1}$ represents all the history of previous realizations of $\varepsilon_t$ and $x_t$ up to period $t-1$. Importantly, note that the above definition does not condition for future realizations of $x_t$. Hence, if $\gamma \neq 0$, an initial unit impulse in $x_t$ does not imply that $x_{t+j} = 0$.\footnote{This impulse response is equivalent to $\mathcal{R}(h) = \E \left[ y_{t+h} | \varepsilon_t=1, \varepsilon_{t+1}=0,...,\varepsilon_{t+h}=0, \Omega_{t-1} \right]  - \E \left[ y_{t+h} | \varepsilon_t=0, \varepsilon_{t+1}=0,...,\varepsilon_{t+h}=0, \Omega_{t-1} \right]$. See, for example, \citet{koop1996impulse}.} In other words, equation~\eqref{eq:IRF_per} describes dynamic responses that include the possible persistence of the shock $x_t$. For example:
\begin{eqnarray*}
\mathcal{R}(0) &=&  \frac{\partial y_{t}} {\partial x_{t}} = \delta \\
\mathcal{R}(1) &=&  \frac{\partial y_{t+1}} {\partial x_{t}} = \delta \gamma\\
\mathcal{R}(2) &=&  \frac{\partial y_{t+2}} {\partial x_{t}} = \delta \gamma^2\\
&&\dots
\end{eqnarray*}

However, the researcher might also be interested in the response to the shock \textit{as if} the shock had no persistence. We call this second IRF $\mathcal{R}(h)^{*}$ and define it as:
\begin{equation} \label{eq:IRF_iid}
\mathcal{R}(h)^* = \E \left[ y_{t+h} | x_t=1, x_{t+1},..., x_{t+h}, \Omega_{t-1} \right]  - \E \left[ y_{t+h} | x_t=0, x_{t+1},..., x_{t+h}, \Omega_{t-1} \right].
\end{equation}

Contrary to $\mathcal{R}(h)$, $\mathcal{R}(h)^*$ explicitly controls for future realizations of $x_t$ so that it describes dynamic responses that do not incorporate the effect of persistence (regardless of the value of $\gamma$), i.e., the responses are observationally equivalent to those that would arise from a data generating process with $\gamma=0$:\footnote{The definition of $\mathcal{R}(h)^*$ is not new. When $x_t$ is the shock variable of interest, this impulse response is referred to as the ``traditional impulse response function'' by \citet{koop1996impulse}: $\mathcal{R}(h)^* = \E \left[ y_{t+h} | x_t=1, x_{t+1}=0,...,x_{t+h}=0, \Omega_{t-1} \right]  - \E \left[ y_{t+h} | x_t=0, x_{t+1}=0,...,x_{t+h}=0, \Omega_{t-1} \right]$. It provides an answer to the question ``what is the effect of a shock of size 1 hitting the system at time $t$ on the state of the system at time $t+h$ given that no other shocks hit the system?''.}
\begin{eqnarray*}
\mathcal{R}(0)^* &=&  \frac{\partial y_{t}} {\partial x_{t}} = \delta \\
\mathcal{R}(1)^* &=&  \frac{\partial y_{t+1}} {\partial x_{t}} \Bigr|_{\substack{x_{t+1}}} = 0\\
\mathcal{R}(2)^* &=&  \frac{\partial y_{t+2}} {\partial x_{t}} \Bigr|_{\substack{x_{t+1},x_{t+2}}} = 0\\
&&\dots
\end{eqnarray*}

Note that, if $\gamma=0$ (the shock is not persistent), then $\mathcal{R}(h)=\mathcal{R}(h)^*$  $\forall$ $h$. By contrast, if $\gamma \neq 0$, then $\mathcal{R}(h) \neq \mathcal{R}(h)^*$  $\forall$ $h>0$.

\subsection{Differences between DLMs  and LPs under persistence}
We now consider the two most frequently used methods to estimate impulse responses when a shock is independently identified, DLMs  and LPs, and compare the objects that they identify when the shock is persistent. We first consider the case of DLMs. The use of these models is widespread in applied macroeconomics.\footnote{See, for example, \citet{romer2004measure}, \citet{cerra2008myth}, \citet{romer2010macroeconomic}, \citet{alesina2015output}, \citet{arezki2017news}, \citet{coibion2018cyclical} for interesting applications based on DLM methods, or  \citet{baek2019abcs} for a discussion of their properties. As mentioned in the introduction,  these methods are also a special case of more general specifications such as VARs with exogenous variables (or VAR-X). We develop this point further in Appendix~\ref{app:VAR}, when generalizing some of the results of the paper.} In the case of system~\eqref{eq:DGP},  note that we can recover the response function $\mathcal{R}(h)^{DLM}$ using the following regression:\footnote{This regression should include as many lags as the response horizon $h=0,1,\ldots,H$.}
\begin{equation}
\label{eq:MA}
y_t = \theta_0 x_t + \theta_1 x_{t-1} + \theta_2 x_{t-2} + \theta_3 x_{t-3} + \theta_4 x_{t-4} + \ldots + e_t,
\end{equation}

\noindent and it follows that $\mathcal{R}(h)^{DLM} = \frac{\partial y_{t+h}} {\partial x_{t}} =\theta_h$ $\forall$ $h$. 

The second main method to compute impulse responses is LPs, proposed by \citet{jorda2005estimation}. LPs are more robust to certain sources of misspecification and for this reason, their use  has increased in recent times (see \citet{ramey2016macroeconomic} for examples). LPs compute impulse responses  by estimating an equation for each response horizon $h=0,1,\ldots,H$:
\begin{equation}
\label{eq:LP}
	y_{t+h} = \delta_{h} x_{t} + \xi_{t+h}, 
\end{equation}
\noindent where the sequence of coefficients $\{\delta_{h}\}_{h=0} ^{H}$ determines the response of the variable of interest $\mathcal{R}(h)^{LP}=\delta_h$ for each horizon $h$.\footnote{Unrelated to our case at hand, note that  the structure of the LPs induce serial correlation in the residuals  $\xi_{t+h}$. This is usually corrected by computing autocorrelation-robust standard errors (\citet{jorda2005estimation}). See \citet{olea2020local}  for a recent contribution on inference in LPs.} 

We now consider under which conditions both methods identify the same objects.

\begin{prop} \label{prop:equivalence}
Given the data generating process described by system~\eqref{eq:DGP}, if the shock $x_t$ is serially uncorrelated, then the response functions identified by DLMs  and LPs are equal for all response horizons, that is:  \\
\indent If $\gamma = 0$, then $\mathcal{R}(h)^{DLM} =  \mathcal{R}(h)^{LP}=  \mathcal{R}(h)^{*} = \mathcal{R}(h)$ $\forall h$.
\\
If the shock is serially correlated, then the response functions identified by DLMs  and LPs are different for all $h>0$: \\
\indent If $\gamma \neq 0$ and $h=0$, then $\mathcal{R}(h)^{DLM} =  \mathcal{R}(h)^{LP}=  \mathcal{R}(h)^{*} = \mathcal{R}(h)$. \\
\indent If $\gamma \neq 0$ and $h\geq 1$, then $\mathcal{R}(h)^{DLM} =    \mathcal{R}(h)^{*} \neq \mathcal{R}(h)^{LP}= \mathcal{R}(h)$.
\end{prop}
\begin{proof}
See Appendix~\ref{app:proof1}.
\end{proof}



Following the above proposition, when $\gamma\neq0$, LPs recover a dynamic response that includes three dynamic effects: (i) the effect that $x_t$ has directly on $y_{t+h}$ (due to a lagged impact of the shock), (ii) the effect that $x_t$ has through the persistence of $y_t$, and (iii) the effect that $x_t$ has on $y_{t+h}$ through $x_{t+h}$ (since $cov(x_t,x_{t+h}) \neq 0$ when $\gamma\neq0$). The first two effects are independent of $\gamma$ and are shut down in our simple specification of system~\eqref{eq:DGP} (we will incorporate them in our simulation exercises in the next subsection). The last effect (the \textit{persistence} effect of $x_t$) drives the difference between $\mathcal{R}(h)^{DLM}$ and  $\mathcal{R}(h)^{LP}$. In particular,  $\mathcal{R}(h)^{LP} = \mathcal{R}(h)^{} = \delta \gamma^h$, while $\mathcal{R}(h)^{DLM} =\mathcal{R}(h)^{*} = 0$ for all $h \geq 1$. 

To understand why LPs, unlike DLMs, incorporate this third effect, consider the LPs when $h=1$:
\begin{equation}
\label{eq:LPh1}
	y_{t+1} = \delta_{1} x_{t} + \xi_{t+1}, 
\end{equation}

\noindent where $\delta_{1}=\mathcal{R}(1)^{LP}$. The direct effect of $x_t$ on $y_{t+1}$ is 0. If $x_t$ had no persistence, then $\delta_1$ would be 0. However, when $\gamma \neq 0$, we can use system~\eqref{eq:DGP} to express $y_{t+1}$ as a function of $x_t$:
\begin{eqnarray*}
	y_{t+1} &=& \delta x_{t+1} + u_{t+1} \\
		&=& \delta \left(\gamma x_{t} + \varepsilon_{t+1}  \right)+ u_{t+1} \\
		&=& \delta \gamma x_{t} + u^*_{t+1}, 
\end{eqnarray*}
\noindent where $u^*_{t+1}=\delta \varepsilon_{t+1}+u_{t+1}$. This shows that the coefficient $\delta_1$ in equation~\eqref{eq:LPh1} will also recover the persistence effect of $x_t$: $\delta_1  = \delta\gamma$. 
The intuition is that between period $t$ and period $t+1$, $x_t$ affects $x_{t+1}$ when $\gamma \neq 0$. Since $x_{t+1}$ is not a regressor in equation~\eqref{eq:LPh1}, then this effect is absorbed by $\delta_1$.\footnote{This omitted variables problem is also briefly mentioned in \citet{alesina2015output} in the particular context of fiscal consolidation plans.} 

When impulse responses are identified using DLMs, the treatment of the persistence of $x_t$ is different. Consider a version of equation~\eqref{eq:MA} expressed in terms of $t+1$:
\begin{equation}
\label{eq:MA_fw1}
y_{t+1} = \theta_0 x_{t+1} + \theta_1 x_{t} + \theta_2 x_{t-1} + \theta_3 x_{t-2} + \theta_4 x_{t-3} +  \ldots + e_{t+1}. 
\end{equation}

\noindent As noted earlier, the sequence of coefficients $\theta_h$ determines the response function. Consider the response when $h=1$, i.e., $\mathcal{R}(1)^{DLM} = \theta_1$. Note that, while we know from system~\eqref{eq:DGP} that  $\frac{\partial y_{t+1}} {\partial x_{t}} = \delta \gamma$, the coefficient recovered by $\theta_1$ is indeed $\left.  \frac{\partial y_{t+1}} {\partial x_{t}} \right | _{x_{t+1}} = 0$. That is, since the DLM controls for $x_{t+1}$, the persistence effect of $x_t$ is accounted for.

In other words, DLMs  identify: 
\begin{equation*}
	\mathcal{R}(h) ^{DLM} = \E \left[ y_{t+h} | x_t=1, \Omega_{t-1}, x_{t+h-1},..., x_{t+1}  \right]  - \E \left[ y_{t+h} | x_t=0, \Omega_{t-1}, x_{t+h-1},..., x_{t+1} \right],
\end{equation*}
while LPs identify:
\begin{equation*}
	\mathcal{R}(h) ^{LP} = \E \left[ y_{t+h} | x_t=1, \Omega_{t-1} \right]  - \E \left[ y_{t+h} | x_t=0, \Omega_{t-1} \right].
\end{equation*}

Note that the difference between $\mathcal{R}^{LP}$ and $\mathcal{R}^{DLM}$  is positive (negative) when $\gamma >0$ ($\gamma <0$). In empirical applications, $\gamma$ may be positive or negative.\footnote{For example, $\gamma$ seems to be positive in \citet{ramey2018government}, and negative in \citet{romer2004measure}.}

\subsection{Reestablishing the equivalence between DLMs  and LPs}

In this subsection we lay out two methods that can render the responses from DLMs  and LPs identical, even under the presence of persistence. 
\subsubsection{Adapting LPs to exclude the effect of serial correlation}
A researcher may be interested in recovering responses as if the shock were serially uncorrelated ($\mathcal{R}(h)^{*}$). (We discuss in Section \ref{sec:discussion} when the object of interest may be $\mathcal{R}(h)^{*}$, or $\mathcal{R}(h)$ instead.) However, we have shown that $\mathcal{R}^{LP} (h) \neq \mathcal{R}(h)^{*}$ if $\gamma \neq 0$ and $h\geq1$. 

Two apparent methods to avoid LPs picking up the effect of persistence in $x_t$ are: (i) to include lags in the regression~\eqref{eq:LP}, or (ii) to replace $x_t$ with the error term that purges out the persistence:
\begin{equation}\label{eq:epshat}
 {\varepsilon}_t = x_t-{\gamma}x_{t-1}.
 \end{equation}
However, neither of these methods  yields $\mathcal{R}^{*} (h)$. 
The reason is that  replacing $x_t$ with ${\varepsilon}_t$ does not include any further information between $t$ and $t+h$, so the responses of the dependent variable will still be affected by $x_{t+h}$. 
This point is further developed in Appendix~\ref{sec:app_lags}. 

A third potential method to exclude the effect of persistence would be recasting system~\eqref{eq:DGP} as a VAR that includes the shock as an endogenous variable. However, since in this case LPs and a VAR would identify the same impulse responses (see \citet{plagborg2018local}) the VAR responses would also include an effect due to the persistence of the shock---we explore this in more detail in Appendix~\ref{app:VAR}.

Instead, we propose a method based on the inclusion of leads of the persistent shock variable. In particular, given that the DGP of system \eqref{eq:DGP} poses an AR(1) for $x_t$, one should regress:
\begin{equation} \label{eq:LPleads}
y_{t+h}= \delta_{h,0}x_{t} + \delta_{h,1}x_{t+1} + \xi_{t+h},
\end{equation}

\noindent where $\delta_{h,0}$ is the $h$-horizon response identified by LPs that include leads of the shock $x_t$, which we denote as $\mathcal{R}^{F} (h)$. In more general processes, in which the autocorrelation of the shock may be of an order larger than one, the optimal choice of leads can be derived adapting the procedure from \citet{choi2012model}.\footnote{See also \citet{lee2020lag} for lag order selection in LPs.} The most conservative procedure would be to include $h$ leads of the shock in each period $h$. This is the choice implemented in Section~\ref{sec:applications}, when considering  empirical applications.

\begin{prop} \label{prop:LPtoMA}
Given the data generating process described by system~\eqref{eq:DGP}, the response function identified by modified LPs to a shock $x_t$ as described in equation~\eqref{eq:LPleads} is equal to the response as if the shock had no persistence (and to the response obtained from DLMs  as in equation~\eqref{eq:MA}), that is:  \\
\indent $\mathcal{R}(h)^{F} =  \mathcal{R}(h)^{*}  = \mathcal{R}(h)^{DLM} $ $\forall$ $\gamma$ and $h$. 
\end{prop}
\begin{proof}
See Appendix~\ref{app:proof2}.
\end{proof}

Intuitively, leads of $x_t$ in equation~\eqref{eq:LPleads} act as controls for the persistence of the shock throughout the response horizon, so that the parameter $\delta_{h,0}$ reflects the dynamic response to a counterfactual serially-uncorrelated shock, that is, controlling for the effect due to $\frac{\partial x_{t+1}}{\partial x_t}\neq0$ built in system~\eqref{eq:DGP} when $\gamma \neq 0$.

\subsubsection{Adapting DLMs  to include the effect of persistence} \label{sec:MAtoLP}
As noted earlier, $\mathcal{R}(h)^{DLM} =  \mathcal{R}(h)^{*}$ regardless of the value of $\gamma$. However, in some instances, the researcher may be interested in the response that includes the effect of persistence ($\mathcal{R}(h)^{} $). In this subsection, we show how to adapt DLMs  to recover these responses. Intuitively, the idea is to compute the impulse responses in system~\eqref{eq:DGP} with respect to $\varepsilon_t$ instead of $x_t$.

Consider a recursive substitution of $x_t$ in system~\eqref{eq:DGP}:
\begin{equation}\label{eq:recursion}
	y_t = \delta \gamma^t x_0 + \delta \sum_{i=0}^t \gamma^i {\varepsilon}_{t-i} + {u}_t.
\end{equation}

The responses of $y_t$ to $\varepsilon_t$, which we denote by $\mathcal{R}(h)^{DLM-per}$, can be obtained from the coefficients $\tilde{\theta}_h$ in:
\begin{equation}
\label{eq:MA_eps}
y_t = \tilde{\theta}_0 \varepsilon_t + \tilde{\theta}_1 \varepsilon_{t-1} + \tilde{\theta}_2 \varepsilon_{t-2} + \tilde{\theta}_3 \varepsilon_{t-3} + \tilde{\theta}_4 \varepsilon_{t-4} +  \ldots + e_t. 
\end{equation}

\begin{prop} \label{prop:MAtoLP}
Given the data generating process described by system~\eqref{eq:DGP}, the response function identified by DLMs  of $y_t$ to the innovation $\varepsilon_t$ as described in equation~\eqref{eq:MA_eps} is equivalent to the response that includes the effects of persistence (and to the response obtained from LPs as in equation~\eqref{eq:LP}):  \\
\indent $  \mathcal{R}(h)^{DLM-per} = \mathcal{R}(h)^{} = \mathcal{R}(h)^{LP} $ $\forall$ $\gamma$ and $h$. 
\end{prop}
\begin{proof}
See Appendix~\ref{app:proof3}.
\end{proof}

%


Proposition \ref{prop:MAtoLP} establishes a direct equivalence between the coefficients obtained from equation~\eqref{eq:MA_eps} and those obtained from LPs in equation~\eqref{eq:LP}: $\tilde{\theta}_h = \delta_h$ $\forall$ $h$. The former are also related to the coefficients estimated from the DLM in terms of $x_t$, as in equation~\eqref{eq:MA}: $\theta_0= \tilde{\theta}_0 = \delta$, $\theta_1= \tilde{\theta}_1 -{\gamma}  \tilde{\theta}_0, \ldots, \theta_h= \tilde{\theta}_h -{\gamma}  \tilde{\theta}_{h-1}$. Intuitively, the response of $y_{t+1}$ to $x_t$ has an overall effect of $\delta_1=\tilde{\theta}_1$, which includes (i) the direct effect of $x_t$ on $y_{t+1}$ (0, in our simple case) and (ii) the effect on $y_{t+1}$ that is due to the persistence in $x_t$ (given by ${\gamma} \delta$). The \textit{standard} DLM estimation from equation~\eqref{eq:MA}, since it accounts for the evolution of $x_t$ over the response horizon,  is implicitly subtracting the part of the response that is  given by the persistence of $x_t$ from the overall effect. 

\subsection{Examples} \label{sec:simul}

In this subsection, we perform stochastic simulations of the asymptotic  behavior of the impulse response functions using both LPs and DLMs. Our goal is twofold. First, to evaluate quantitatively the conclusions reached in the previous subsection using a plausible calibration of the parameters that determine the model. Second, to consider a slightly more complex (and realistic) version of the data generating process that includes richer features frequently present in real empirical applications. In particular, we consider the following process:
\begin{eqnarray} \label{eq:DGPsimul}
y_t &=& \rho y_{t-1} + B_0 x_t + B_1 x_{t-1} + u_t \nonumber \\
x_t &=& \gamma x_{t-1} + \varepsilon_t,
\end{eqnarray}

\noindent where $\E(\varepsilon_{t-s} u_{t-r})=0$ $\forall s,r \gtreqless 0$, and $u_t$ and $\varepsilon_t$ follow $\mathcal{N}(0,1)$ distributions. We set $B_0 = 1.5$, $B_1 = 1$ and  $\rho = 0.9$.

Compared to system~\eqref{eq:DGP}, the new DGP described in system~\eqref{eq:DGPsimul} includes persistence in the outcome variable through $\rho$, and allows the shock $x_t$ to have lagged effects on $y_t$ through $B_1$.\footnote{We introduce this extra lag of the shock to make explicit the distinction between the effect due to the persistence of the shock and the effect of lagged values of the shock on current outcomes.}

We simulate system~\eqref{eq:DGPsimul} for 100 million periods and recover the dynamic responses of $y_t$ to the shock $x_t$ using LPs:
\begin{eqnarray} \label{eq:LPleadsSIMUL}
y_{t+h}=\rho y_{t-1} + \beta_{h,0}x_{t}+ \beta_{h,1}x_{t-1} + \beta_{h,f}x_{t+1} + \xi_{t+h}.
\end{eqnarray}

We consider three cases: (i) no persistence ($\gamma= 0$), without including leads in the estimation (i.e., setting $\beta_{h,f}=0$); (ii) some persistence ($\gamma = 0.2$) and still $\beta_{h,f}=0$; (iii) some persistence ($\gamma = 0.2$) and including a lead of the explanatory  variable (i.e., allowing $\beta_{h,f}\neq0$).\footnote{The choice of $\gamma=0.2$ is based on an empirical application that we will present in Section \ref{sec:applications}. Of course, larger values of $\rho$ would yield higher biases due to the persistence of the process.} 

Note that equation~\eqref{eq:LPleadsSIMUL} must include a lag of shock $x_{t}$ to capture the effect of $B_1$ in system~\eqref{eq:DGPsimul}. However, this does \emph{not} control for the potential persistence of shock $x_{t}$, as will be apparent in the simulations.


Figure~\ref{fig:simul} shows the results of our simulations. In case (i) (dark-blue solid line), the response has a contemporaneous effect of  $\hat{\beta}_{1,0}=1.5$ and peaks at the following period due to the the fact that both $\rho$ and $B_1$ have positive values. Using the language of the previous section, the impulse response function estimated by LPs with no persistence is asymptotically equivalent to the one obtained directly  from equation~\eqref{eq:LPleadsSIMUL}, that is, $\hat{\mathcal{R}}(h)^{LP}  \rightarrow  \mathcal{R}(h)^{*}$. 

\begin{figure} 
\caption{Simulated responses using LPs}\label{fig:simul}
\begin{center}
		{\includegraphics[width=10cm]{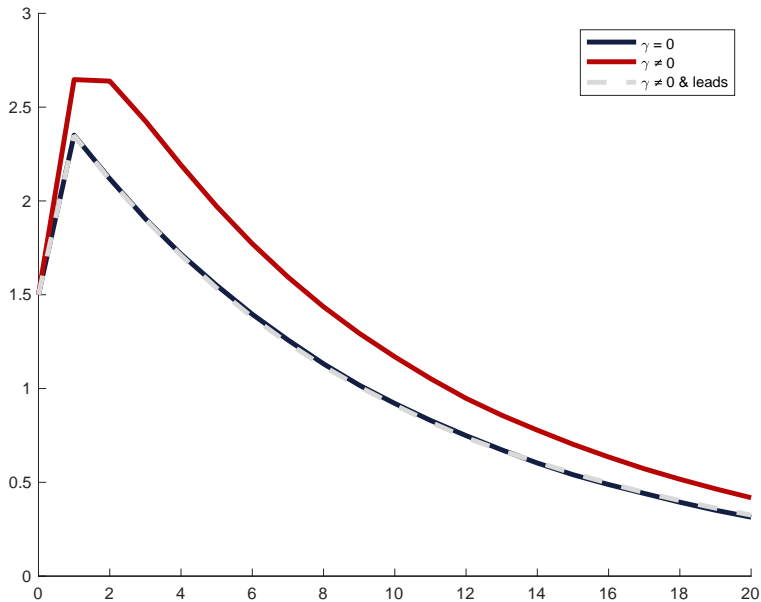}}	
\end{center}
\footnotesize This figure shows the response of a simulated outcome variable to a shock with different degrees of persistence, using LPs. The dark blue line shows the results of estimating equation~\eqref{eq:LPleadsSIMUL} assuming $\gamma = 0$ in equation~\eqref{eq:DGPsimul}. The red line shows the same estimation when $\gamma = 0.2$. The dashed grey line shows the response after including leads of the shock as in equation~\eqref{eq:LPleadsSIMUL} and still assuming $\gamma= 0.2$.
\end{figure}

In case (ii) (red solid line), the introduction of persistence in the shock $x_t$ results in a larger effect on $y_t$ on all horizons after impact. This has potentially important implications:  if a macroeconomist is interested in the effects of a serially-uncorrelated shock (as in most general equilibrium models), but naively estimates equation~\eqref{eq:LPleadsSIMUL}, implicitly setting $\beta_{h,f}=0$, then the dynamic response is upwardly biased due to the persistence of the shock, i.e., $\hat{\mathcal{R}}(h)^{LP}  >  \mathcal{R}(h)^{*}$ for $h >0$. Given the assumptions on the autocorrelation of the process $x_t$, the bias is particularly large in the short and medium run. Higher values of the persistence parameters $\gamma$ and $\rho$ would increase the difference between both responses (blue and red lines in Figure~\ref{fig:simul}).

In case (iii) (dashed grey line in Figure~\ref{fig:simul}), we see that the inclusion of leads of $x_t$ renders the response of the outcome variable to a persistent shock identical to the one obtained when considering a shock without persistence, i.e., $\hat{\mathcal{R}}(h)^{F}  \rightarrow  \mathcal{R}(h)^{*}$. In Appendix~\ref{app:simulreal} we provide an alternative simulation where the shock $x_t$ in \eqref{eq:DGPsimul} is not assumed to follow an AR(1) process but it is instead taken from  actual data.

Next, we use these simulations to show that the computation of impulse responses using DLMs  always yields the same estimates regardless of the persistence in $x_t$, that is, $\hat{\mathcal{R}}^{DLM} (h) \rightarrow  \mathcal{R}^{*} (h)$ for any value of $\gamma$. 

First, note that, since $\rho<1$, system~\eqref{eq:DGPsimul} can be inverted and re-written as:

\begin{equation}
\label{eq:MAinvert}
y_t = \left(1-\rho L\right)^{-1} \left( B_0 + B_1 L \right)x_t + \left(1-\rho L\right)^{-1} u_t,
\end{equation}

\noindent where  $L$ represents the lag operator. 

Given the independence of $u_t$ and $x_t$, the representation from equation~\eqref{eq:MAinvert} suggests that the dynamic responses of $y_t$ from $x_t$ can be obtained from the coefficients $\vartheta_h$ in the following regression:

\begin{equation}
\label{eq:MAsimul}
y_t = \vartheta_0 x_t + \vartheta_1 x_{t-1} + \vartheta_2 x_{t-2} + \vartheta_3 x_{t-3} + \ldots + \vartheta_H x_{t-H} + \xi_t,
\end{equation}

\noindent where $H$ is the response horizon.\footnote{\citet{baek2019abcs} show that for autoregressive distributed lag models, setting the lag order to H is a necessary condition to achieve consistency. }

We estimate equation~\eqref{eq:MAsimul} fo three different cases: (i) assuming that $\gamma=0$ in the data generating process described in system~\eqref{eq:DGPsimul}, (ii) assuming that $\gamma=0.2$ and (iii) replacing $x_t$ with $\hat{\varepsilon_t}$ in  equation~\eqref{eq:MAsimul} (i.e., following equation~\eqref{eq:MA_eps}). 

\begin{figure}
\caption{Simulated responses using DLMs }\label{fig:simul_genregMA}
\begin{center}
		{\includegraphics[width=10cm]{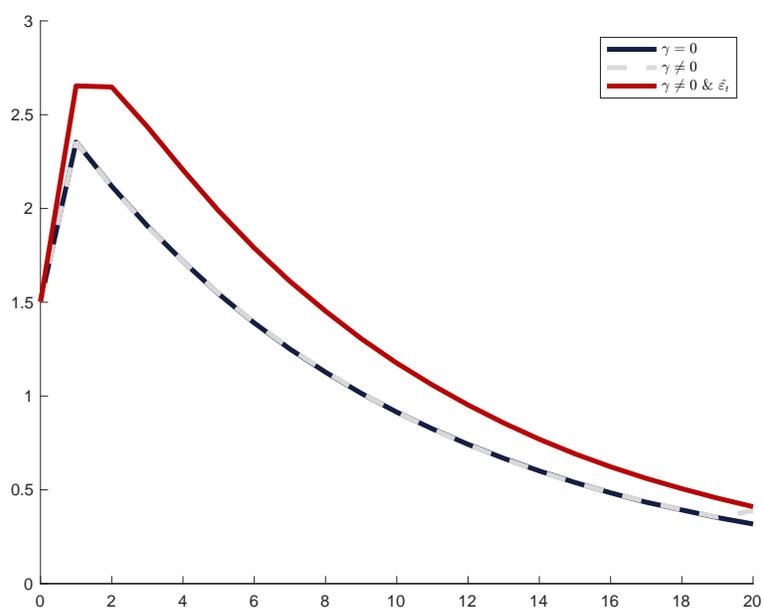}}	
\end{center}
\footnotesize This figure shows the response of a simulated outcome variable to a shock with different degrees of persistence, using DLMs . The dark blue line shows the results of estimating equation~\eqref{eq:MAsimul} assuming $\gamma = 0$ in system~\eqref{eq:DGPsimul}. The dashed grey line shows the same estimation when $\gamma = 0.2$. The red line shows the response when substituting $x_t$ in equation~\eqref{eq:MAsimul} by $\hat{\varepsilon}_t$, an OLS estimate of $\varepsilon_t$ (see equation \eqref{eq:epshat}), where serial correlation has been removed.
\end{figure}

The results are shown in Figure~\ref{fig:simul_genregMA}. Cases (i) and (ii) are displayed in blue and dashed grey lines, respectively. As argued earlier, since equation~\eqref{eq:MAsimul} controls for all potential dynamic effects of $x_t$, including its persistence, the coefficients $\vartheta_h$ reflect the responses to a shock as if the variable $x_t$ showed no persistence, regardless of the value of $\gamma$. Hence, we have that  $\hat{\mathcal{R}}(h)^{DLM} \rightarrow  \mathcal{R}(h)^{*}$ for any  $\gamma$. Note that these impulse response functions are the same as those obtained with LPs ($\hat{\mathcal{R}}(h)^{LP}$) when $\gamma = 0$, or when we include leads in the LPs ($\hat{\mathcal{R}}(h)^{F}$).

Case (iii) is shown in the red line in Figure~\ref{fig:simul_genregMA}. As argued in the previous subsection, when computing the impulse response with respect to $\varepsilon_t$, we are allowing the DLMs  to pick up the effect that is due to the persistence in $x_t$. In other words, since we do not implicitly control for the leads of $x_t$ but for those of $\varepsilon_t$ in the DLM, we are not taking into account the persistence of $x_t$. In this case, the responses are equal to those obtained from LPs when $\gamma \neq 0$: $\hat{\mathcal{R}}(h)^{DLM-per} = \hat{\mathcal{R}}(h)^{LP} \rightarrow \mathcal{R}(h)$.


\section{Discussion: A guide to practitioners} \label{sec:discussion}

In the presence of a persistent shock, a researcher needs to determine what object to identify. Table~\ref{tab:options} summarizes the adjustments required in LPs and DLMs  depending on the choice of the object of interest.

\begin{table}
\renewcommand{\arraystretch}{1.5} 
\caption{Adapting LPs and DLMs  when shocks are persistent} 
\label{tab:options}
\begin{center}
\begin{tabular}{l | c c}
Object of interest / Method                        & LPs                   & DLMs                   \\
\hline
Response as if no persistence ($ \mathcal{R}(h)^{*} $) & include leads     & no action needed \\
Response with persistence ($\mathcal{R}(h)$)  & no action needed & replace $x_t$ with ${\varepsilon}_{t}$
\end{tabular}
\end{center}
\end{table}

The researcher faces two options: to identify the response \textit{as if} the shock were uncorrelated ($\mathcal{R}(h)^{*}$) or the response that includes the effect of persistence ($\mathcal{R}(h)$). There are arguments in favor of both. 
Ultimately, deciding for one or the other may depend on what specific question the researcher is trying to address.


Since $\mathcal{R}(h)^{*}$ can be understood as the IRF resulting from a \textit{standardized} shock (so that it becomes serially uncorrelated), it should be the desired object when the researcher wants to establish comparisons across dynamic responses. There are at least three instances when $\mathcal{R}(h)^{*}$ can facilitate comparisons. First, a shock identified from within a model (say, a structural VAR) or the innovation to a stochastic process in a DSGE model are, by construction, a m.d.s. (they are non-persistent). Given the absence of serial correlation, the thought experiment carried out in such cases  is equivalent to constructing and IRF such as the shock takes the value of 1 on impact and 0 afterwards. Contrary to VAR-identified shocks or innovations in a DSGE model, narratively-identified shocks may display serial correlation. If this is the case, $\mathcal{R}(h)$ (resulting, for example, from standard LP) will identify a different object, since the effect of serial correlation is included in the IRFs. In this instance, $\mathcal{R}(h)^{*}$ will provide the same macroeconomic experiment as, for example, a DSGE model.\footnote{
As mentioned earlier, when $x_t$ is the shock of interest, $\mathcal{R}(h)^{*}$ is defined as the ``traditional impulse response''  in \citet{koop1996impulse}.}

Second, $\mathcal{R}(h)^{*}$ can also be an object of interest when the researcher wants to compare the effects of different shocks, e.g., whether fiscal or monetary policy is more effective in stimulating output. For example, it may be the case that fiscal  shocks tend to show more persistence or that a given identification procedure tends to generate shocks with different degree of serial correlation. If the effect of persistence amounts to a non-negligible amount of the dynamic response, this could wrongly lead to the conclusion that one shock is more effective than the other when the true underlying cause is that the DGP of both shocks is different. Since $\mathcal{R}(h)^{*}$ effectively standardizes the dynamic responses to shocks with different data generating processes, this would facilitate such comparison. 

Third, in a similar vein,  $\mathcal{R}(h)^{*}$ can be useful  when the researcher wants to compare the effects of the same shock using data from different countries. This is because $\mathcal{R}(h)^{*}$ provides a standardization of the data generating processes of the shocks, which may be heterogeneous across countries.\footnote{Consider the following example: we want to compare the effects of fiscal policy in the US (using a news variable) and in another country (where we have availability of an alternative news variables). Consider the case that the news variables have different amounts of serial correlation and we obtain estimates of the government spending multipliers in both countries. Could we conclude that government spending is more effective in one country versus the other? Potentially, both policies could be equally effective but their sources of identification (news variables) may have different DGPs (one with more serial correlation than other), what leads to different multipliers.}

On the other hand, $\mathcal{R}(h)$ should be the object of interest when the researcher is interested in estimating the \textit{most likely dynamic response} of a variable to a shock according to the historical data. This argument is similar to the one posed by \citet{fisher2010stock} and  \citet{ramey2018government} to support the use of the cumulative multiplier (the ratio of the integral of the output response to that of the government spending response) to evaluate the effectiveness of fiscal policies. If we consider the effects of a monetary policy shock that cuts the policy rate by one percentage point, it is important to note that, if that shock displays persistence, then the total monetary policy action (the evolution of the nominal interest following the initial tightening) may be different to what would occur if the shock were non-persistent.


Importantly, and regardless of the experiment that one wants to run, looking at the difference between $\mathcal{R}(h)$ and $\mathcal{R}(h)^{*}$ is informative by itself, as it speaks about how much of the dynamic response is due to the implied DGP of the shock variable. Put differently, it informs the researcher of a propagation mechanism: $\mathcal{R}(h)$ includes the propagation through the persistence of $x_t$ while $\mathcal{R}(h)^{*}$ does not. 

Further to this, the methods that underlie the construction of $\mathcal{R}(h)^{*}$ when using local projections can be exported to more general uses. Hence the inclusion of leads of different variables can help in decomposing an IRF in different channels of propagation (where serial correlation is just one of them). This avenue could be particularly informative in highlighting what economics models can bring the dynamic responses closer to the data.

\section{Application} \label{sec:applications}

In this section we use the empirical work of \citet{ramey2018government} to show the quantitative relevance of serial correlation in an actual example. We do so by computing two types of IRFs, $\mathcal{R}(h)$ and $\mathcal{R}(h)^{*}$, as described above.\footnote{In the appendix, we consider additional applications, based on \citet{guajardo2014expansionary}, \citet{romer2004measure}, \citet{gertler2015monetary}, and \citet{romer2010macroeconomic}.} 



\citet{ramey2018government}, building on previous work by \citet{ramey2011identifying} and \citet{owyang2013government}, produce a series of announces about future defense spending between 1890q1-2014q1, scaled by previous quarter trend real GDP.\footnote{\citet{ramey2018government} estimate trend GDP as sixth degree polynomial for the logarithm of GDP and multiplier by the GDP deflator. In fact,  it is the use of the GDP deflator and trend GDP as a way to scale the shocks what seems to induce the persistence. The persistence is also present when the shock is scaled by previous-quarter GDP, as in  \citet{owyang2013government}.} This series, plotted in panel D of Figure~\ref{fig:series}, has a positive autocorrelation of $18.4\%$ (47.0\% in the subsample after WWII).\footnote{This positive autocorrelation is significant at a confidence level of 90\% when considering standard errors that are robust to the presence of heteroskedasticity and persistence (with more than one lag) for the whole sample. For the subsample starting after WWII, the autocorrelation is significant at any level. } 


\citet{ramey2018government} use LPs to estimate the response of output and government spending to a shock in future defense spending.
We follow their same approach and sample and estimate the following equations for output ($y_t$) and government spending ($g_t$):
\begin{eqnarray} \label{eq:RameyLP}
y_{t,h} &=& \beta^y_h shock_t  + \sum_{j=1}^P\rho_{j,h}^z z_{t-j}  + \sum_{f=1}^h\gamma_{f,h} shock_{t+f} + \xi_t \nonumber \\
g_{t,h} &=& \beta^g_h shock_t  +  \sum_{j=1}^P\rho_{j,h}^z z_{t-j}  + \sum_{f=1}^h\gamma_{f,h} shock_{t+f}  + \varepsilon_t,
\end{eqnarray}

\noindent where $z_{t}$ includes $P$ lags of  $y_t$, $g_t$ and $shock_t$. Note that, following the discussion in previous sections, we include $h$ leads of the variable $shock_t$. 
In particular, for each horizon $h$ we include $h$ leads.

To replicate \citet{ramey2018government}'s estimates, we set $\gamma_{f,h}=0$, $\forall$ $f,h$. The  black, solid line in Figure~\ref{fig:Ramey_responses} represents the estimated responses of output (left panel) and government spending (right panel) to the shock.\footnote{Figure~\ref{fig:Ramey_responses_doubleCI} also replicates the original 95\% confidence intervals computed using the Newey-West correction.} As noted in Section~\ref{sec:econometric}, these dynamic responses are the equivalent to the $\mathcal{R}(h)$ as defined in equation \eqref{eq:IRF_per} (with the only difference being that $\Omega_{t-1}$ includes now the past history of $z_t$). The results closely resemble those in  \citet{ramey2018government} (Figure 5 of their paper).\footnote{We drop the last $h$ observations of the sample, so that the specifications with and without leads can be fully comparable. This does not have any discernible effect when replicating the original results from \citet{ramey2018government}.} 

\begin{figure}
\caption{Output and government spending responses, with and without leads} \label{fig:Ramey_responses}
\begin{center}
		{\includegraphics[width=10cm]{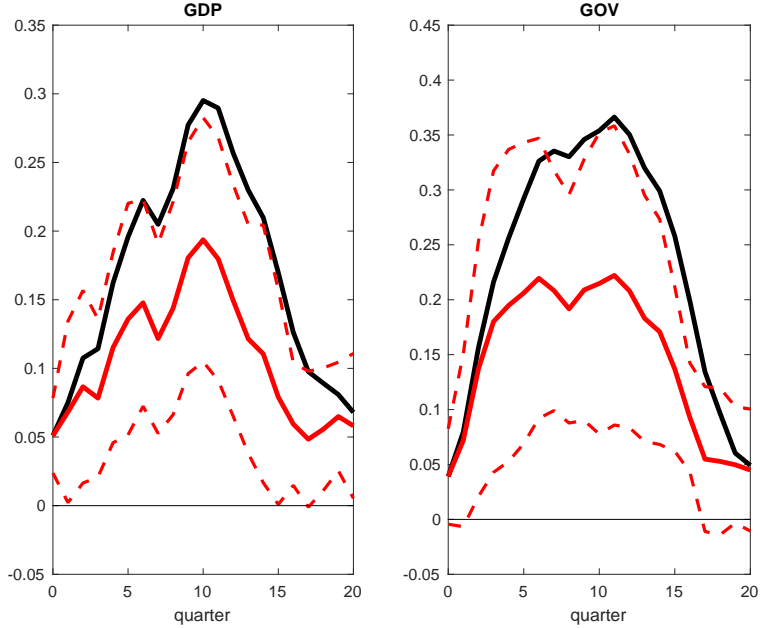}}	
\end{center}
\footnotesize Black lines show the results of estimating the system~\eqref{eq:RameyLP} without including any lead (as in \citet{ramey2018government}). 
Red solid lines represent the results of estimations when including $h$ leads of the \citet{ramey2018government} news variable (with 95\% confidence intervals).
\end{figure}

Next, we allow $\gamma_{f,h}\neq0$. As discussed in Section~\ref{sec:econometric}, this amounts to estimating $\mathcal{R}(h)^*$ as defined in equation~\eqref{eq:IRF_iid}. In the red lines in Figure~\ref{fig:Ramey_responses}, we observe that the dynamic responses change considerably when the leads are included. For example, after two years, output and government spending are 40\% lower than in \citet{ramey2018government}'s estimates. The large observed difference between $\mathcal{R}(h)$ and $\mathcal{R}(h)^*$ suggests that the persistence of the news variable plays a non-negligible role in explaining the dynamic transmission of the fiscal shock to output and government spending.

Whether to include leads or not also has implications for inference. The 95\% confidence intervals when leads are included (shown in dashed lines in Figure~\ref{fig:Ramey_responses}) are substantially narrower than when they are not (grey areas in Figure~\ref{fig:Ramey_responses_doubleCI}). The latter are around 50\% broader after two years, and more than twice as big after three years.

The dynamic responses of output and government spending are informative about the expected path of these variables after a shock. To obtain a measure of the efficiency of fiscal policy (i.e., the increase of output per each dollar increase in government spending), \citet{ramey2018government}  use the cumulative multiplier, computed as:\footnote{\citet{ramey2018government} show that the cumulative multiplier can be obtained in one step yielding identical results to those obtained combining equations \eqref{eq:RameyLP} and \eqref{eq:multiplier}.} 

\begin{equation} \label{eq:multiplier}
	M_{t,h} = \frac{\sum_{i=1}^h \beta^y_h}{\sum_{i=1}^h \beta^g_h}.
\end{equation}

We find that this statistic is not substantially affected by persistence of the shock (Figure~\ref{fig:Ramey_multiplier}). Given that both output and government spending react similarly when including leads of the shock, taking the ratio of the two variables attenuates the differences between both specifications.\footnote{Even though the multiplier does not change much when accounting for persistence, the fact that the expected responses of output and government spending do change substantially is very relevant from a policy-maker point of view. For example, a higher response of government spending can affect other important variables such as public debt or future changes in tax liabilities.}


\paragraph{Non-linear effects.} We now investigate whether the effect of persistence in the shock can affect the responses in a non-linear setting, i.e., if government spending multipliers are different in expansions and recessions.\footnote{See \citet{ramey2018tenyears} for a recent summary of this debate. For example, an influential study by \citet{auerbach2012measuring} finds that government spending multipliers are higher during recessions using a non-linear VAR. \citet{alloza2017fiscal}  highlights the role of the information used to define a period of recession, and finds that output responds negatively to government spending shocks in a post-WWII sample under different identification and estimation approaches.} For this, we follow \citet{ramey2018government} and estimate a series of non-linear LPs:

\begin{align}  \label{eq:LPnonli}
				x_{t+h} =  S_{t-1} \left[ \alpha_{A,h} + \sum_{j=1}^P\rho_{A,j,h} z_{t-j} + \beta_{A,h} shock_t + \sum_{f=1}^h\delta_{A,f,h} shock_{t+f} \right] &+&  \notag \\ 
					  (1-S_{t-1}) \left[ \alpha_{B,h} + \sum_{j=1}^P\rho_{A,j,h} z_{t-j} + \beta_{B,h} shock_t +\sum_{f=1}^h\delta_{B,f,h} shock_{t+f} \right] &+&  \xi_{t+h},
\end{align}

\noindent where $x_t$ is either output or government spending and $S_t$ is a binary variable indicating the state of the economy. When $S_t=1$, the economy is booming and, when $S_t=0$, the economy is in recession, which is defined as when the unemployment rate is above the threshold of 6.5. In this setting, all the variables (and the constant), are allowed to have differential effects during expansions and recessions. 

We first replicate the non-linear responses of output and government spending during booms and recessions obtained by \citet{ramey2018government}. Hence, we estimate equation~\eqref{eq:LPnonli} setting $\delta_{A,f,h}=\delta_{B,f,h}=0$ $\forall f,h$, which identifies $\mathcal{R}(h)$. Our results, shown in Figure~\ref{fig:G_nonlin_resp} in black lines, resemble very closely those from the authors. Next, we repeat the experiment accounting for potential persistence, that is, including leads of the shock (identifying $\mathcal{R}(h)^*$). The results are shown in red lines in Figure~\ref{fig:G_nonlin_resp}.  While relatively similar in the case of expansions, the responses are quantitatively different during recessions.  The estimates that include leads lie outside of the 95\% confidence bands during much of the response horizon. The results suggest that ignoring the effect of persistence could yield responses during recessions that, after 2--3 years,  are twice as large as the responses that account for the effect of persistence. Or, in other words, persistence in the shock is responsible for up to 50\% of the dynamic transmission of the shock during recessions.

\begin{figure}  \caption{Responses during expansions and recessions, with and without leads}  \label{fig:G_nonlin_resp}
\begin{center}
		{\includegraphics[width=10cm]{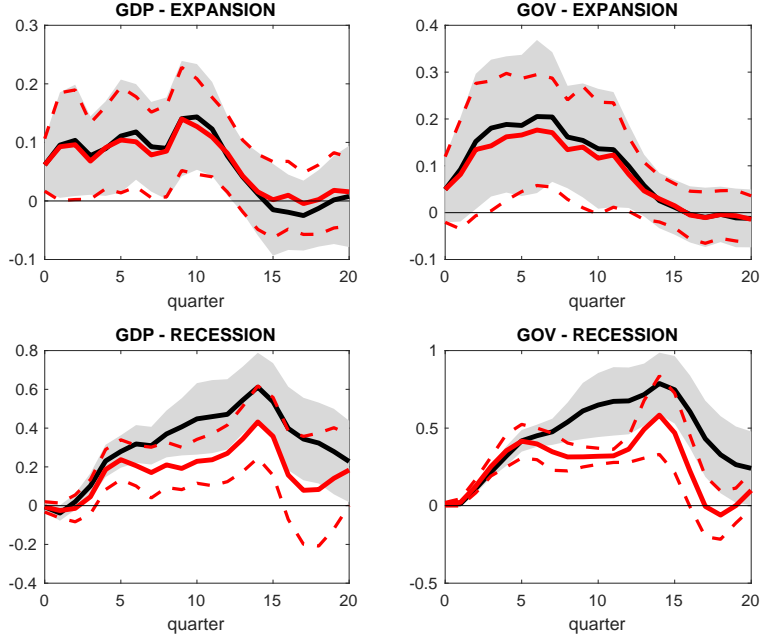}}	
\end{center}
\footnotesize Black lines show the results from system of equations~\eqref{eq:LPnonli} without including any lead (as in \citet{ramey2018government}). Grey areas represent 68 and 95\% Newey-West confidence intervals for these estimates. Red solid lines represent the results of estimations when including $h$ leads of the Ramey's news variable. Red dashed lines represent the 95\% Newey-West confidence intervals for these estimates.
\end{figure}

In Figure~\ref{fig:G_nonli_mults}, we show how these responses map into estimates of non-linear fiscal multipliers. In the case of expansions, the results do not change much depending on whether the persistence is accounted for (red solid line, $\mathcal{R}(h)^*$) or not (black solid line, $\mathcal{R}(h)$). In either case, they resemble those in \citet{ramey2018government} (see Figure 6 of their paper). In recessions, however, the results change substantially depending on whether the persistence is controlled for or not. If it is not (black solid line), the multiplier has a negative value upon impact and substantially falls in the following quarter to a value of -2. It becomes positive before the end of the first year, and fully converges to the value of the multiplier during expansions after six quarters. 
If the persistence is excluded from the dynamic responses (red dashed line), the cumulative multiplier is -1 (instead of -2) and becomes positive after the first year. Furthermore, the multiplier during recessions remains lower than the multiplier during expansions for a much longer period. When the persistence is not accounted for, this convergence is achieved after 6 quarters, as mentioned above. However, when including leads of the shock, this convergence is not fully reached during our considered response horizon. 
These results suggest that during the short and medium-run the government spending multiplier could be lower during recessions than during expansions, and part of  this difference may be attributable to the presence of persistence in the shock.

\begin{figure} \caption{Government spending multiplier during expansions and recessions, with and without leads}  \label{fig:G_nonli_mults}
\begin{center}
		{\includegraphics[width=10cm]{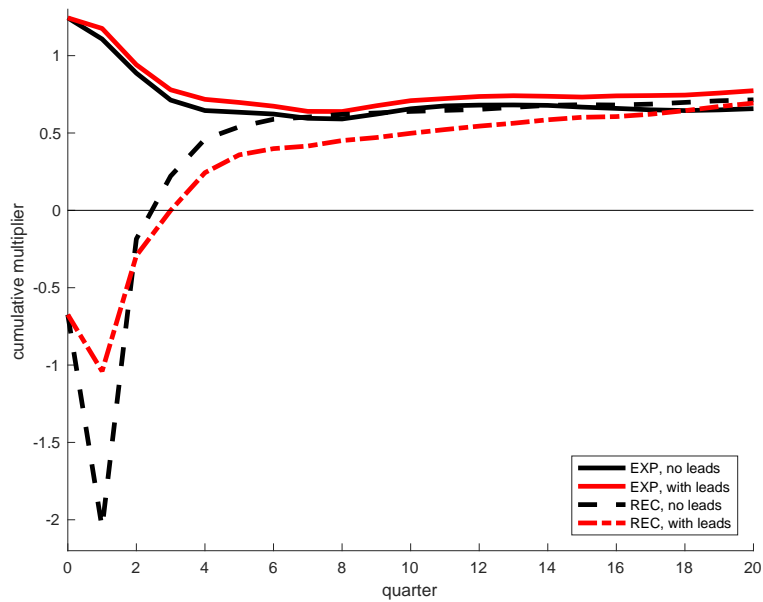}}	
\end{center}
\footnotesize The black solid and dashed lines show the cumulative multiplier during periods of expansion and recession, respectively,  without including any lead (as in \citet{ramey2018government}). The red solid and dashed lines show the cumulative multiplier during periods of expansion and recession, respectively, when including leads of the shock.
\end{figure}

One of the main advantages of LPs is that they allow to accommodate non-linear settings,  as those in equation \eqref{eq:LPnonli}. This is particularly useful since, contrary to threshold VARs, LPs do not impose any restriction on the evolution of state $S_t$ (while non-linear VARs  that interact the shock with a state dummy do assume that $S_t$ remains fixed during the response horizon). The framework explained in the previous section allows to consider additional macroeconomic experiments that can help understand how restrictive this condition is. In particular, by including leads of the state $S_t$ in equation \eqref{eq:LPnonli} we are identifying the counterfactual response to a fiscal shock when the underlying state of the economy is not allowed to change (as in threshold VARs). We perform this experiment and report the multipliers during booms in recessions in green lines in Figure~\ref{fig:G_nonli_mults_withSTATE}. We observe that, when the state is not allowed to change, the multiplier during recessions is slightly higher in the short run, but essentially unchanged at medium and longer horizons. This exercise allows us to illustrate how the use of leads of variables in conjunction with LPs can help understand interesting counterfactual exercises and shed light on the dynamic transmission of shocks.
\section{Conclusions} \label{sec:conclusions}

We have shown that persistence results in the estimation of different responses when using LPs \emph{versus} traditional methods based on DLMs . For a researcher interested in the response  \textit{as if} the shock were not persistent, DLMs  yield the desired object, but LPs need to be adapted. The opposite is true if the object of interest is the response to the shock ``as it is''. Regardless of which is the thought experiment that the researcher seeks to carry out, the difference between both types of responses is informative about how much of the dynamic transmission of a shock is due to the presence of persistence.



The use of leads  can be generalized to other interesting contexts, as it allows to shut down channels of transmission. For example, one may be interested in the effects of monetary policy shocks on output due to a particular instrument while holding other variables (e.g., changes to fiscal policy) constant. In the context of LPs, leads of a selected variable (e.g., tax changes) will deliver responses holding that variable constant. This methodology allows to separate the direct  (due to the impact through the regressor of interest) and indirect effects (due to other variables in the regression). This has often been used in the context of VARs, by imposing restrictions on the coefficients of selected impulse responses. The inclusion of leads achieves a similar goal in LPs, hence allowing to construct interesting macroeconomic experiments. We leave these questions for future research.


%
%

\bibliographystyle{apalikeMario} 
\bibliography{bibfile_persistent_z}

\newpage


\clearpage

\newpage
\clearpage
\setcounter{page}{1}
\begin{appendices}
\section*{Online Appendices}
\setcounter{equation}{0}
\renewcommand{\theequation}{A.\arabic{equation}}

\section{Proofs}\label{app:proofs}
\setcounter{figure}{0}
\setcounter{page}{1}
\setcounter{equation}{0}
\setcounter{footnote}{0}
\renewcommand{\thefigure}{A\arabic{figure}}
\renewcommand{\thesubsection}{A.\arabic{subsection}}
\renewcommand{\theequation}{A.\arabic{equation}}

\subsection{Proof of Proposition~\ref{prop:equivalence}} \label{app:proof1}
Consider equation~\eqref{eq:LP} (rewritten here for convenience):
\begin{equation} \label{eq:LPapp}
	y_{t+h} = \delta_{h} x_{t} + \xi_{t+h},
\end{equation}

where $\delta_{h} = \mathcal{R}(h)^{LP}  $ represents the impact of variable $x_t$ on $y_{t+h}$ (the response function). Since $\delta_{h}$ is the linear projection coefficient of equation~\eqref{eq:LPapp}: 

\begin{equation} \label{eq:LP_coef}
	\delta_{h}  = \frac{cov(y_{t+h},x_t)}{var(x_t)}.
\end{equation}

The dynamic effect of $x_t$ on $y_{t+h}$ can also be obtained from DLMs  as in equation~\eqref{eq:MA}:
\begin{equation*}
y_t = \theta_0 x_t + \theta_1 x_{t-1} + \theta_2 x_{t-2} + \theta_3 x_{t-3} + \theta_4 x_{t-4} \ldots + u_t.
\end{equation*}

Since this expression holds $\forall t$, it can be written as:
\begin{equation*}
y_{t+h} = \theta_0 x_{t+h} + \theta_1 x_{t+h-1} + \theta_2 x_{t+h-2} + \theta_3 x_{t+h-3} + \ldots + \theta_h x_{t}  + u_t, 
\end{equation*}
where the coefficient $\theta_h=\mathcal{R}(h)^{DLM}$ represents the impulse response in period $h$, obtained from:
\begin{equation} \label{eq_MA_coef}
	\theta_{h}  = \frac{cov(y_{t+h},x_t)}{var(x_t)} \Bigr|_{\substack{x_{t+1},...,x_{t+h}}}.
\end{equation}

When $\gamma = 0$ in the process described by system~\eqref{eq:DGP}, we have that $x_t=\varepsilon_t \sim \text{\emph{white noise}}(\mu_\varepsilon,\,\sigma_\varepsilon^{2})$, so:
\begin{equation*}
	\theta_{h}  = \frac{cov(y_{t+h},x_t)}{var(x_t)} \Bigr|_{\substack{x_{t+1},...,x_{t+h}}} = \frac{cov(y_{t+h},x_t)}{var(x_t)}. 
\end{equation*}

In this case, $\delta_{h} = \theta_h$ and LPs and DLMs yield the same responses: $\mathcal{R}(h)^{LP} = \mathcal{R}(h)^{DLM}$ $\forall h$. Note that $\mathcal{R}(h)^{DLM}=\mathcal{R}(h)^{*}$ $\forall h,\gamma$ since (under linearity):
\begin{equation*}
\mathcal{R}(h)^* =  \frac{\partial y_{t+h}} {\partial x_{t}} \Bigr|_{\substack{x_{t+1},...,x_{t+h}}}  = \frac{cov(y_{t+h},x_t)}{var(x_t)} \Bigr|_{\substack{x_{t+1},...,x_{t+h}}}. 
\end{equation*}

When $\gamma \neq 0$, equation~\eqref{eq:LP_coef} becomes (using the equations in system~\eqref{eq:DGP}):
\begin{equation} \label{eq:LP_ols_final}
	\delta_{h} = \frac{cov(y_{t+h},x_t)}{var(x_t)} = \frac{cov(\delta x_{t+h}+u_{t+h},x_t)}{var(x_t)} =  \delta\frac{cov( x_{t+h},x_t)}{var(x_t)}=\delta \gamma^h,
\end{equation}
using the expression:
\begin{equation}\label{eq:x_back}
x_{t+h}= \gamma^h x_t + \sum_{j=0}^{h-1}\gamma^j \varepsilon_{t+h-j}.
\end{equation}
However, the dynamic response obtained from DLMs  is:
\begin{equation*}
	\theta_{h}  = \frac{cov(y_{t+h},x_t)}{var(x_t)} \Bigr|_{\substack{x_{t+1},...,x_{t+h}}} = \delta\frac{cov(x_{t+h},x_t)}{var(x_t)} \Bigr|_{\substack{x_{t+1},...,x_{t+h}}}. 
\end{equation*}

When $h=0$, the above expression becomes $\theta_{0}  = \delta\frac{cov(x_{t},x_t)}{var(x_t)} =\delta$. For $h>0$, we have that $\theta_{h}  = \delta\frac{cov(x_{t+h},x_t)}{var(x_t)} \Bigr|_{\substack{x_{t+1},...,x_{t+h}}} =0$. This shows that when $\gamma \neq 0$, we have that $\mathcal{R}(h)^{LP} = \mathcal{R}(h)^{DLM}$ if and only if $h=0$. 
\quad $\blacksquare$

\subsection{Proof of Proposition~\ref{prop:LPtoMA}}  \label{app:proof2} 

Consider equation~\eqref{eq:LPleads} (rewritten here for convenience):
\begin{equation} \label{eq:LPleads_app}
y_{t+h}= \delta_{h,0}x_{t} + \delta_{h,1}x_{t+1} + \xi_{t+h},
\end{equation}

where $\delta_{h,0} = \mathcal{R}(h)^{F}  $ represents the impact of the shock $x_t$ on $y_{t+h}$ when including leads of the former. Since $\delta_{h,0}$ is the linear projection coefficient of equation~\eqref{eq:LPleads_app}, then:  

\begin{equation} \label{eq:LP_coef_F}
	\delta_{h,0}  = \frac{cov(y_{t+h},x_t)}{var(x_t)} \Bigr|_{\substack{x_{t+1}}} = \delta\frac{cov(x_{t+h},x_t)}{var(x_t)} \Bigr|_{\substack{x_{t+1}}}.
\end{equation}

Note that  the data generating process in system~\eqref{eq:DGP} considers that $x_t$ is an AR(1) so it can be represented in terms of $x_{t+1}$ (see equation~\eqref{eq:x_back}).\footnote{Note also that the results easily generalize to cases when $x_t$  is an autoregressive process of  higher order.} Then, we have that DLMs  and LPs with leads recover the same object:
\begin{equation*}
	\theta_{h}  = \frac{cov(y_{t+h},x_t)}{var(x_t)} \Bigr|_{\substack{x_{t+1},...,x_{t+h}}} = \delta\frac{cov(x_{t+h},x_t)}{var(x_t)} \Bigr|_{\substack{x_{t+1}}} = \delta_{h,0}. 
\end{equation*}

To see this, note that in period $h=0$ we have that:
\begin{equation*}
	\theta_{0}  = \frac{cov(y_{t},x_t)}{var(x_t)} \Bigr|_{\substack{x_{t+1},...,x_{t+h}}} = \frac{cov(y_{t},x_t)}{var(x_t)}  = \delta = \delta_{0,0}.
\end{equation*}

In periods $h>0$, we can rewrite equation~\eqref{eq:LP_coef_F} as:
\begin{equation} 
	\delta_{h,0}  = \delta\frac{cov(x_{t+h},x_t)}{var(x_t)} \Bigr|_{\substack{x_{t+1}}} = \delta \gamma^{h-1} \frac{cov(x_{t+1},x_t)}{var(x_t)} \Bigr|_{\substack{x_{t+1}}} = 0.
\end{equation}

Similarly, equation~\eqref{eq_MA_coef} becomes: 
\begin{equation} 
	\theta_{h}  = \frac{cov(y_{t+h},x_t)}{var(x_t)} \Bigr|_{\substack{x_{t+1},...,x_{t+h}}} = \delta \gamma^{h-1} \frac{cov(x_{t+1},x_t)}{var(x_t)} \Bigr|_{\substack{x_{t+1}}} = 0.
\end{equation}

So we have $ \mathcal{R}(h)^{F}= \mathcal{R}(h)^{DLM} $ $\forall$ $h$,$\gamma$. And we know that  $ \mathcal{R}(h)^{DLM}= \mathcal{R}(h)^{*}$ , from the section above.
\quad $\blacksquare$

\subsection{Proof of Proposition~\ref{prop:MAtoLP}}  \label{app:proof3} 
Consider a version of equation~\eqref{eq:MA_eps} rewritten here for convenience:

\begin{equation}
\label{eq:MA_eps_app}
y_t = \tilde{\theta}_0 \varepsilon_t + \tilde{\theta}_1 \varepsilon_{t-1} + \tilde{\theta}_2 \varepsilon_{t-2} + \tilde{\theta}_3 \varepsilon_{t-3} + \tilde{\theta}_4 \varepsilon_{t-4} \ldots + u_t, 
\end{equation}

where $\tilde{\theta}_0  = \mathcal{R}(h)^{DLM-per}  $ represents the impact of variable $\varepsilon_t$ on $y_{t+h}$. Note that $\varepsilon_t$ is not observable but can be obtained if we know the data generating process described in system~\eqref{eq:DGP}. Since equation~\eqref{eq:MA_eps_app} represents the linear projection of $y_t$ on $\varepsilon_t$ and its lags, with  $\varepsilon_t \sim \text{\emph{iid}}(\mu_\varepsilon,\,\sigma_\varepsilon^{2})$, we have:  

\begin{equation} \label{eq:MA_OLS_app}
	\tilde{\theta}_{h}  = \frac{cov(y_{t+h},\varepsilon_t)}{var(\varepsilon_t)} \Bigr|_{\substack{\varepsilon_{t+1},...,\varepsilon_{t+h}}} = \frac{cov(y_{t+h},\varepsilon_t)}{var(\varepsilon_t)} = \delta\frac{cov(x_{t+h}, \varepsilon_t)}{var(\varepsilon_t)}.
\end{equation}

This expression is equivalent to equation~\eqref{eq:LP_coef} which implies that $\tilde{\theta}_{h} = \delta_h$ and  $\mathcal{R}(h)^{DLM-per}  =\mathcal{R}(h)^{LP}$. To see this, substitute for $x_{t+h}$ in equation~\eqref{eq:MA_OLS_app} using expression~\eqref{eq:x_back} and  system~\eqref{eq:DGP}:
\begin{equation}  \label{eq:MA_OLSclean_app}
	\tilde{\theta}_{h}  = \delta \gamma^h \frac{cov(  x_t + \sum_{j=0}^{h-1}\gamma^j \varepsilon_{t+h-j},\varepsilon_t)}{var(\varepsilon_t)} = \delta \gamma^h \frac{cov(  x_t,\varepsilon_t)}{var(\varepsilon_t)} =
	\delta \gamma^h.
\end{equation}

Note that the above expression yields the same result as equation~\eqref{eq:LP_ols_final}, which shows that $\tilde{\theta}_{h} = \delta_h$ $\forall$ $h$.
\quad $\blacksquare$

\newpage
\section{Additional results}
\setcounter{figure}{0}
\setcounter{equation}{0}
\setcounter{footnote}{0}
\renewcommand{\thefigure}{B\arabic{figure}}
\renewcommand{\thesubsection}{B.\arabic{subsection}}
\renewcommand{\theequation}{B.\arabic{equation}}

\subsection{Including the shocks as endogenous variables in a VAR}  \label{app:VAR}

A researcher may consider including a shock with persistence as an endogenous variable in a VAR. Does this approach eliminate the effect of the persistence of the shock on the impulse responses?  A VAR, since it explicitly models the persistence of the shock, includes this effect in the estimated impulse responses and, hence, yields the same dynamic effects as LPs (contrary to what is obtained when including the shock as a distributed lag structure within a VAR).\footnote{\citet{bloom2009uncertainty}, \citet{romer2010macroeconomic}, and \citet{ramey2011identifying} are examples of studies that include shocks as endogenous variables in a VAR. These specifications are also known as hybrid VARs (see \citet{coibion2012monetary}). \citet{plagborg2018local} formally show that VARs and LPs identify the same impulse responses. Here we illustrate that when one of the endogenous variables in the VAR is a persistent shock, this effect will be carried over to the dynamic responses.}

To see this in an intuitive way, consider the data generating process given by system~\eqref{eq:DGPsimul} and rewritten here for convenience (with a slightly different notation):
\begin{eqnarray} \label{eq:DGPsimul_forVAR}
y_t &=& \rho y_{t-1} + \delta_0 x_t + \delta_1 x_{t-1} + \varepsilon^y_t \nonumber  \\
x_t &=& \gamma x_{t-1} + \varepsilon^x_t.
\end{eqnarray}

This process can be recast as a structural VAR of the form $\bm A_0 \bm Y_t = \bm B^* \bm Y_{t-1} +\bm \varepsilon_t$, with:

\begin{equation}
\begin{bmatrix}
 1 & 0 \\ 
 -\delta_0 & 1 
\end{bmatrix}
\begin{bmatrix}
 x_t \\ 
y_t
\end{bmatrix} =
\begin{bmatrix}
 \gamma & 0 \\ 
 \delta_1 & \rho 
\end{bmatrix}
\begin{bmatrix}
 x_{t-1} \\ 
y_{t-1}
\end{bmatrix} +
\begin{bmatrix}
\varepsilon^x_t \\ 
\varepsilon^y_t
\end{bmatrix}.
\end{equation}

An econometrician would estimate the following reduced-form VAR:
 \begin{equation} \label{eq:appVAR}
  \bm Y_t = \bm B\bm Y_{t-1} +\bm u_t,
 \end{equation} 

\noindent where  $\bm B = \bm A_0^{-1} \bm B^*$; and $\bm u_t= \bm A_0^{-1} \bm \varepsilon_t$ are reduced-form residuals. Since the data generating process given by equation \eqref{eq:DGPsimul_forVAR} already incorporates restrictions on the contemporaneous behavior of the variables, a researcher may identify the structural impulse responses by computing the Choleski decomposition (when $x_t$ is ordered first) of the variance-covariance matrix of reduced-form residuals $\bm u_t$.

However, note that, when $\gamma \neq 0$, the response of $y_t$ to $\varepsilon_t^x$ will include an effect due to the  persistence of the shock $x_t$. Intuitively, consider the case of $\rho=\delta_1=0$. In this scenario, a researcher may be interesting in recovering a \textit{one-off} shock to $x_t$. However, the response of $y_t$ will be given by $\mathcal{R}(h)^{VAR}=\delta_0 \sum_{k=0}^\infty \gamma^k \varepsilon^x_{t-k}$, that is, the \textit{one-off} shock will still have effects along the response horizon because of the peristence of $x_t$ (when $\gamma \neq 0$). 

To see this point in a more general way, consider the same calibration as used in the main text (i.e., $\rho=0.9$, $\delta_0=1.5$, $\delta_1=1$, and either $\gamma=0$ or $\gamma=0.2$). We compute the impulse responses of variables $y_t$ and $x_t$ (the measured shock) to $\varepsilon^x_t$, shown in Figure~\ref{fig:IRF_VAR}. We also estimate the same impulse-response functions for $y_t$ using LPs, as in Section~\ref{sec:simul}.

\begin{figure} \caption{Responses in a VAR to  $\varepsilon_t^x$}
\label{fig:IRF_VAR}
\vspace{0.5cm}
\begin{tabular}{c c}
Panel A) Response of $y_t$ to $\varepsilon_t^x$ & Panel B) Response of $x_t$ to $\varepsilon_t^x$  \\ 
		{\includegraphics[width=8cm]{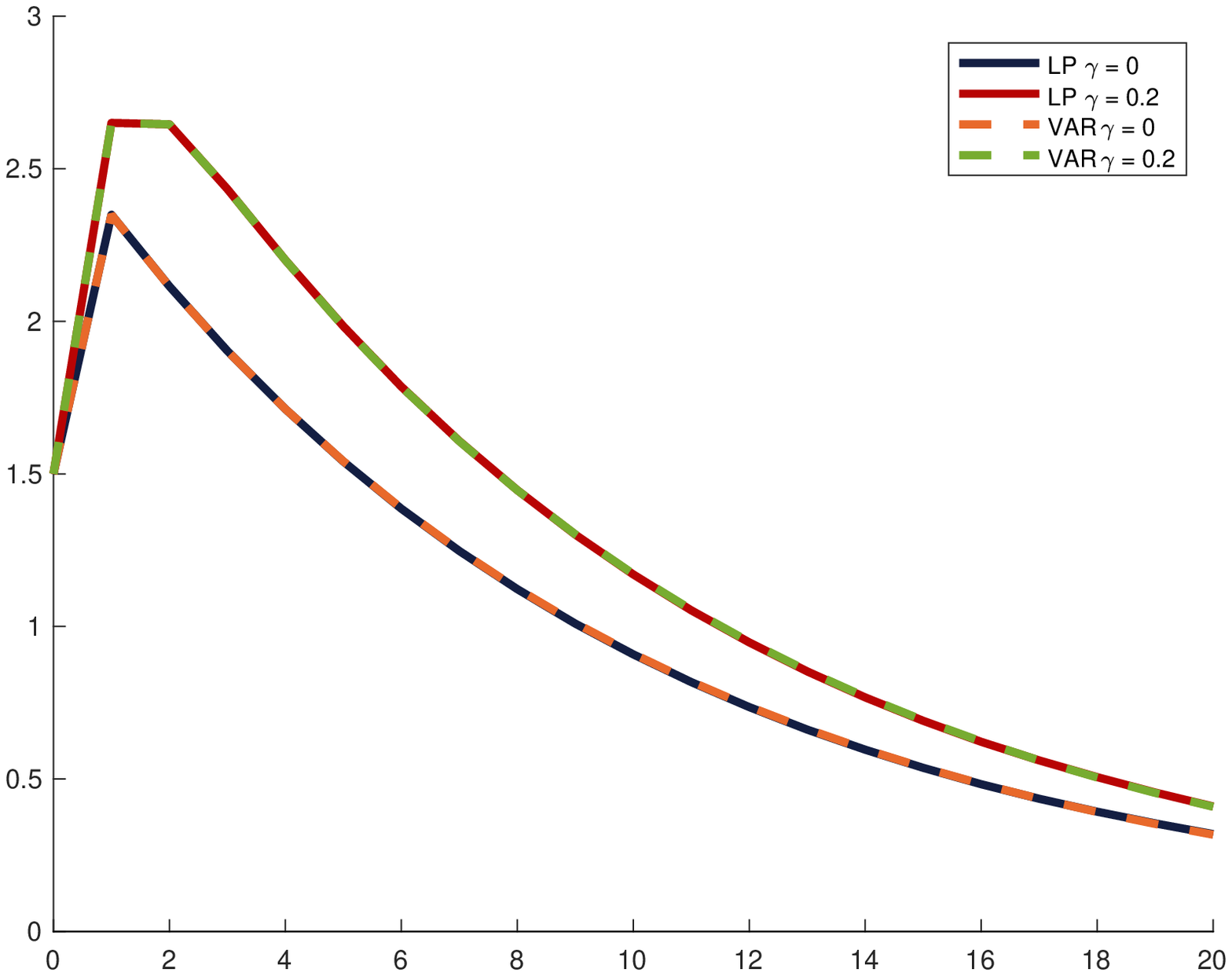}} &	 {\includegraphics[width=8cm]{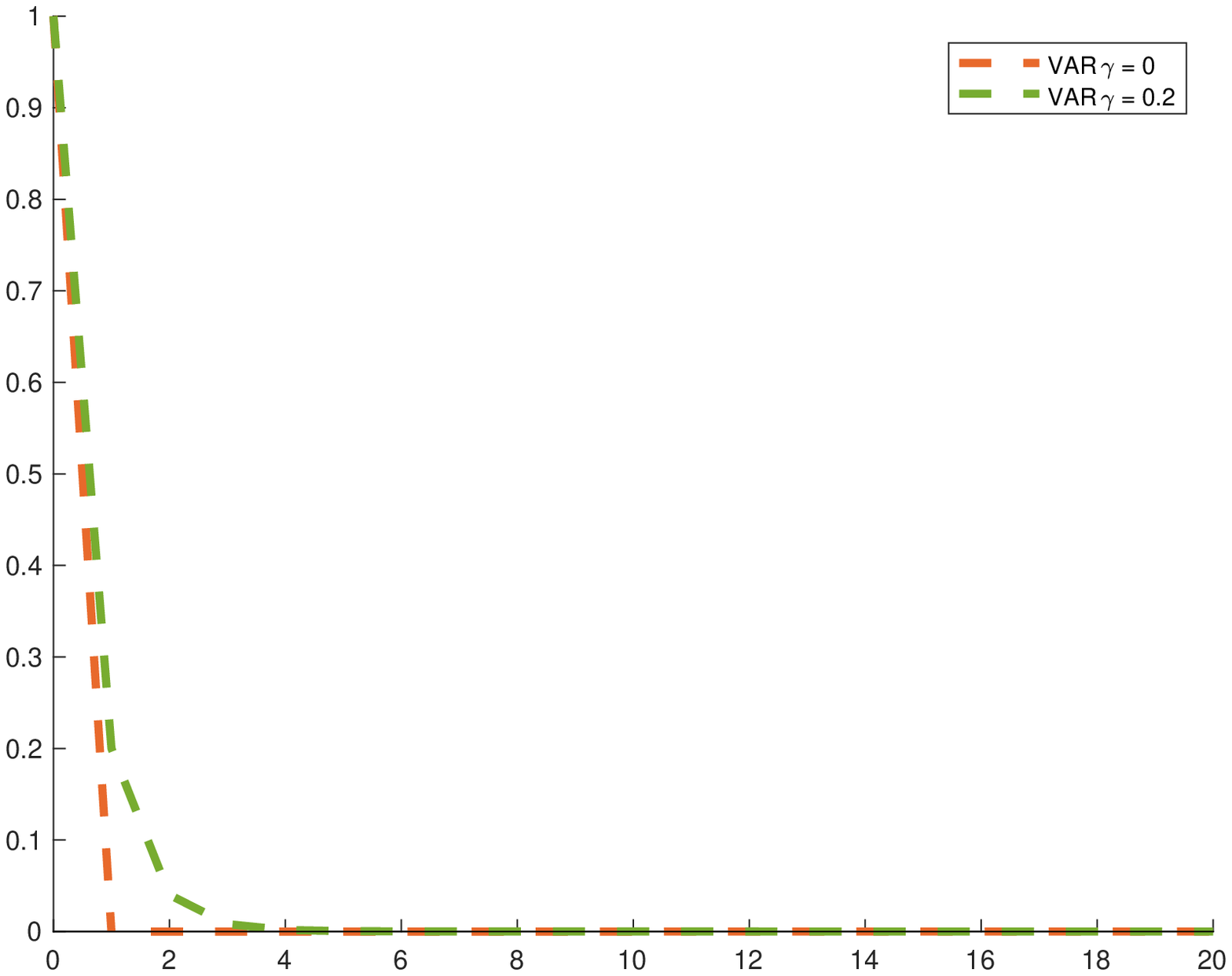}}
\end{tabular}
\footnotesize{The figure shows the VAR responses of $y_t$ (Panel A) and $x_t$ (Panel B) to  $\varepsilon^x_t$ estimated from \eqref{eq:appVAR}, under different assumptions of the persistence parameter $\gamma$: dashed orange lines for responses when there is no persistence ($\gamma=0$) and dashed green lines for responses when there is persistence ($\gamma=0.2$). For reference, Panel A also displays responses from the same DGP estimated using LPs for the cases of $\gamma=0$ (solid blue line) and $\gamma=0.2$ (solid red line).}
\end{figure}

The results illustrate that, regardless of the value $\gamma$, a VAR that considers $x_t$ as endogenous variable and LPs estimate the same impulse responses. When there is some persistence in the shock $x_t$, both the VAR (that considers $x_t$ as an endogenous variable) and LPs include a dynamic effect due to the persistence of $x_t$.

Importantly, when $x_t$ displays persistence, the response function estimated by a VAR will vary depending on whether $x_t$ is included as an endogenous variable (as shown before) or an exogenous one (e.g., with a distributed lag or moving average structure, as in \citet{mertens2012empirical} or \citet{favero2012tax}). To show this point, Figure~\ref{fig:IRF-VARx} displays the response of $y_t$ when (i)  $x_t$ is included as an endogenous variable in the VAR, or (ii) when estimating a regression of $y_t$ on $x_t$ and a lag of both variables. As it can be seen, considering $x_t$ as an exogenous regressor always delivers the same dynamic responses as if $x_t$ were serially uncorrelated, regardless of the actual value of $\gamma$. 

The discussion above highlights that the result regarding the equivalence of LPs and DLMs under no serial correlation of the shock can be generalized to multivariate settings of VARs with the shock included as an endogenous variable (in the case of LPs) or as an exogenous variable (in the case of DLMs). 

\begin{figure} 
\caption{Responses when the shock is included as an endogenous or exogenous variable in a VAR}\label{fig:IRF-VARx}
\begin{center}
		{\includegraphics[width=10cm]{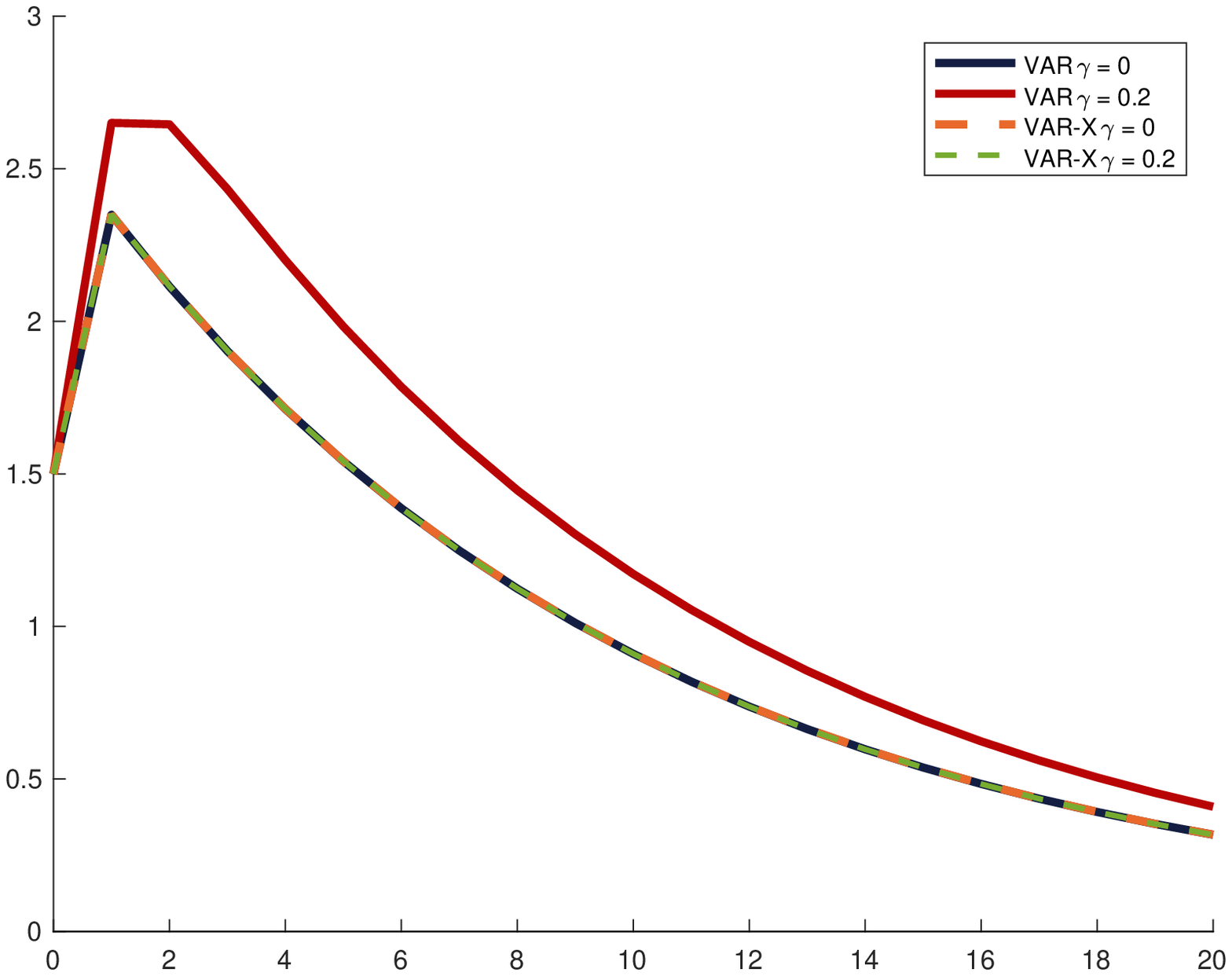}}	
\end{center}
\footnotesize{The figure shows the VAR responses of $y_t$ to $\varepsilon^x_t$ estimated from \eqref{eq:appVAR} under different assumptions of the persistence parameter $\gamma$ and two different specifications: solid lines display the responses when the shock is included as an endogenous variable in the VAR and dashed line shows the responses when the same shock is included as an exogenous variable in the VAR.}
\end{figure}

\subsection{Local projections with instrumental variables}\label{sec:IV}

Recently, there has been an increased attention to the use of external sources of variation as instruments in LPs.\footnote{See \citet{ramey2018government} for an example. Related to this, \citet{ramey2016macroeconomic} discusses the distinction between shock, innovation, and instrument. \citet{barnichon2020identifying} show how independently identified shocks can be used as instruments to estimate the coefficients of structural forward looking macroeconomic equations.} In this section, we investigate how persistence may affect the estimation of dynamic effects when using instrumental variables in local projections (LP-IV).

\citet{stock2018identification} provide the conditions under which a researcher can exploit external variation to estimate impulse response functions. A valid instrument $z_t$  should be both relevant and contemporaneously exogenous, that is, $z_t$ should not be correlated to any shock in the system except with the one that the researcher is interested in. Lastly, \citet{stock2018identification} impose a restriction called lead/lag exogeneity, which implies that the instrument should not be correlated with any lead or lag of any of the shocks in the system. 

Consider the following data generating process:
\begin{eqnarray} \label{eq:LP-IV}
		y_t &=& \beta g_{t} + u_t \nonumber \\
		u_t &=& m_t + a_t \nonumber \\
		g_t &=& \lambda x_t +  (1-\lambda)m_t  \\
		z_t &=& x_t + \nu_t  \nonumber\\
		x_t &=& \gamma x_{t-1} + \varepsilon_t, \nonumber
\end{eqnarray}
	
where $a_t $, $\nu_t$, $m_t$  and $\varepsilon_t$ follow independent $\mathcal{N}(0,1)$ distributions. A researcher may be interested in estimating the dynamic effects of variable $g_t$ on $y_t$ (e.g., the effects of government spending on output). However, $g_t$ is endogenous due to the presence of an omitted variable $m_t$. The researcher may have the availability of an instrument $z_t$,  which is contemporaneously exogenous by construction and relevant when $\lambda \neq 0$. This instrument, since it depends directly on the shock $x_t$,  displays persistence when $\gamma \neq 0$. When there is persistence in the instrument (and the shock), the lead/lag exogeneity condition mentioned above is not satisfied. To illustrate this point, we simulate system~\eqref{eq:LP-IV} setting $\beta=2$ (and different values of $\lambda$ and $\gamma$) for 100 million periods, and estimate the dynamic effects of $g_t$ on $y_t$ using LP.

We first consider the case of $\gamma=0$ and $\lambda=1$, that is, there is no problem of endogeneity or persistence. The estimated effect of $g_t$ on $y_t$ recovered by LPs is represented by a solid blue line in Panel A of Figure~\ref{fig:LP-IV}. As expected, the contemporaneous impact of government spending on output is equal to 2. When considering $\lambda=0.5$ (but still no persistence, i.e., $\gamma=0$), LPs that employ OLS will deliver biased estimates of the contemporaneous effect of $g_t$ (solid red line). The difference between the red and the blue lines in the first period is a measure of the endogeneity bias. The problem of endogeneity can be addressed by using $z_t$ as an instrument for $g_t$ to recover the exogenous variation in government spending (given by $x_t$). This result (still considering $\gamma=0$) is represented by the dashed  grey line in Panel A of Figure \ref{fig:LP-IV}, which shows how the use  of LP-IV  can overcome the presence of endogeneity, delivering a response function identical to the benchmark case without omitted variables bias.

\begin{figure} \caption{LPs with instrumental variables}
\label{fig:LP-IV}
\vspace{0.5cm}
\begin{tabular}{c c}
Panel A) $\gamma=0$ & Panel B) $\gamma = 0.2 $ \\ 
		{\includegraphics[width=8cm]{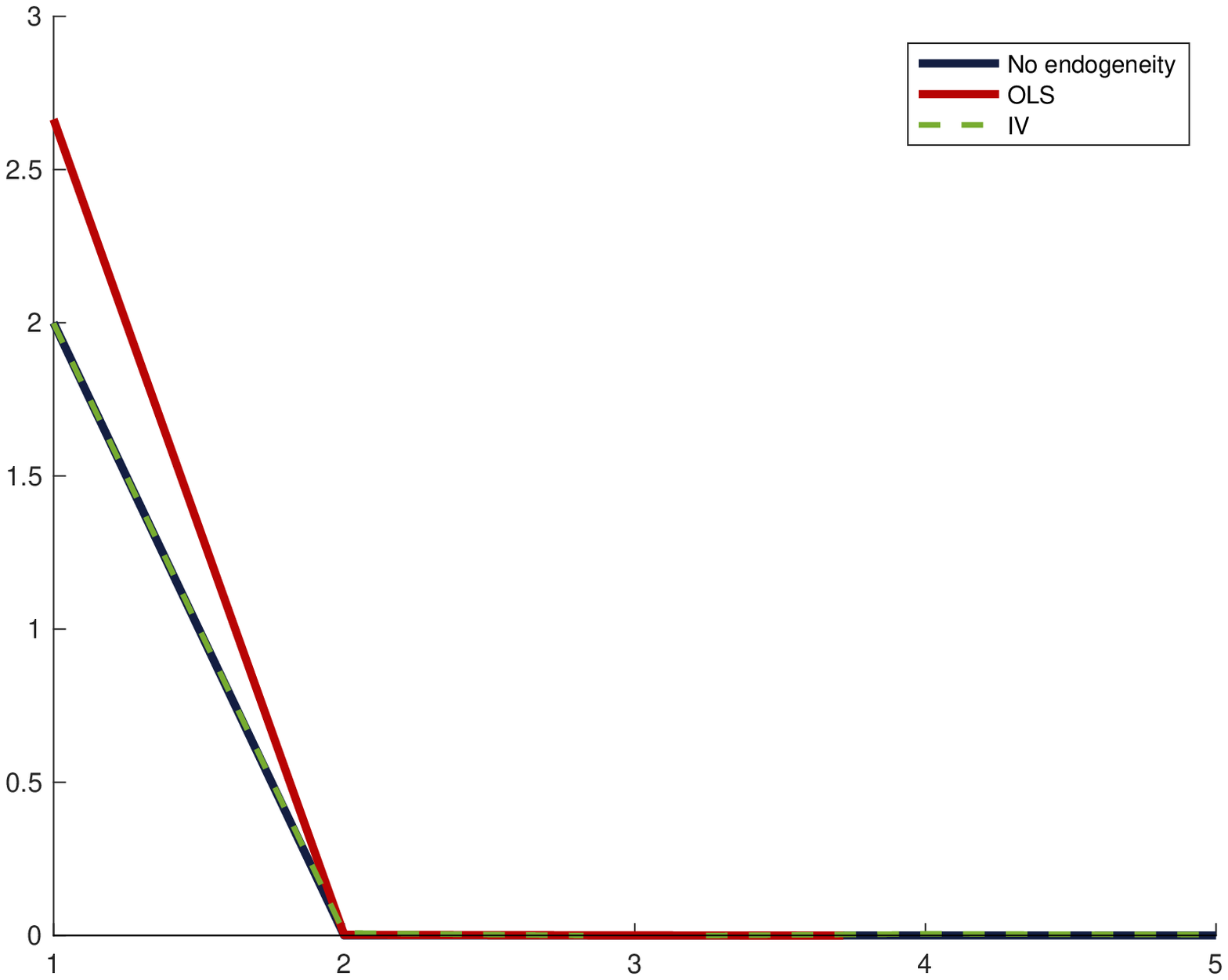}} &	 {\includegraphics[width=8cm]{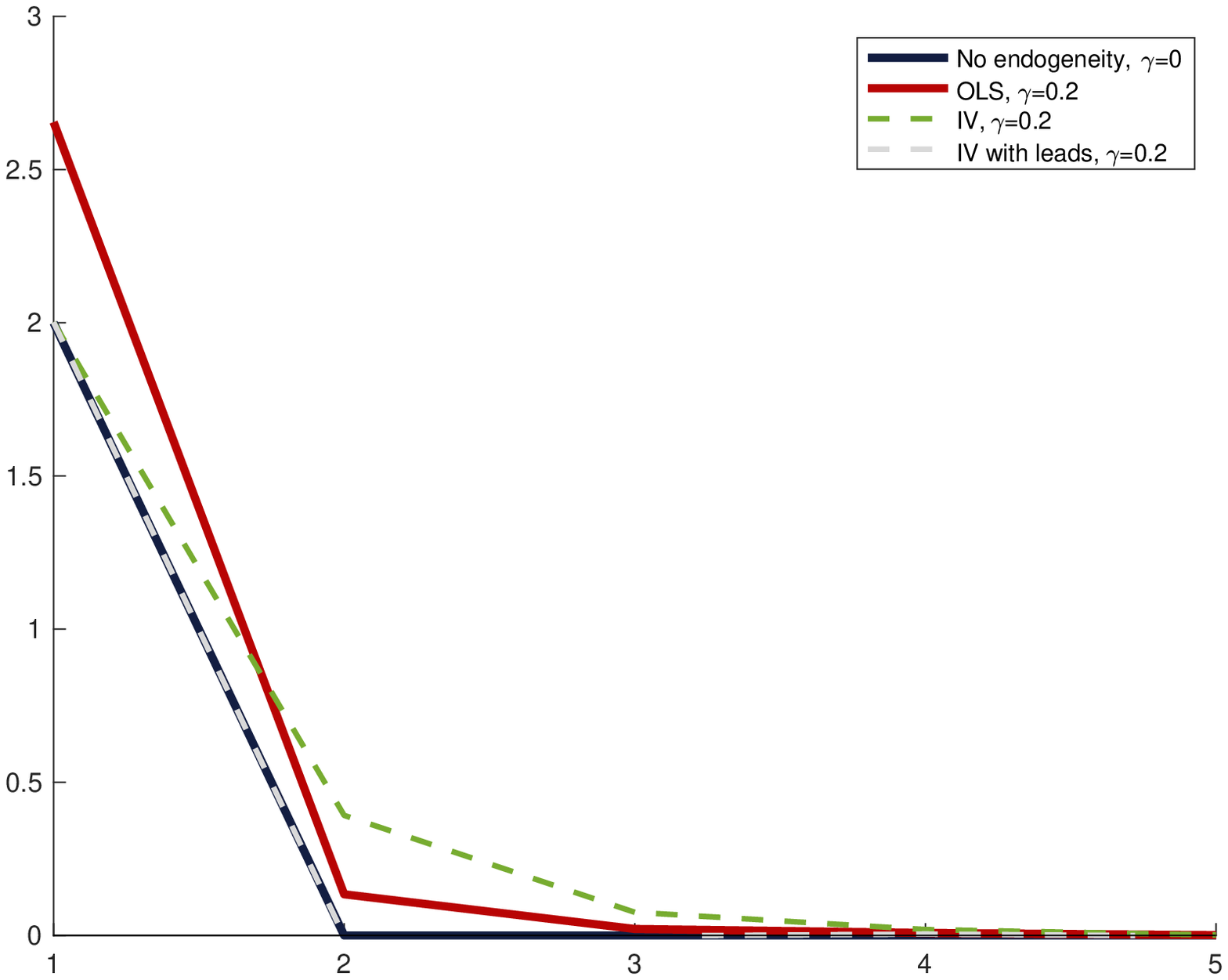}}
\end{tabular}
\footnotesize This figure shows the response of a simulated outcome variable to a shock  using local projections with instruments, with an underlying DGP given by system~\eqref{eq:LP-IV} and calibrated for different degrees of persistence in the shock ($\gamma=0$ in panel A and $\gamma=0.2$ in Panel B). In both panels, red lines refers to estimation using LPs estimated using OLS and green dashed lines refer to LPs estimated using instrumental variables, when the DGP generates endogeneity. For reference, the blue solid lines (in both panels) display responses when the DGP does not generate persistence or endogeneity. In Panel B, the grey pointed line displays responses estimated using instrumental variables in LPs and including leads of the shock. 
\end{figure}

Next, we repeat the previous exercise but now we allow for persistence in the instrument (due to persistence of the shock); in particular, we set $\gamma=0.2$. The results are shown in Panel B of Figure \ref{fig:LP-IV} (we still represent, in solid blue line, the benchmark case of $\gamma=\lambda=1$ for reference). When there is endogeneity and persistence, LPs estimates of the dynamic effects of $g_t$ are affected by both an endogeneity bias on impact, and by  the effect of persistence in the instrument during the rest of the response horizon (as shown in the previous section). This result is displayed by the solid red line in Panel B of Figure \ref{fig:LP-IV}, which is different from zero after impact. Now consider estimating the dynamic effects using LP-IV with an instrument $z_t$ (that displays persistence). The results (dashed green line) show that the use of the instrument addresses the problem of endogeneity (on impact, the effect from the LP-IV estimates is able to recover the true effect of $\beta=2$). However, the dynamic effect from the rest of the response horizon still reflects the presence of persistence. 

As discussed above, persistence in the instrument violates the lead/lag exogeneity condition. \citet{stock2018identification} state that, in general, this condition could be satisfied by the inclusion of further controls in the LP-IV regression. If the source of persistence is strictly restricted to the instrument, \citet{stock2018identification} show that the lead/lag exogeneity condition could be reestablished by including lags of the instrument. However, in cases like system~\eqref{eq:LP-IV}, where the instrument inherits its persistence from the shock, lags of the instrument will not satisfy the lead-lag exogeneity condition.  We build on the  intuition from \citet{stock2018identification} and adapt it to the problem of persistence  by including leads of the instrument in the set of exogenous variables in the LP-IV estimates. The results, shown in dashed grey lines in Panel B of Figure \ref{fig:LP-IV}, corroborate this intuition: despite the presence of both endogeneity and persistence, enhancing the LP-IV estimates with leads of the shock allows to recover the dynamic effects \emph{as if} the instrument were not serially correlated, i.e., $\mathcal{R}(h)^{*}$.

In sum, the presence of persistence can potentially violate the lead-lag exogeneity assumption, invalidating inference under LP-IV. The solution to reestablish this condition will depend on the source of persistence in the model. When the instrument inherits its persistence from the shock, our  proposed solution builds on the general intuition from \citet{stock2018identification}, showing that the inclusion of leads of the instrument can deliver valid inference under LP-IV.

\subsection{Alternative simulation: using the persistence from an actual shock} \label{app:simulreal}

In this subsection we compute the impulse response of a simulated variable $y_t$ to a shock $x_t$ with the following DGP:
\begin{equation} \label{eq:sim_realX}
y_t =  \rho y_{t-1} +  B_0 x_t + B_1 x_{t-1} + u_t, 
\end{equation}

\noindent where $x_t$ is the actual government spending  shock from \citet{ramey2018government} as shown in Panel D of Figure~\ref{fig:series}. $u_t$ is a random variable following $u_t \sim \mathcal{N}(0,\,1)$. We set $\rho=0.9$, $B_0=1.5$, and $B_0=1$.

Equation~\eqref{eq:sim_realX} is simulated for $497$ periods (the length of \citet{ramey2018government}'s shock), and we then compute the relevant IRFs. We repeat this process 10,000 times, and compute the average impulse responses across all repetitions. The results are shown in Figure~\ref{fig:simul_realX}. 

\begin{figure}[ht] 
\caption{Simulated responses using LPs with persistence from an actual shock}\label{fig:simul_realX}
\begin{center}
		{\includegraphics[width=10cm]{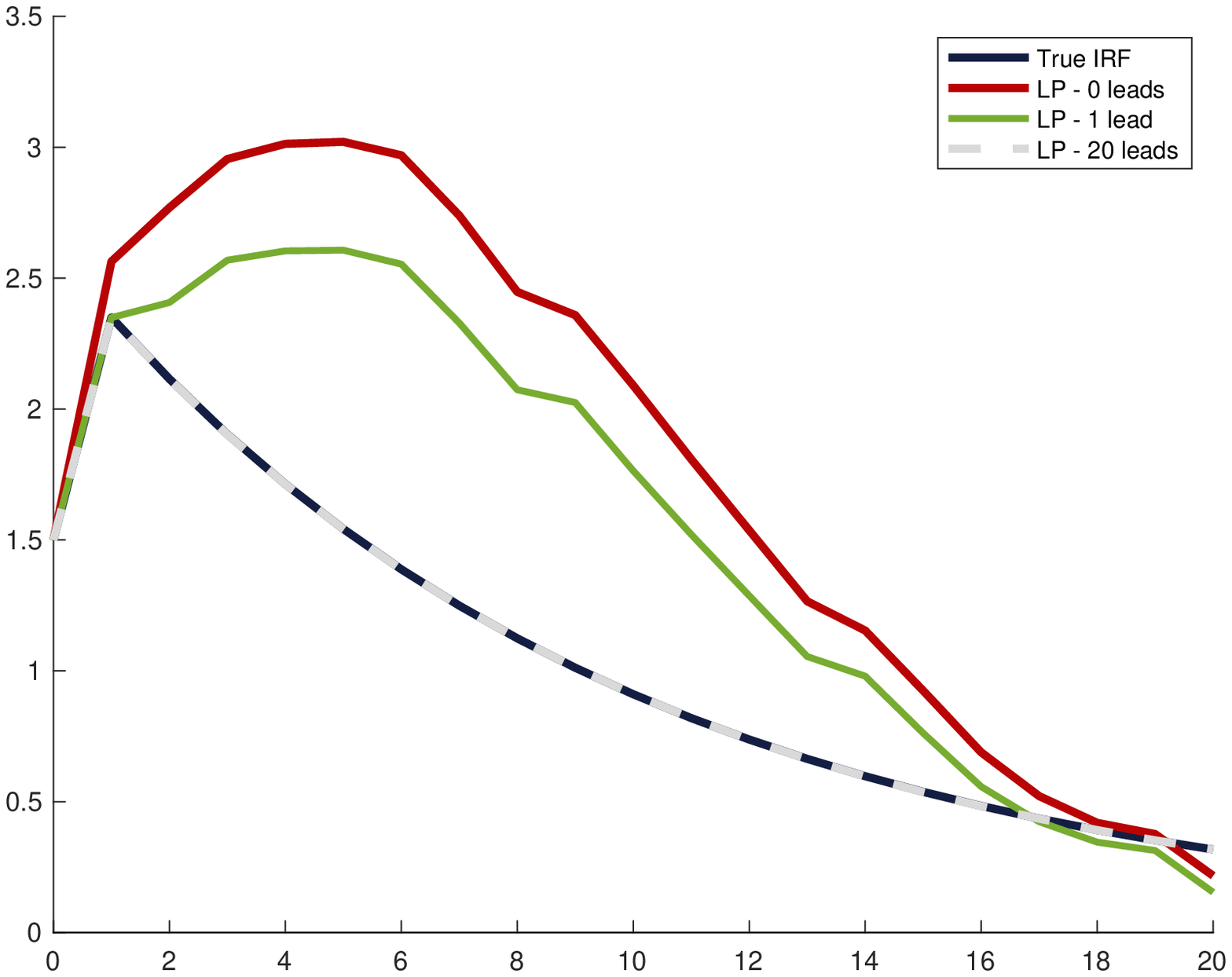}}	
\end{center}
\footnotesize This figure shows the response of a simulated outcome variable to  the government spending  shock from \citet{ramey2018government}. The dark blue line is the theoretical impulse-response to a shock that shows no persistence. The red line shows the LPs estimation of the impulse-response to the \citet{ramey2018government} without including any lead. Green line repeats the same estimation adding one lead. Dashed grey line shows the response when including 20 leads.
\end{figure}

When computing the dynamic response with standard LPs (i.e., without including any lead), the estimates diverge from the expected response when the shock has no persistence (distance between red and dark blue lines in Figure~\ref{fig:simul_realX}). Adding one lead improves the estimates, bringing the impulse-response into line with the theoretical response in the first period (green line). The accuracy of the impulse-response converges to the theoretical response when more leads are included. When we include as many leads as periods in the response horizon (20), the dynamic response estimated from LPs using the actual shock (with persistence) is equivalent to the response to a non-serially correlated shock (dashed grey line).

\subsection{Responses in LPs using variables adjusted for serial correlation}  
\label{sec:app_lags}
An apparent potential alternative to the use of leads proposed in the main text might be to adjust the shock $x_t$ so that it does not display persistence (e.g., by regressing $x_t$ on its own lags and using the resulting residual). Once the persistence is removed, one may expect the dynamic responses not to include the effect due to the persistence of the shock. However, this is not the case in a LPs setting, as we show next.

Consider the case where we obtain  a variable adjusted for serial correlation: ${\varepsilon}_t = x_t - {\gamma} x_{t-1}$, as shown in equation~\eqref{eq:epshat}. Then, $\varepsilon_t$ can be used as substitute of the original shock $x_t$. Assuming $B_1=0$ in system~\eqref{eq:DGPsimul} (for simplicity) consider the following series of LPs:

\begin{equation}
\label{eq:genregLP}
y_{t+h} = \rho_h y_{t-1} + \lambda_h  {\varepsilon}_{t}  +\xi_{t+h}.
\end{equation}

To obtain the dynamic responses of $y_t$ to the shock ${\varepsilon}_t$ (adjusted for persistence), we rewrite the first equation in system~\eqref{eq:DGPsimul} as a function of ${\varepsilon}_t$ and compute the relevant partial derivatives. For the cases of $h=0$ and $h=1$ these are:
\begin{eqnarray}\label{eq:genregLParray}
	\lambda_0 =\frac{\partial y_{t+1}}{\partial {\varepsilon}_t} &=& B_0 \nonumber \\
	\lambda_1 =  \frac{\partial y_{t+1}}{\partial {\varepsilon}_t} &=& \rho \frac{\partial y_t}{\partial {\varepsilon}_t} + B_0 \frac{\partial x_{t+1}}{\partial{\varepsilon}_t} = \rho B_0 + B_0 {\gamma} = B_0 ({\gamma}+\rho).
\end{eqnarray}

	That is, even after correcting for the persistence in shock $x_t$, conventional LPs yield responses $\mathcal{R}(h)$, i.e.,  still containing the effect of persistence of the shock. 

While this result may seem counter-intuitive, it arises from the fact that LPs do not have an explicit dynamic structure as a DLM . Hence, removing the persistence from $x_t$ does not eliminate its effect on $y_{t+1}$, $y_{t+2}$, etc.

To empirically show this point, we simulate series of $y_t$ and $x_t$ following system~\eqref{eq:DGPsimul} and the calibration used in Section~\ref{sec:simul} (we now allow $B_1\neq0$). We then obtain the residuals $\hat{\varepsilon}_t$ as an estimate of $\varepsilon_t$ described above and estimate the following equation:

\begin{equation} \label{eq:genregLPleadsSIMUL}
y_{t+h}=\rho y_{t-1} + \lambda_{h,0} \hat{\varepsilon}_t +  \lambda_{h,1} \hat{\varepsilon}_{t-1} + \xi_{t+h}.
\end{equation}

\begin{figure} 
\caption{Simulated responses using $\hat{\varepsilon}_t$}\label{fig:simul_genreg}
\begin{center}
		{\includegraphics[width=10cm]{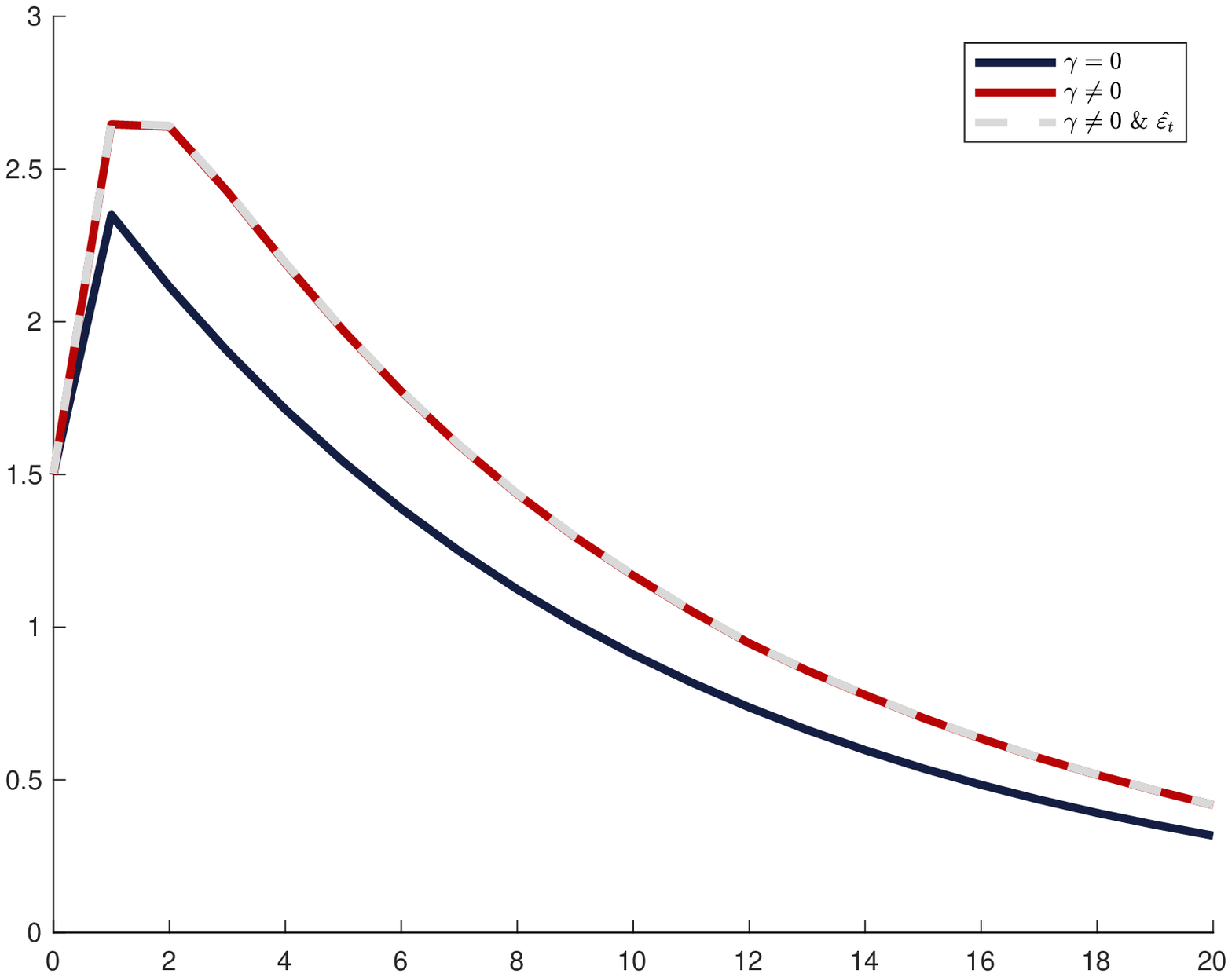}}	
\end{center}
\footnotesize This figure shows the response of a simulated outcome variable to a shock with different degrees of persistence. The dark blue line shows the results of estimating equation~\eqref{eq:genregLPleadsSIMUL} assuming $\gamma = 0$ in equation~\eqref{eq:DGPsimul}. The red line shows the same estimation when $\gamma = 0.2$. The dashed grey line shows the response when including a predicted regressor where persistence has been removed as explanatory variable (as in equation~\eqref{eq:genregLP}).
\end{figure}

Results are shown in Figure~\ref{fig:simul_genreg}. The simulations corroborate the above results and we find that the use of a variable adjusted for serial correlation as $\hat{\varepsilon}_t$ in equation~\eqref{eq:genregLP} fails to retrieve an impulse response as the one obtained when $\gamma = 0$ in equation~\eqref{eq:DGPsimul}.

\clearpage
\newpage
\section{Additional empirical applications} \label{app:monetary}
\setcounter{figure}{0}
\setcounter{table}{0}
\setcounter{equation}{0}
\setcounter{footnote}{0}
\renewcommand{\thefigure}{C\arabic{figure}}
\renewcommand{\thesubsection}{C.\arabic{subsection}}
\renewcommand{\theequation}{C.\arabic{equation}}
\renewcommand{\thetable}{C.\arabic{table}}


\subsection{\citet{guajardo2014expansionary}} 

In this subsection we explore the relevance of our results in the context of episodes of  fiscal consolidation, as produced in \citet{guajardo2014expansionary}. 
The authors  employ a panel of OECD economies to analyze the response of  economic activity to discretionary changes in fiscal policy motivated by a desire to reduce the budget deficit and not correlated with the short-term economic outlook.\footnote{A detailed description of these shocks can be found in \citet{pescatori2011new}.} As mentioned in Table~\ref{tab:survey}, this measure of fiscal changes  exhibits some degree of persistence.\footnote{Regressions of the fiscal consolidations measure (expressed as \% of GDP) on its own lags and including time and country fixed effects reveal persistence in the previous two or three years (depending on the number of lags included). Intuitively, some degree of persistence is expected in these series since they often involved multi-year plans, as noted in \citet{alesina2015output} and \citet{alesina2017effects}.} 

To explore the effects of persistence in this context, we compute the responses estimating a series of LPs:\footnote{Note that \citet{guajardo2014expansionary} do not construct responses using LPs and hence their computed responses do not show the effect of persistence, as noted in the previous section. There are, however, a number of studies that employ their fiscal consolidations dataset with LPs (see, for example,  \citet{barnichon2017understanding} or \citet{goujard2017cross}).}


\begin{equation} \label{eq:GLP2016}
y_{i,t+h}=\mu_{h,i} + \lambda_{h,t} + \beta_{h,0} shock_{i,t} +\sum_{f=1}^h  \beta_{h,f} shock_{i,t+f} +  \bm \beta_{h,s} \bm X_{i,t} + \xi_{i,t+h},
\end{equation}

\noindent where $y_{i,t}$ is a measure of economic activity (either private consumption or real GDP), $\mu_{h,i}$ and $\lambda_{h,t}$ represent country and time fixed effects, respectively, and $\bm X_{it}$ is a vector of variables that includes a lag of the shock, output, and private consumption, and a deterministic trend. In our setting, responses to the fiscal shocks are given by the estimates of coefficients $\beta_{h,0}$ for different horizons $h$.

We first estimate equation~\eqref{eq:GLP2016} by setting $\beta_{h,f}=0$ $\forall h,f$. The results, shown in black solid lines in Figure~\ref{fig:Guajardo_respGDP} qualitatively replicate the benchmark results of \citet{guajardo2014expansionary}, with a fiscal consolidation shock significantly reducing output during the first 6 years.\footnote{\citet{guajardo2014expansionary} focus on the dynamic effects of output and private consumption during 6 years after the shock. We also compute results for private consumption, shown in Figure~\ref{fig:Guajardo_respCONS} in Appendix~\ref{app:morefigs}. As in the original paper, we also find a significant reduction in this variable during the first 6 years after a consolidation shock.}

\begin{figure}
\caption{Output response to a fiscal consolidation shock, with and without leads} \label{fig:Guajardo_respGDP}
\begin{center}
		{\includegraphics[width=10cm]{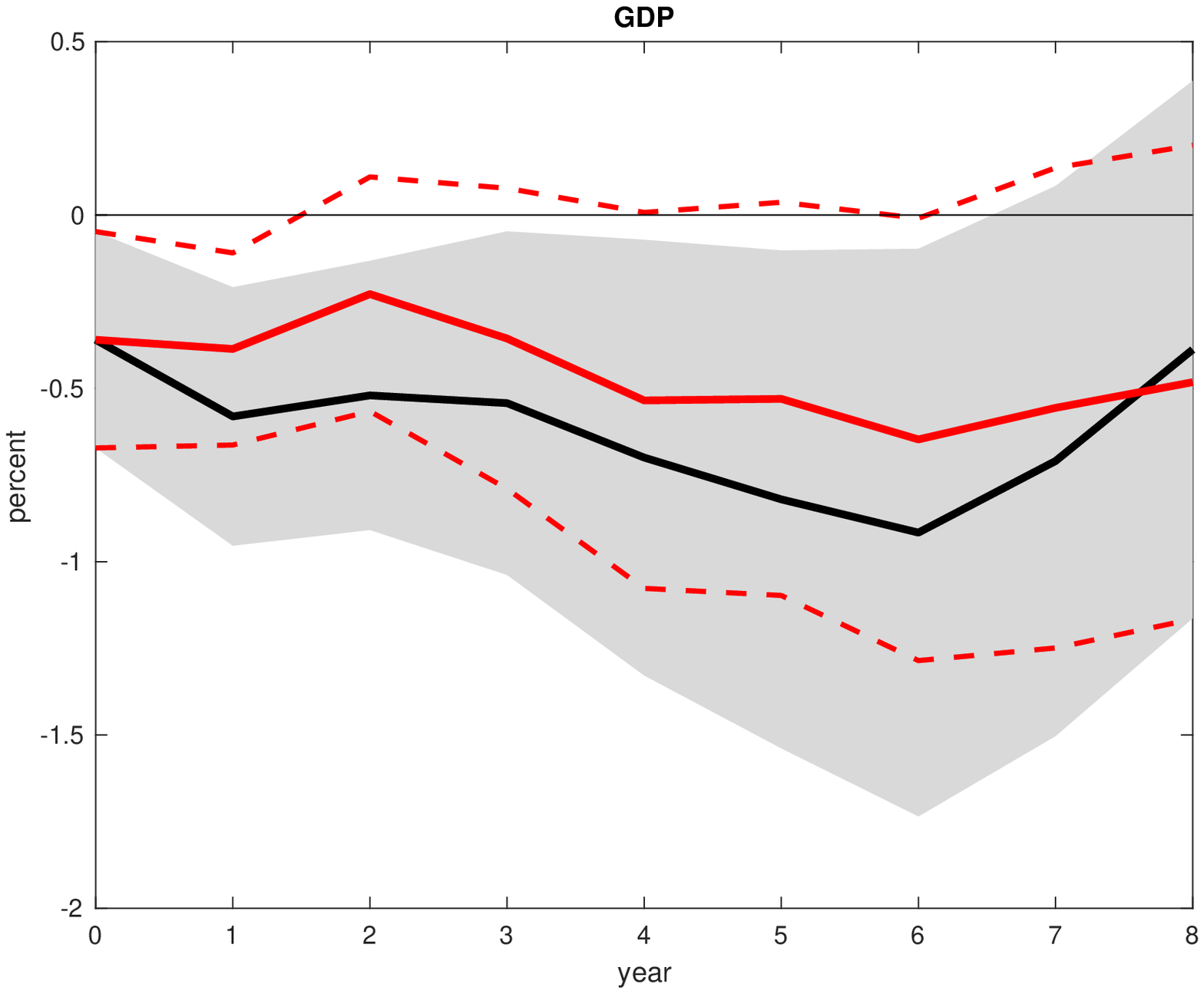}}	
\end{center}
\footnotesize Black lines show the results from equation~\eqref{eq:GLP2016} with output as dependent variable and setting $\beta_{h,f}=0$ $\forall h,f$, i.e., without including any leads  of the shock. Grey areas represent  90\% Newey-West confidence intervals for these estimates (as in \citet{guajardo2014expansionary}). Red solid lines represent the results of estimations when allowing  $\beta_{h,f} \neq 0$ and including $h$ leads of the consolidations variable.
\end{figure}

Next, we estimate equation~\eqref{eq:GLP2016} but allow $\beta_{h,f} \neq 0$ (red lines in Figure~\ref{fig:Guajardo_respGDP}). Three points are worth noting regarding these results. First, when accounting for the effects of persistence, the point estimates are smaller in absolute value. On average, the new responses are 35\% lower during the first six years after the shock. Two years after a fiscal consolidation, output is almost 60\% smaller when accounting for persistence (-0.2 \textit{vs} -0.5). 

Second, when including leads of the shock, the estimates are more precise, which translates into smaller confidence intervals (set at 90\% as in the original paper of \citet{guajardo2014expansionary}). During the first six years, these intervals are about 20\% smaller on average in the specifications that include leads of the shock. 

Third, these narrower intervals now include zero for most of the response horizon. Ignoring the persistence of the shock would lead to the conclusion that the output contraction after a fiscal consolidation is significant throughout the six years after the shock. However, when accounting for persistence, the effect of the shock is significant only during the first year after the shock, while it seems less plausible to conclude that the effect is statistically different from zero during the rest of the response horizon. 

This exercise suggests that the policy implications from fiscal consolidations may be different when estimating $\mathcal{R}(h)$ \textit{vs.} $\mathcal{R}(h)^{*}$.


\subsection{\citet{romer2010macroeconomic}}


What happens when including leads of non-persistent shocks? In this section we conduct a \textit{placebo} test based on \citet{romer2010macroeconomic}, who investigate the output effects of legislated tax changes. \citet{romer2010macroeconomic} identify exogenous changes in tax revenues by classifying fiscal reforms according to their motivation (i.e., whether or not they are the response to changing macroeconomic conditions). As discussed in Section \ref{sec:evidence}, it is the only shock considered here for which we unambiguously fail to reject the null hypothesis of no persistence. Hence, the inclusion of leads of the shock should not have a discernible impact on the estimation of dynamic responses. Beyond corroborating the previous statement, this subsection shows that the unnecessary inclusion of leads does not negatively affect inference in this application.

We estimate the response of output to exogenous tax changes following \citet{romer2010macroeconomic}.  We adapt the original estimation from the authors to the LPs setting:\footnote{Adding controls such as lags of output or the own shock do not affect the obtained results shown next.}

\begin{equation} \label{eq:RR2010}
\frac{ y_{t+h}-y_{t-1}}{y_{t-1}}= \beta_{h,0} shock_{t} +  \sum_{f=1}^h \beta_{h,f} shock_{t+f} + \xi_{t+h}.
\end{equation}

In our first exercise, we set $\beta_{h,f}=0$ $\forall h,f$ in equation~\eqref{eq:RR2010} to replicate the results from \citet{romer2010macroeconomic}. The results are shown Figure~\ref{fig:RomerRomer2010} (black lines). The response of output is similar to that in \citet{romer2010macroeconomic}: it falls persistently after a tax hike of 1\% of GDP, with a peak effect reached in the 10th quarter.\footnote{The difference with the original estimations from \citet{romer2010macroeconomic} are only quantitative: the peak tax multiplier is about 3 in the 10th quarter. Our estimations suggest a peak multiplier of 2.25 also reached in the same quarter.}

\begin{figure}
\caption{Response of output to \citet{romer2010macroeconomic} tax shocks, with and without leads}\label{fig:RomerRomer2010}
\begin{center}
		{\includegraphics[width=10cm]{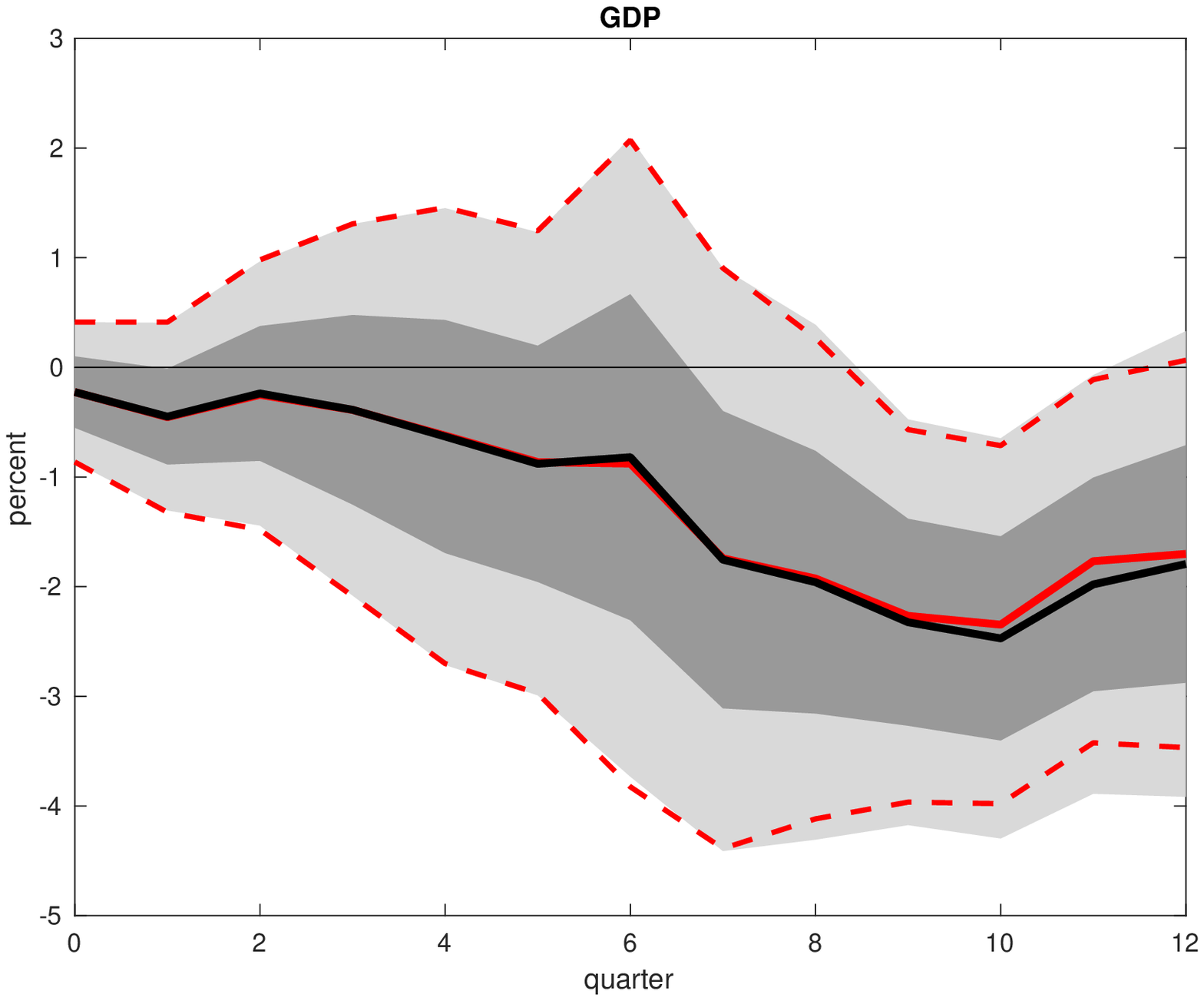}}
\end{center}
\footnotesize Black solid line shows the responses to a tax shock estimated from equation\eqref{eq:RR2010} with  $\beta_{h,f}=0$, i.e., without including any lead. Grey areas represent 68 and 95\% Newey-West confidence intervals for these estimates. Red solid line shows the responses to a tax shock estimated from equation~\eqref{eq:RR2010} with  $\beta_{h,f} \neq 0$ and including $h$ leads of the shock. Red dashed lines represent 95\% Newey-West confidence intervals for these estimates.
\end{figure}

Next, we allow for $\beta_{h,f}\neq0$. The results, shown in Figure~\ref{fig:RomerRomer2010} (red lines),  suggest that the inclusion of leads does not significantly affect the results. The point estimations with and without leads of the shock overlap each other for most of the response horizon and only diverge slightly during the quarters 8 to 11th.

While, given the results of Table~\ref{tab:survey} we should not expect a change in the point estimates (which we have corroborated) the same cannot be say about issues regarding inference. However, Figure~\ref{fig:RomerRomer2010} shows that confidence bands are not distinguishable between both specifications during the first seven quarters and differ only slightly afterwards.

In sum, this placebo exercise is reassuring in that the inclusion of leads only matters when the explanatory variable displays some persistence. These results suggest that including leads in LPs is a conservative way to address the effects of persistence when there is a suspicion that the shock is persistent and the researcher wants to identify  $\mathcal{R}(h)^{*}$.\footnote{See \citet{alloza2020local} for another example that adds leads to LPs using a non-persistent shock. Similarly to the evidence provided in this section, they also show that adding leads does not affect inference.}

\subsection{ \citet{romer2004measure}}

We now consider the measure of monetary policy shocks produced by  \citet{romer2004measure}. The authors identify exogenous monetary policy changes following a three-step procedure. First, they follow narrative methods to identify the Federal Reserve's intentions for the federal funds rate around FOMC meetings. Second, they regress the resulting measure on the Federal Reserve's internal forecasts (Greenbook) to account for all relevant information used by the Fed. Lastly, the series is aggregated from FOMC frequency to monthly frequency.

As shown in Table~\ref{tab:survey}, the resulting measure displays some degree of  persistence.\footnote{The degree of persistence is higher when using the updated series produced by \citet{coibion2012monetary}.} Interestingly, the correlogram of the series seems to show a pattern consistent with  \textit{negative} persistence (see Panel E in Figure~\ref{fig:ac}). This implies that standard LPs that do not account for persistence in the shock will identify $\mathcal{R}(h)$.

\citet{romer2004measure} estimate the response of output to the monetary policy shock using a lag-distributed regression of the log of industrial output and the measure of monetary policy shocks. Here, we  adapt the estimation to a LPs setting by following the exact data and specification from  \citet{ramey2016macroeconomic} (adapted in turn from \citet{coibion2012monetary}) for the original sample of 1969m3-1996m12):

 \begin{equation} \label{eq:RR2004}
 y_{t+h}= \beta_{h,0} shock_{t} + \bm{ \theta_h(L)  {x_t} }+ \beta_{h,f} \sum_{f=1}^h shock_{t+f} + \xi_{t+h},
\end{equation}

 \noindent where $y_t$ is either the federal funds rates, the log of industrial production, the log of consumer price index, the unemployment rate, or the log of a commodity price index, and $shock_t$ is the original \citet{romer2004measure} measure of monetary policy shocks. The regressions include a set of controls $\bm x_t$, with two lags and the contemporaneous values of all dependent variables, and two lags of the shock. By including the contemporaneous values of the the dependent variables, we are implementing the recursiveness assumption often used in VARs to identify monetary policy shocks.\footnote{This assumption implies that the monetary shock does not affect macroeconomic variables (such as output, prices, employment...) contemporaneously, and monetary variables (e.g., money stock, reserves...) do not affect the federal funds rates within a month. See \citet{christiano1999monetary} for further details. Later on, we show estimates that relax this assumption (Figure~\ref{fig:RR2004_nonrecur}).}
 
The results of estimating equation~\eqref{eq:RR2004} when we set  $\beta_{h,f}=0$ $\forall h,f$ are shown in solid black lines (with 90\% confidence bands) in Figure~\ref{fig:RR2004}. Since we employ the same data and specification, they replicate the results from \citet{ramey2016macroeconomic}  (Figure 2B). \citet{ramey2016macroeconomic} argues that the responses using LPs show more plausible dynamics than those obtained from a standard VAR (a persistent fall in industrial output and a rise in unemployment that slowly converge to 0). The drop in output after a monetary shock is broadly consistent with the original results from \citet{romer2004measure} but there are, however, two important differences. First, the trough in the response of output is reached after two years. In the estimates of Figure~\ref{fig:RR2004} and \citet{ramey2016macroeconomic}, the trough is reached after a first year and lasts for about twelve months with a slight rebound in between. Second, although both results refer to the same impulse (a realization of the policy measure of one percentage point), the magnitude of the output fall in the original estimates of \citet{romer2004measure} is substantially bigger than when using LPs (-4.3 \textit{vs} -1.7). 
 
 \begin{figure}
\caption{Responses to monetary policy shock from  \citet{romer2004measure}, with and without leads}\label{fig:RR2004}
\begin{center}
		{\includegraphics[width=10cm]{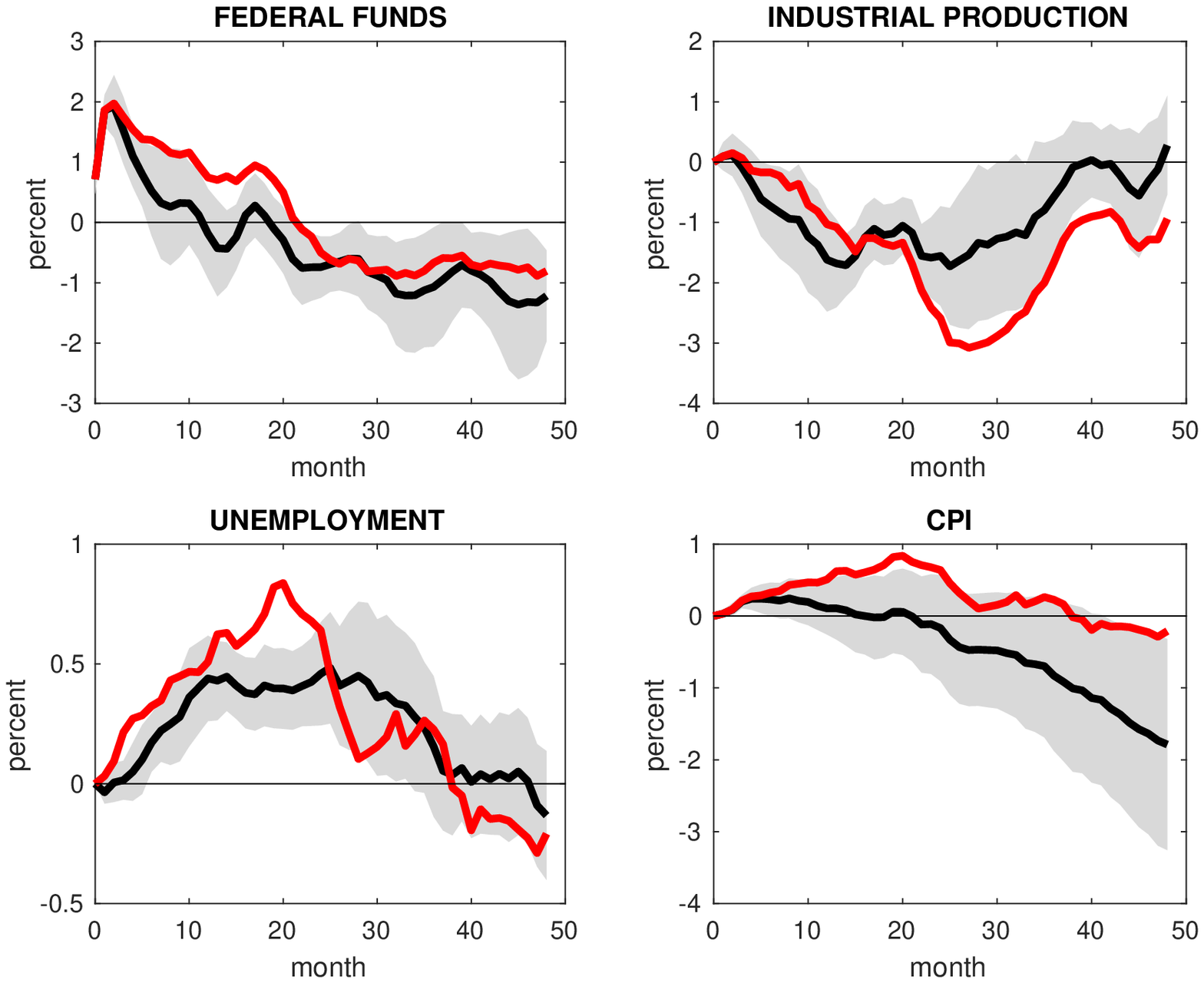}}	
\end{center}
\footnotesize Black solid lines refer to a benchmark specification that preserves the recursive assumption and does not include leads of the shock. Grey areas show 90\% Newey-West confidence intervals. Red solid lines include $h$ leads of the monetary shocks. \end{figure}


Next, we investigate whether accounting for the persistence in  the shock has an effect on the dynamic responses. We re-estimate equation~\eqref{eq:RR2004}, but allowing  $\beta_{h,f} \neq 0$. The results are shown in red lines in Figure~\ref{fig:RR2004}. We observe that the dynamics of output are closer to the original estimates of  \citet{romer2004measure}: a continuous drop in output  that reaches the trough after the second year. Furthermore, the magnitude of the fall is now substantially higher (-3.1)  and closer to the results from  \citet{romer2004measure}. Another noticeable difference is that the effects on unemployment and the initial positive reaction on prices (the so-called \textit{price puzzle}) are now larger. All in all, the results from Figure~\ref{fig:RR2004} suggest that accounting for the persistence of the monetary policy shock can lead to larger estimates of the dynamic responses.

\citet{ramey2016macroeconomic} also investigates the role of the recursiveness assumption  in the dynamic responses (the inclusion of contemporaneous values for some variables in the LPs estimation to replicate the identification in a VAR). She finds that relaxing this assumption results in weird dynamics of unemployment in the short run. We replicate these results by dropping the contemporaneous values in  $\bm x_t$ in equation~\eqref{eq:RR2004} and setting $\beta_{h,f}=0$. We indeed find that unemployment rate significantly drops in the first months after a monetary policy contraction (black solid lines in Figure~\ref{fig:RR2004_nonrecur}). We investigate whether these strange dynamics may be the result of the persistence in the monetary policy shock. We estimate again equation~\eqref{eq:RR2004}  relaxing both the recursiveness assumption and allowing  $\beta_{h,f} \neq 0$. The results  (red solid lines in Figure~\ref{fig:RR2004_nonrecur}) are very similar to those from Figure~\ref{fig:RR2004}. Interestingly, unemployment responds positively to the monetary policy contraction. 

\begin{figure}[ht]
\caption{Responses to monetary policy shock from  \citet{romer2004measure} with no recursive assumption, with and without leads}\label{fig:RR2004_nonrecur}
\begin{center}
		{\includegraphics[width=10cm]{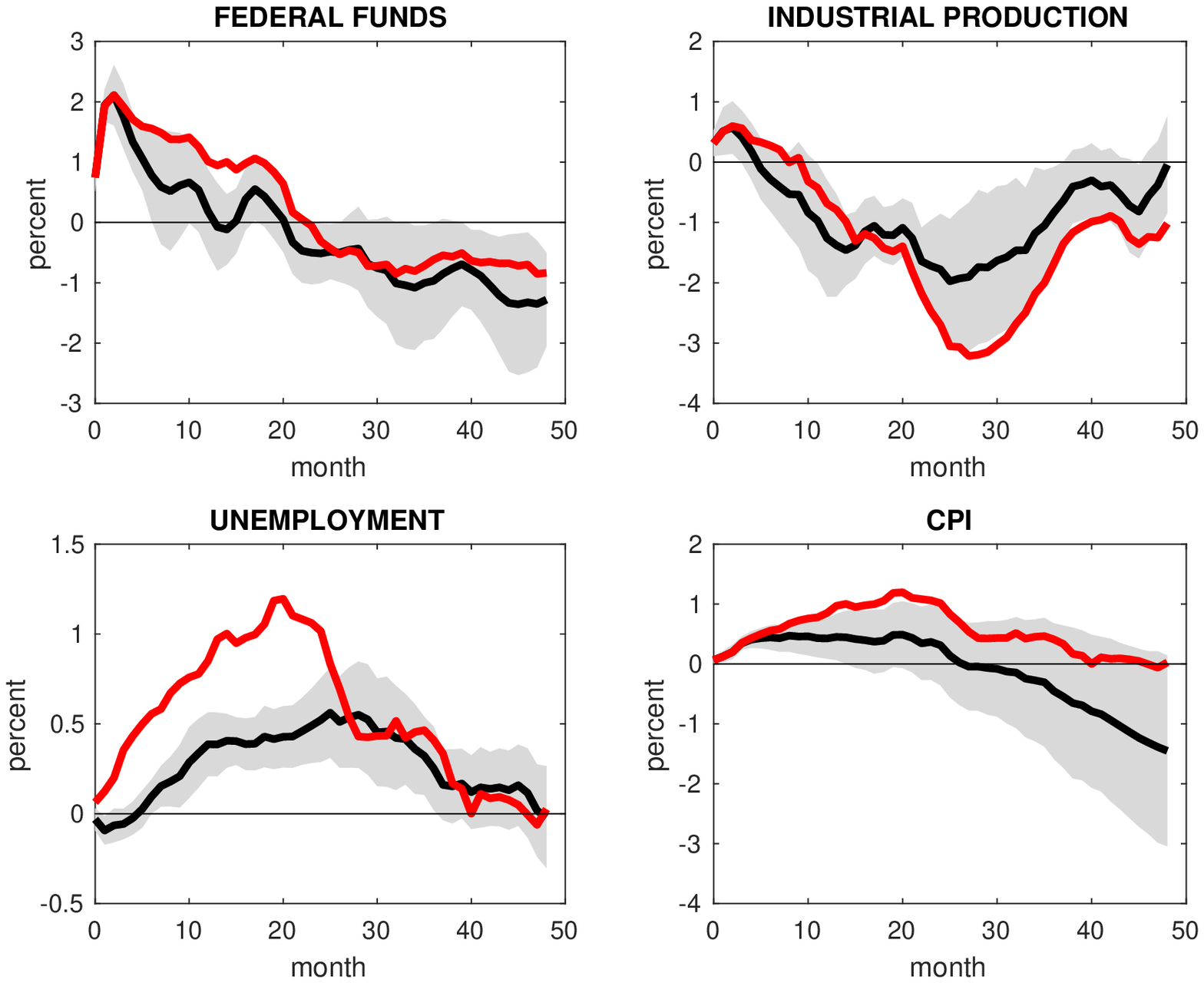}}	
\end{center}
\footnotesize Black solid lines refer to a benchmark specification that relaxes the recursive assumption and does not include leads of the shock. Grey areas show 90\% Newey-West confidence intervals. Red solid lines include $h$ leads of the monetary shocks. 
\end{figure}


\subsection{ \citet{gertler2015monetary}}

In this section we explore another application of the effects of monetary policy shocks based on \citet{gertler2015monetary}. The authors identify exogenous changes in monetary policy by looking at variations in the 3-month-ahead futures of the federal funds within a 30-minute window of a FOMC announcement.\footnote{This scheme is often denoted as High-Frequency Identification (HFI). See \citet{ramey2016macroeconomic} for a comparison with other identification procedures.} By relying on this identification scheme, rather than on standard timing assumptions (e.g., \citet{christiano1999monetary}), the authors are able to explore the effects on measures of financial market frictions or other variables that are often assumed to be contemporaneously invariant to a monetary policy shock.

As shown in Table~\ref{tab:survey}, this measure of monetary policy shock displays some persistence. This was first noted by \citet{ramey2016macroeconomic}, who highlights that the procedure followed by \citet{gertler2015monetary} to convert FOMC shocks (expressed at FOMC frequency) to monthly frequency introduces this serial correlation.\footnote{In particular, \citet{gertler2015monetary}  cumulate the surprises on any FOMC days during the last 31 days, effectively introducing a first-order moving-average structure. This is a variation of the procedure followed by \citet{romer2004measure} and that also results in a measure of monetary policy shocks that displays persistence.}

\citet{gertler2015monetary} embedded the identified monetary  policy shocks in a VAR, using the measure of monetary policy surprises as an instrument of the residuals in the VAR. Here we explore what consequences the persistence of the shock might have if the researcher were to use standard LPs (and estimate  $\mathcal{R}(h)$).

To do so, we implement the following specification, suggested by  \citet{ramey2016macroeconomic}:

\begin{equation} \label{eq:GK2015}
 y_{t+h}= \beta_{h,0} shock_{t} + \bm{ \theta_h(L)  {x_t} }+ \beta_{h,f} \sum_{f=1}^{min\{h,12\}} shock_{t+f} + \xi_{t+h},
\end{equation}

 \noindent where $y_t$ is either the 1-year government bond rate, the log of industrial production, the excess bond premium spread from \citet{gilchrist2012spreads}, or the log of consumer price index, and $shock_t$ is the measure of monetary policy shocks from \citet{gertler2015monetary}. The regressions also include a set of controls $\bm x_t$, with two lags and the contemporaneous values of all dependent variables, and two lags of the shock.\footnote{Note that while the inclusion of lags of the shocks are meant to account for persistence in the shock, our analysis from Section~\ref{sec:econometric} shows that they are not effective for this role. In our results, the inclusion of lags of the shock did not have any noticeable effect. We keep them here in order to replicate the results from \citet{ramey2016macroeconomic}.  } Following \citet{ramey2016macroeconomic}, we estimate equation~\eqref{eq:GK2015} for a sample of 1991m1-2012m6.\footnote{ \citet{gertler2015monetary} proceed in two steps: they first estimate the dynamic coefficients and residuals from a VAR during the period 1979-2012. Then they estimate the contemporaneous effects of monetary policy using both the residuals from the previous step and the monetary policy instrument in a proxy VAR during 1991-2012.} Given this reduced sample, we limit the number of leads introduced in the estimation to a maximum of 12 (i.e., we use $h$ leads for $h <12$ and $12$ leads for longer horizons).\footnote{We also preserve the sample at the end of the period until 2012m06 (which will be otherwise reduced when including leads) by considering values of the leads of the shock equal to 0 for the last 12 periods. Although this is not important for our results, it allows us to compare our estimates to those from \citet{ramey2016macroeconomic}. }

The results of these estimations are shown in  Figure~\ref{fig:GK2015} for two cases: setting $\beta_{h,f} = 0$ $\forall h,f$ (black solid lines with 68 and 90\% confidence intervals) and allowing $\beta_{h,f} \neq 0$ (red lines). Benchmark estimations that do not account for persistence reproduce the results from \citet{ramey2016macroeconomic} (Figure 3B). She notes that LPs using directly the \citet{gertler2015monetary} instrument as an explanatory variable give rise to puzzling results, namely: a sluggish response of the policy rate, output increases after the monetary expansion, and the credit spread and prices do not show significant dynamics during most of the response horizon. \citet{ramey2016macroeconomic} suggests that the persistence exhibited by \citet{gertler2015monetary} and potential predictability of the series could be sources of concern. When we incorporate leads of the shock in the estimation of equation~\eqref{eq:GK2015} (red lines in Figure~\ref{fig:GK2015}) we see that both the response of the government bond rate and industrial production seem to be overestimated when persistence is not accounted for (by contrast, the results do not change much for the excess bond premium and prices).\footnote{Alternatively, \citet{ramey2016macroeconomic}  concludes that these differences may be due to the fact that the reduced-form parameters (used to construct the impulse responses) are estimated for a longer sample (1970-2012 instead of 1991-2012) or to potential misspecification of the original VAR estimates due to the rising importance of forward guidance, which may lead to a problem of non-fundamentalness in the VAR. } 

\begin{figure}
\caption{Responses to monetary policy shock from  \citet{gertler2015monetary}, with and without leads}\label{fig:GK2015}
\begin{center}
		{\includegraphics[width=10cm]{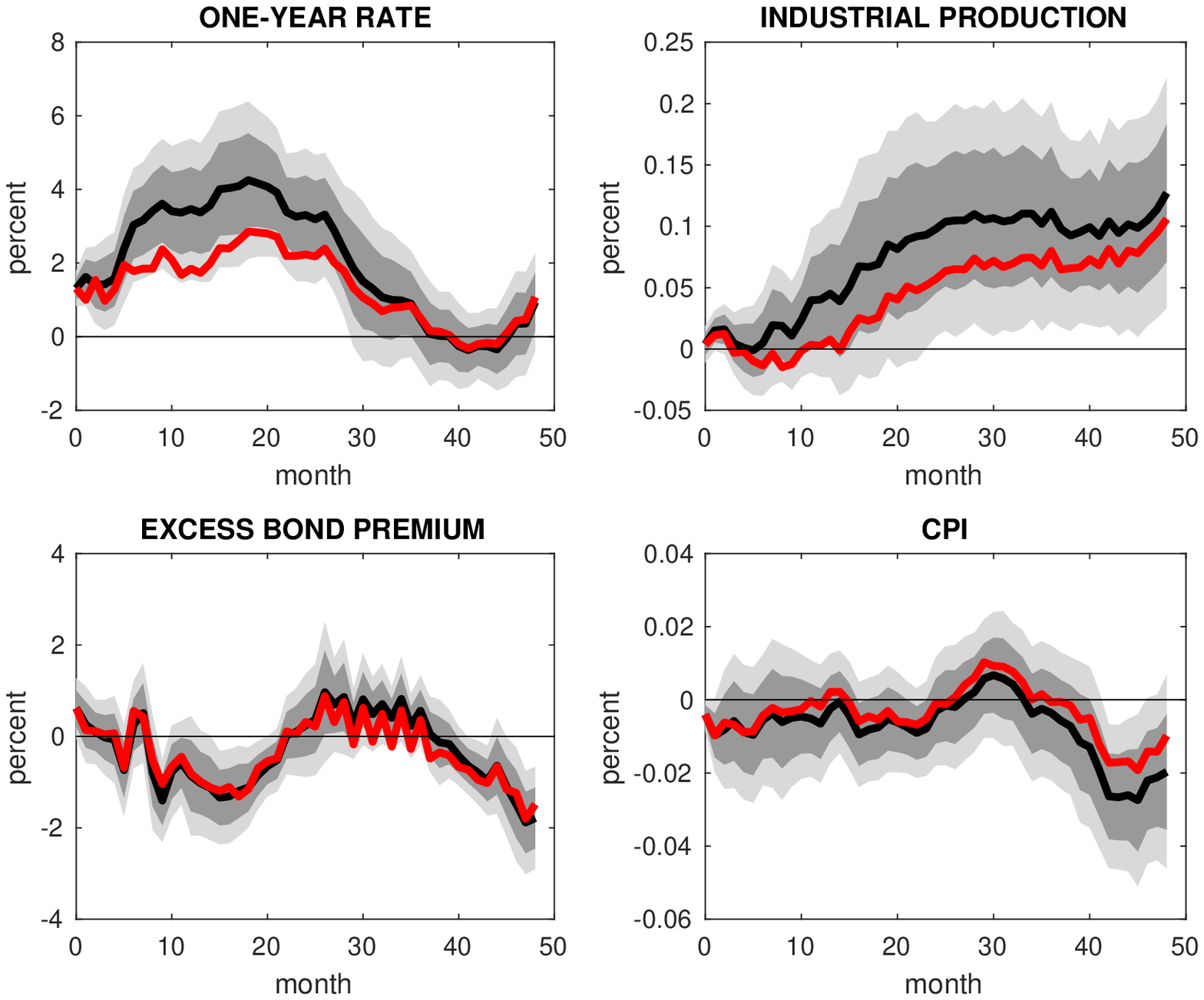}}	
\end{center}
\footnotesize Black solid lines refer to a benchmark specification that does not include leads of the shock. Grey areas show 68 and 90\% Newey-West confidence intervals. Red solid lines include $h$ leads of the monetary shocks up to horizon $h=12$, after then, the number of leads is kept to $12$. 
\end{figure}

\newpage
\section{Additional Tables and Figures}\label{app:morefigs}
\setcounter{figure}{0}
\setcounter{table}{0}
\setcounter{equation}{0}
\setcounter{footnote}{0}
\renewcommand{\thefigure}{D\arabic{figure}}
\renewcommand{\thesubsection}{D.\arabic{subsection}}
\renewcommand{\theequation}{D.\arabic{equation}}
\renewcommand{\thetable}{D.\arabic{table}}

\begin{table}[ht]
\caption{Robustness: different lag structures for  tests} 
\label{tab:survey_rob}
\begin{center}
{\small
\begin{tabular}{ ccccccc } 
 \hline
 \\[-0.5em]
&5 lags&10 lags&20 lags&40 lags&60 lags  \\ 
\\[-0.5em]
\hline
\\[-0.5em]
\citet{arezki2017news}&175.944&175.953&176.049&177.903&177.907\\ \vspace{0.2cm}
&(0.000)&(0.000)&(0.000)&(0.000)&(0.000)\\ 
\citet{cloyne2013tax}&11.365&21.521&40.041&98.751&120.270\\ \vspace{0.2cm}
&(0.045)&(0.018)&(0.005)&(0.000)&(0.000)\\
\citet{cloyne2016monetary}&17.723&20.771&47.357&84.422&103.001\\ \vspace{0.2cm}
&(0.003)&(0.023)&(0.001)&(0.000)&(0.001)\\
\citet{gertler2015monetary}&53.802&84.284&106.133&124.568&131.030\\ \vspace{0.2cm}
&(0.000)&(0.000)&(0.000)&(0.000)&(0.000)\\
\citet{guajardo2014expansionary}&160.740&173.315&182.866&185.810&185.810\\ \vspace{0.2cm}
&(0.000)&(0.000)&(0.000)&(0.000)&(0.000)\\
\citet{ramey2018government}&79.298&89.916&104.414&182.950&190.974\\ \vspace{0.2cm}
&(0.000)&(0.000)&(0.000)&(0.000)&(0.000)\\
\citet{romer2004measure}&15.536&23.965&43.824&53.758&64.576\\ \vspace{0.2cm}
&(0.008)&(0.008)&(0.002)&(0.072)&(0.320)\\
\citet{romer2010macroeconomic}&1.578&3.080&6.562&19.023&24.783\\ \vspace{0.2cm}
&(0.904)&(0.980)&(0.998)&(0.998)&(1.000)\\
 \hline
\end{tabular}
 }
\end{center}
\vspace{-0.2cm}
 {\footnotesize The columns report the values of a \citet{box1970distribution} test (with \citet{ljung1978measure} correction) including different lags. P-values are shown in brackets.  in \citet{arezki2017news} and \citet{guajardo2014expansionary} is tested using a generalized version of the autocorrelation test proposed by \citet{arellano1991tests} that specifies the null hypothesis of no autocorrelation at a given lag order.}
\end{table}

\begin{figure}[!htb]
\caption{Autocorrelograms} 
\label{fig:ac}
\begin{center}
\begin{tabular}{ cc } 
Panel A: \citet{cloyne2013tax} & Panel B: \citet{cloyne2016monetary} \\
{\includegraphics[width=7.5cm]{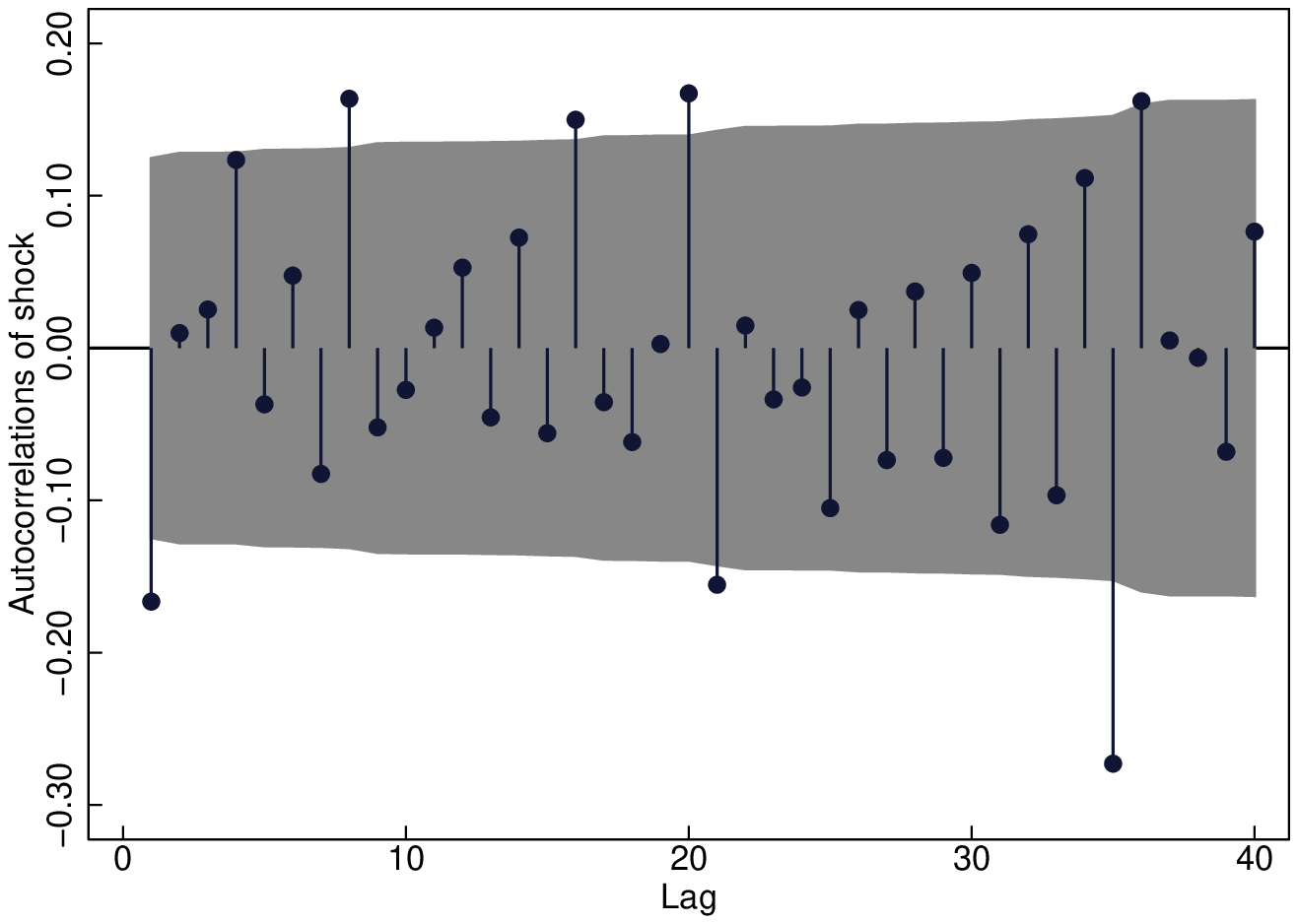}}	 & 	{\includegraphics[width=7.5cm]{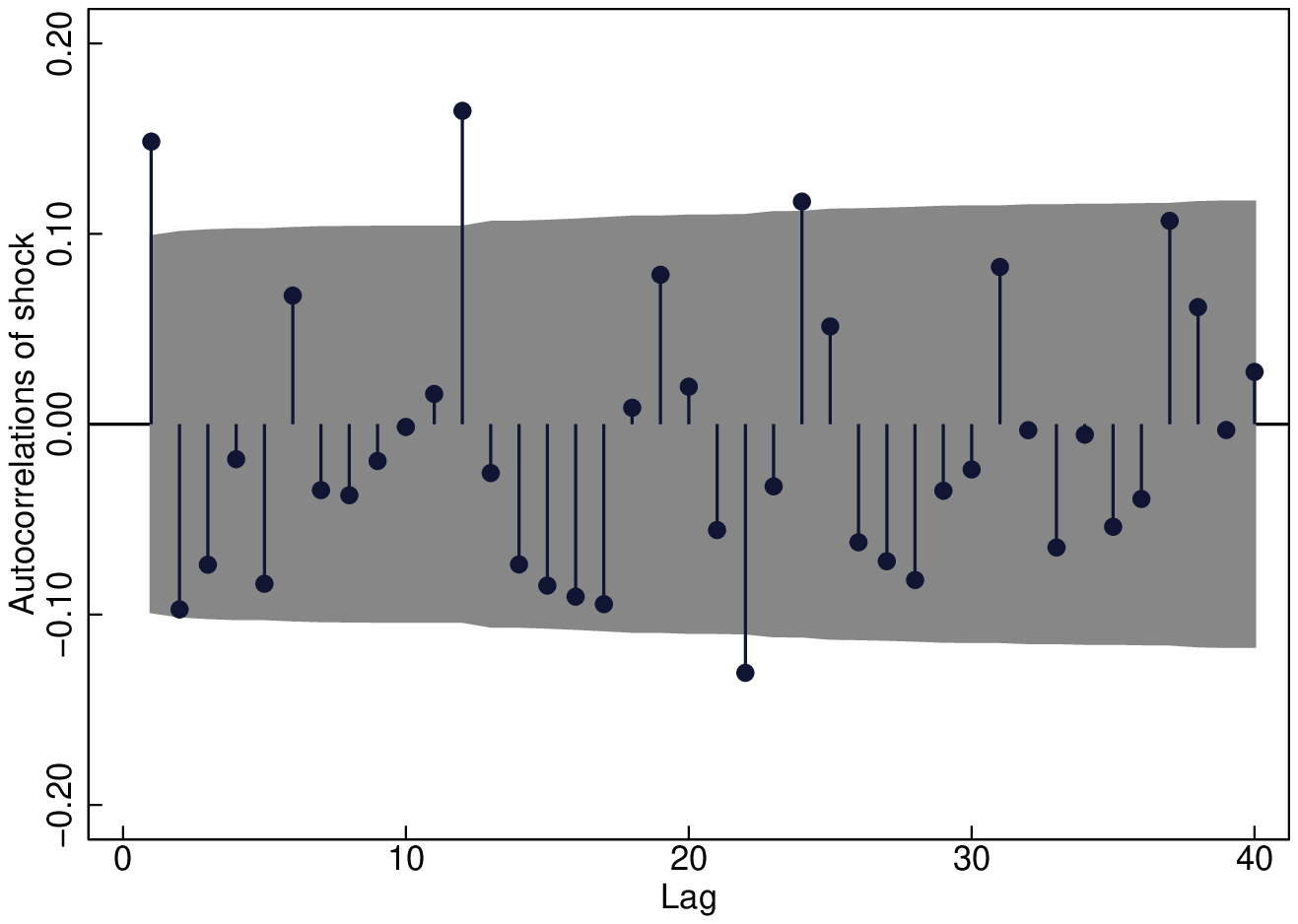}}	\\ 
Panel C: \citet{gertler2015monetary} & Panel D: \citet{ramey2018government} \\
{\includegraphics[width=7.5cm]{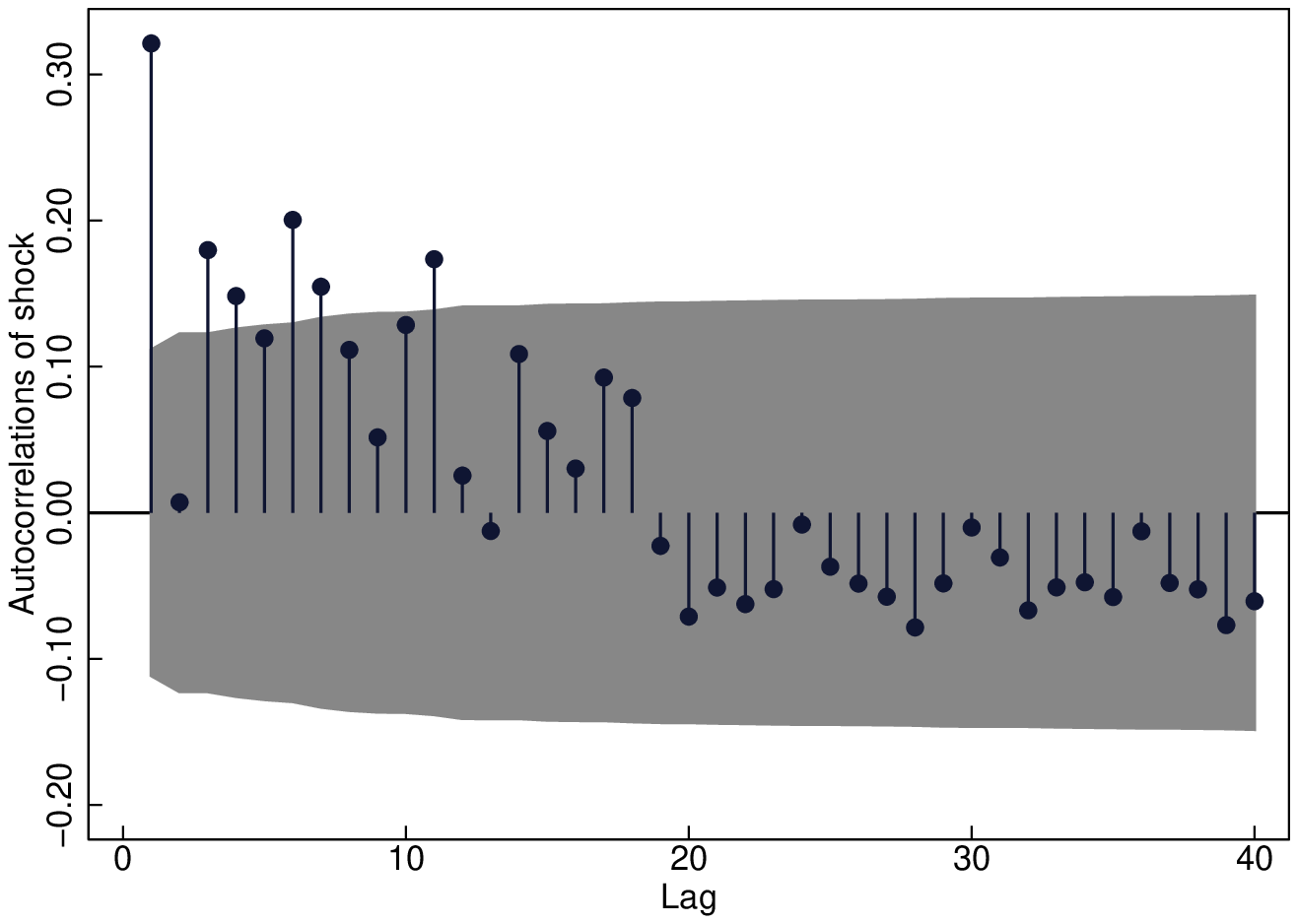}} & {\includegraphics[width=7.5cm]{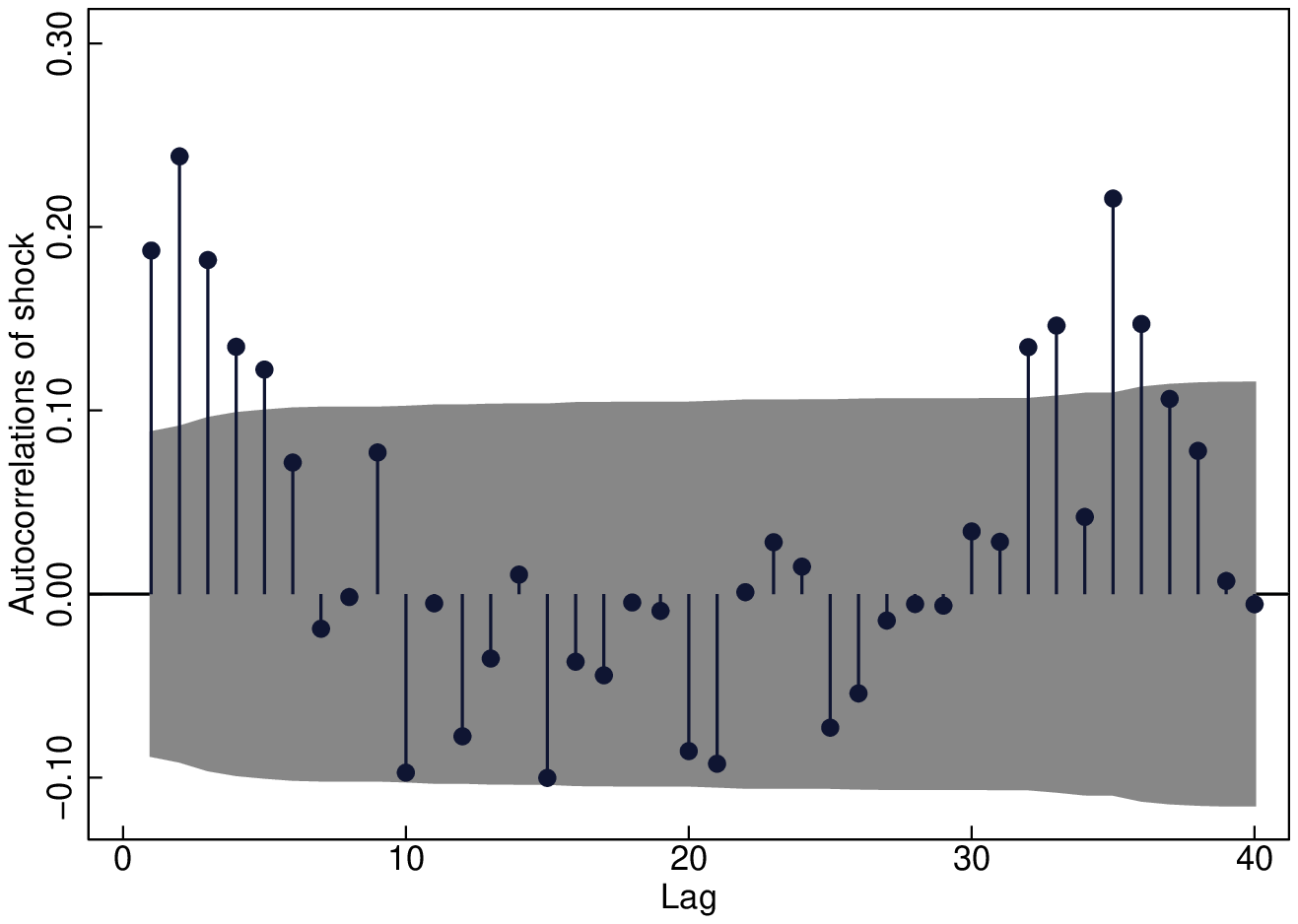}} \\
Panel E: \citet{romer2004measure} & Panel F: \citet{romer2010macroeconomic} \\
{\includegraphics[width=7.5cm]{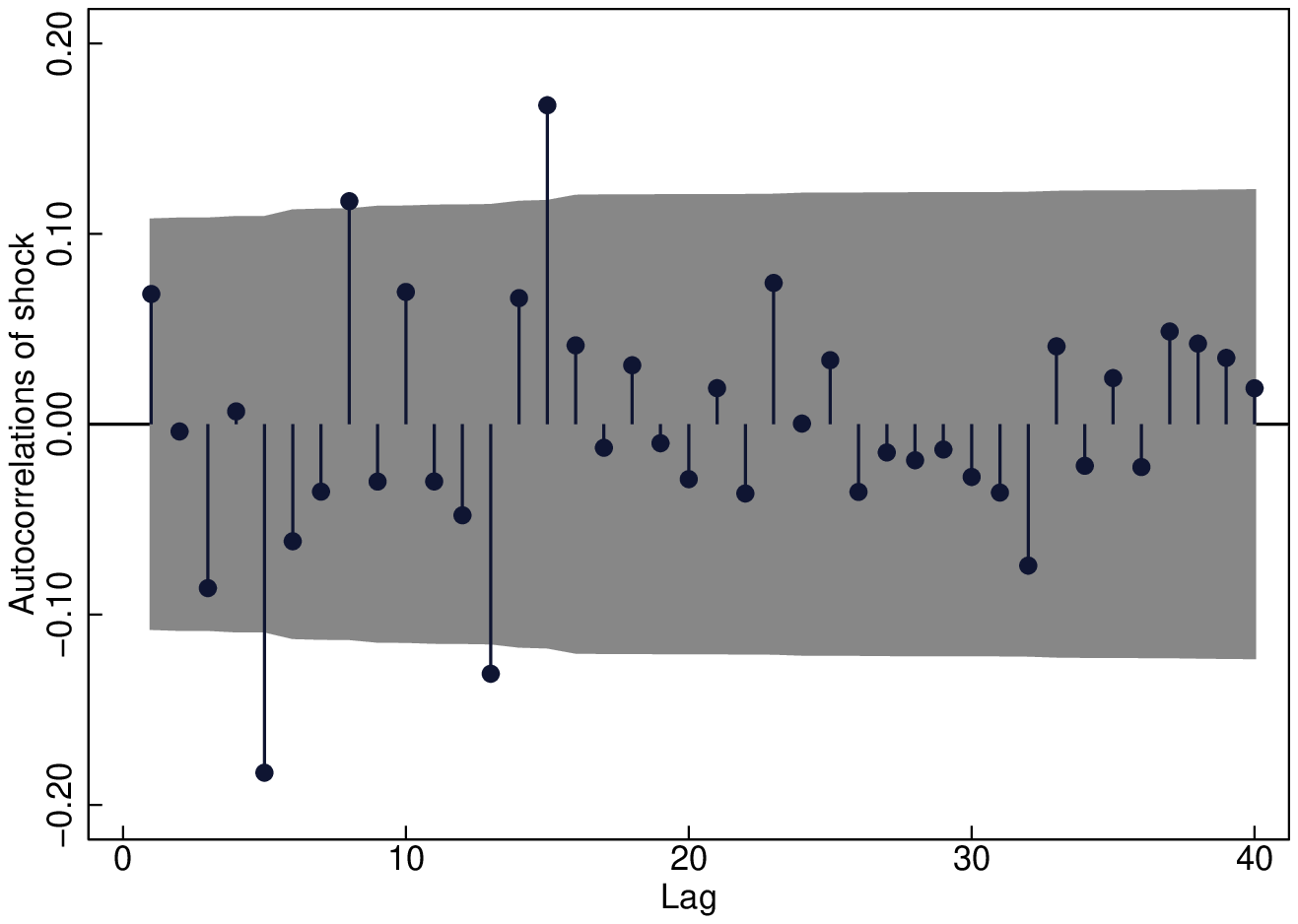}} & {\includegraphics[width=7.5cm]{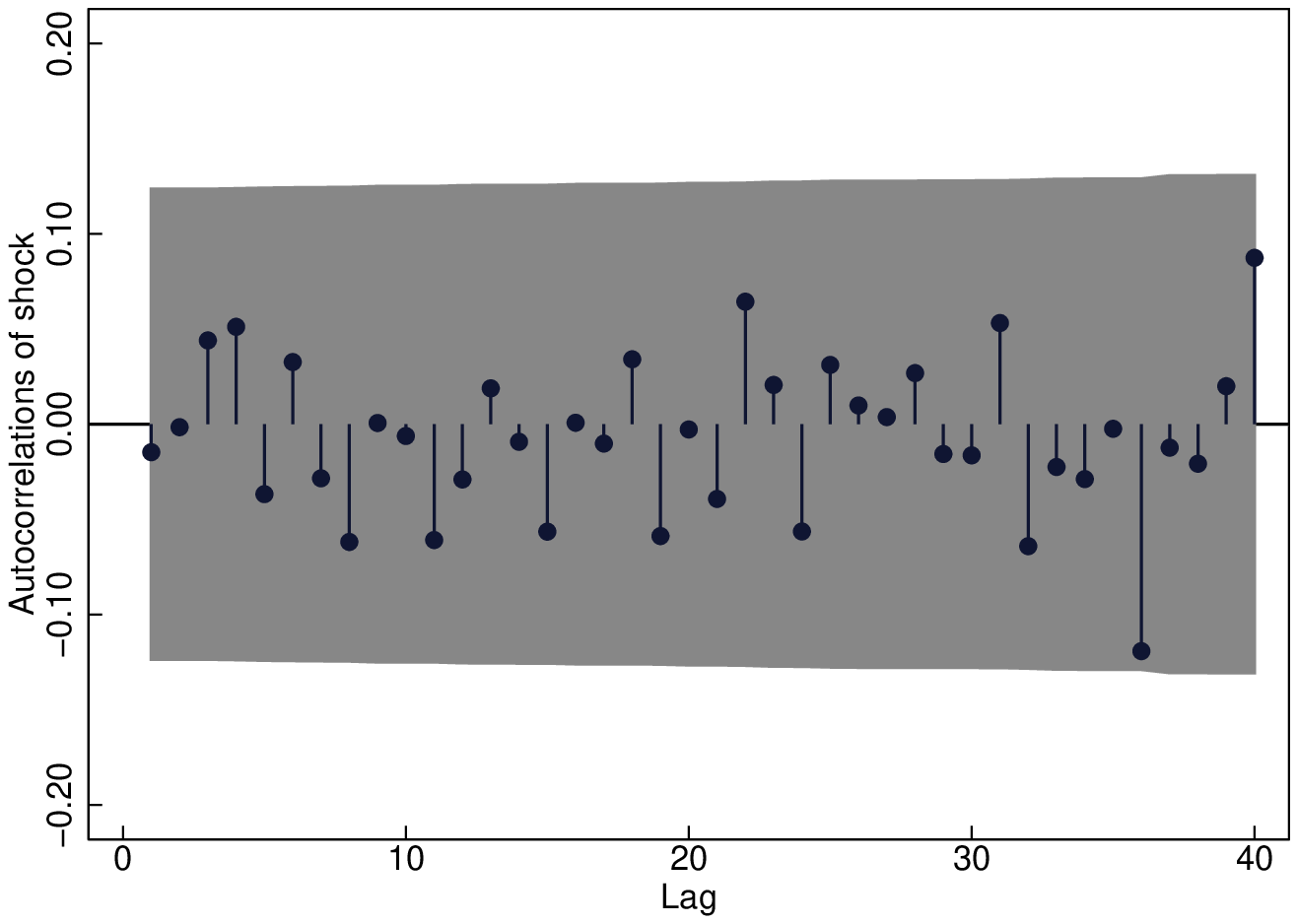}} \\
\end{tabular}
\end{center}
 {\footnotesize 95\% confidence intervals computed using Bartlett's formula for MA(q) processes. }
\end{figure}

\begin{figure}[!htb]
\caption{Time series of shocks} 
\label{fig:series}
\begin{center}
\begin{tabular}{ cc } 
Panel A: \citet{cloyne2013tax} & Panel B: \citet{cloyne2016monetary} \\
{\includegraphics[width=7.5cm]{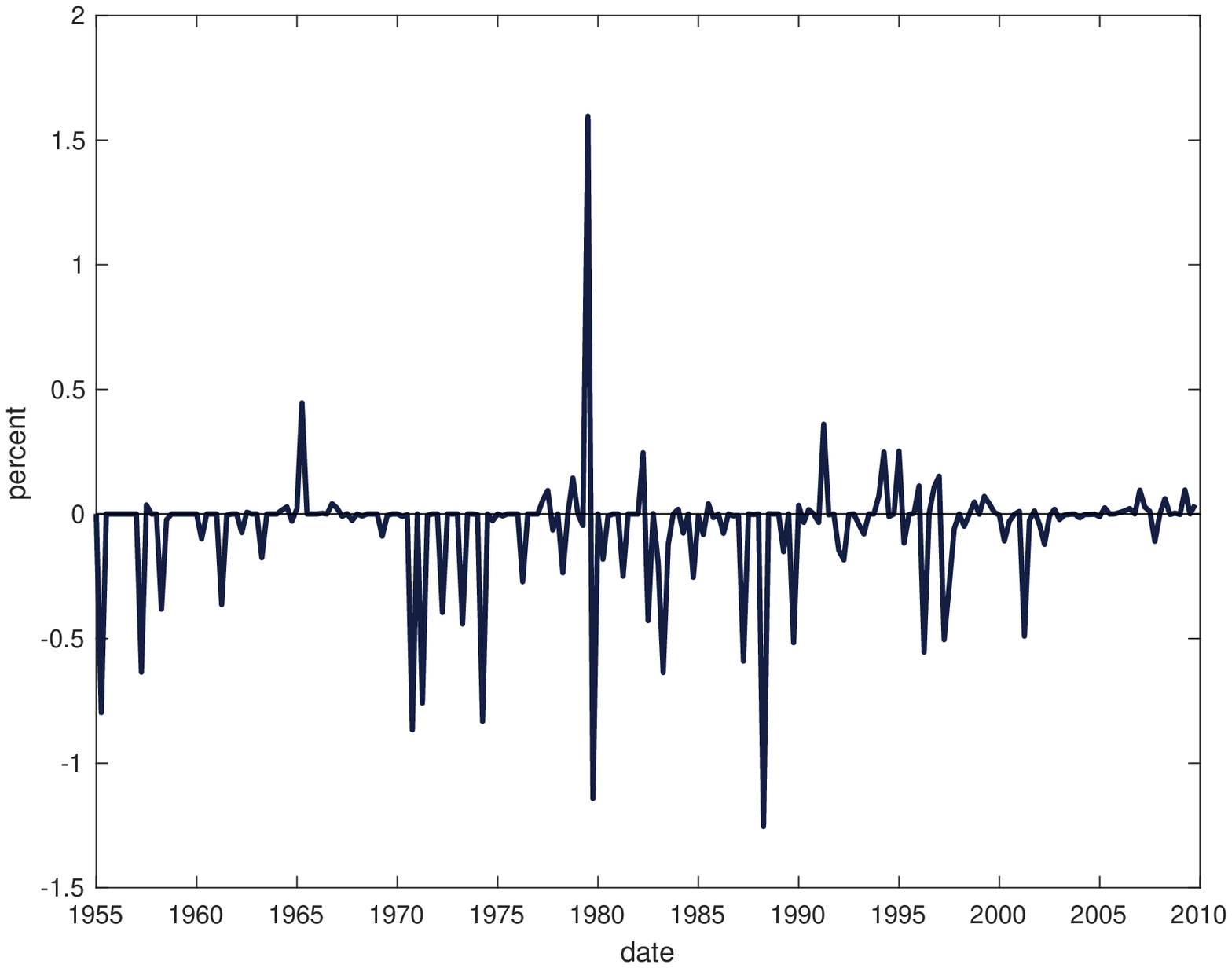}}	 & 	{\includegraphics[width=7.5cm]{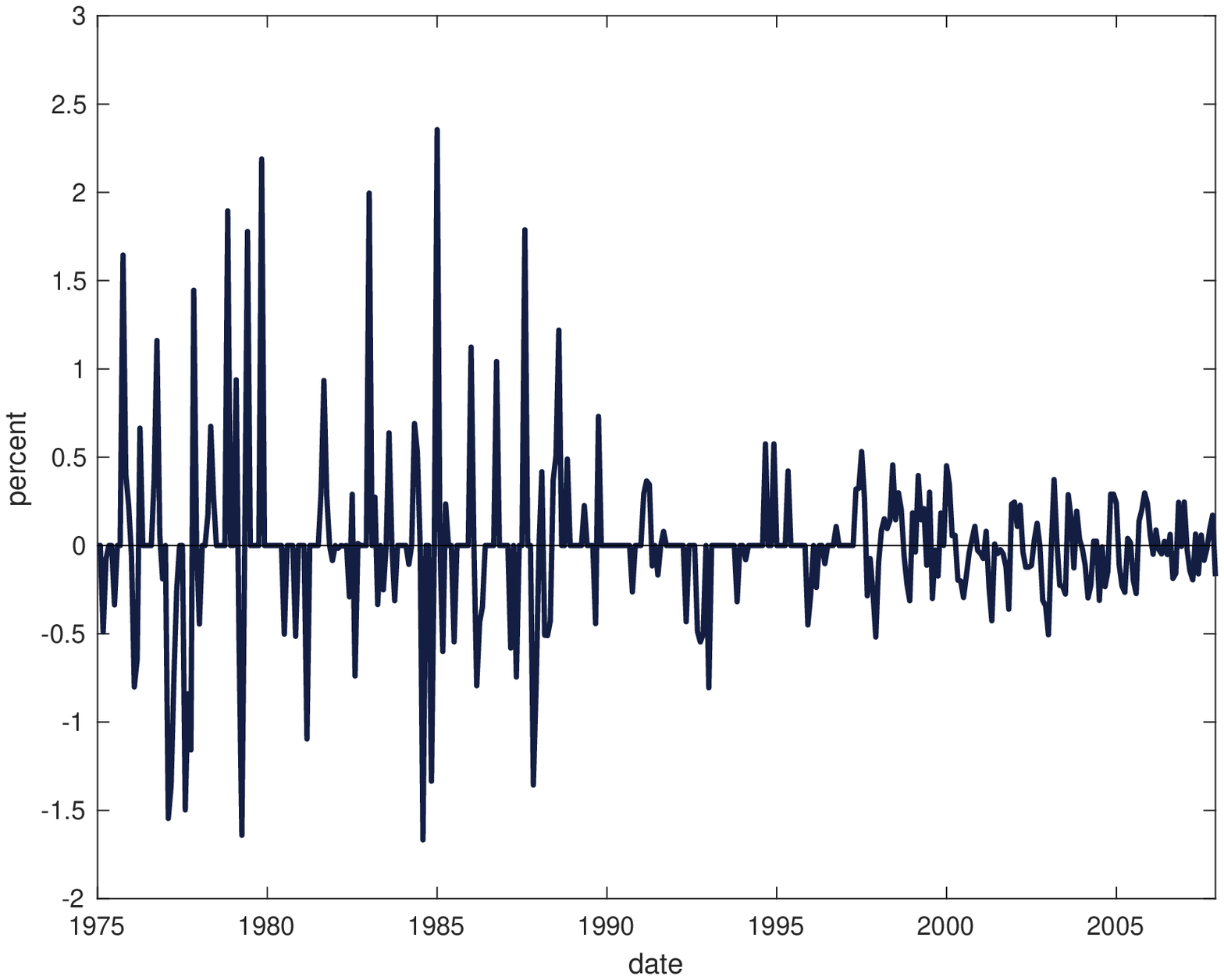}}	\\ 
Panel C: \citet{gertler2015monetary} & Panel D: \citet{ramey2018government} \\
{\includegraphics[width=7.5cm]{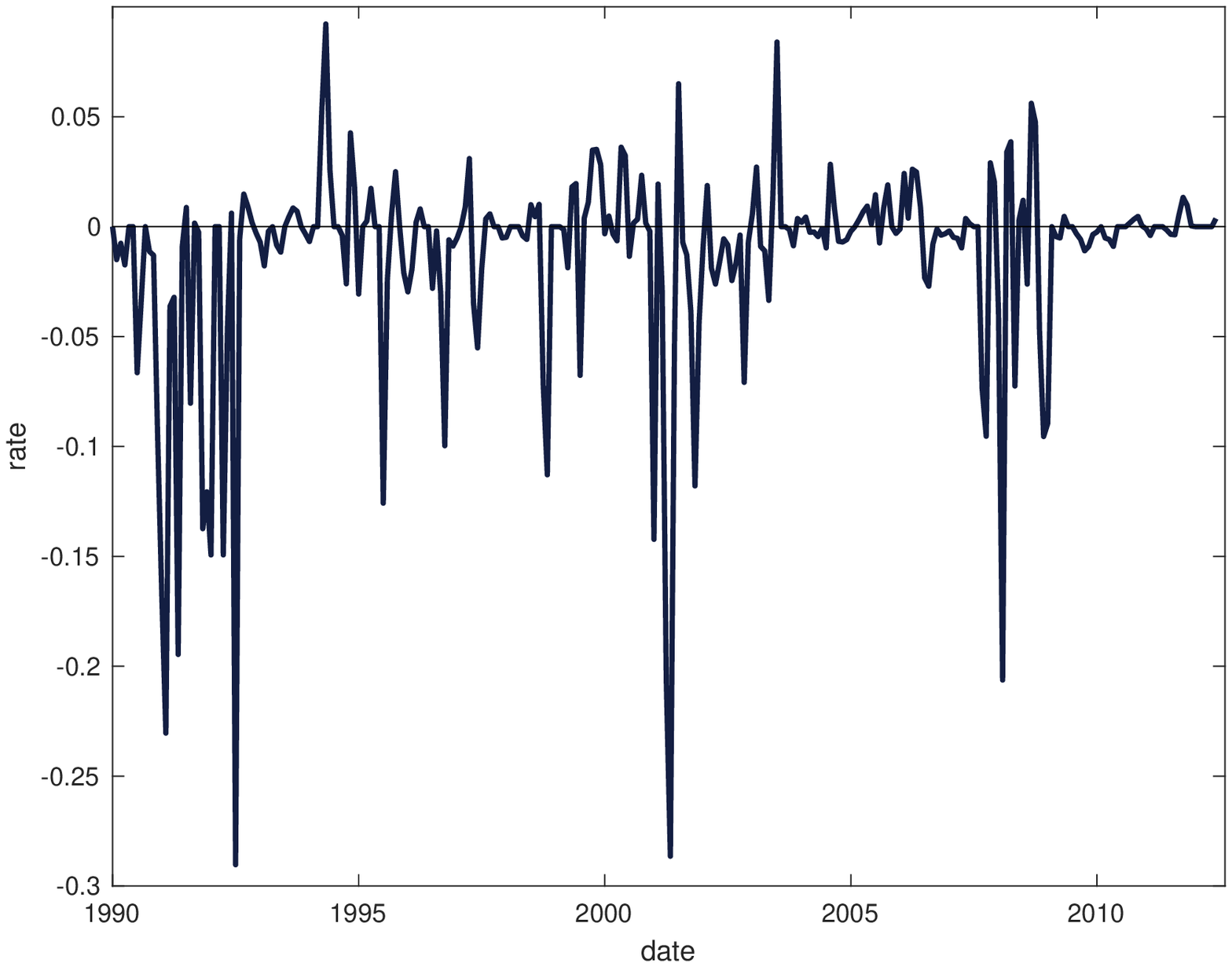}} & {\includegraphics[width=7.5cm]{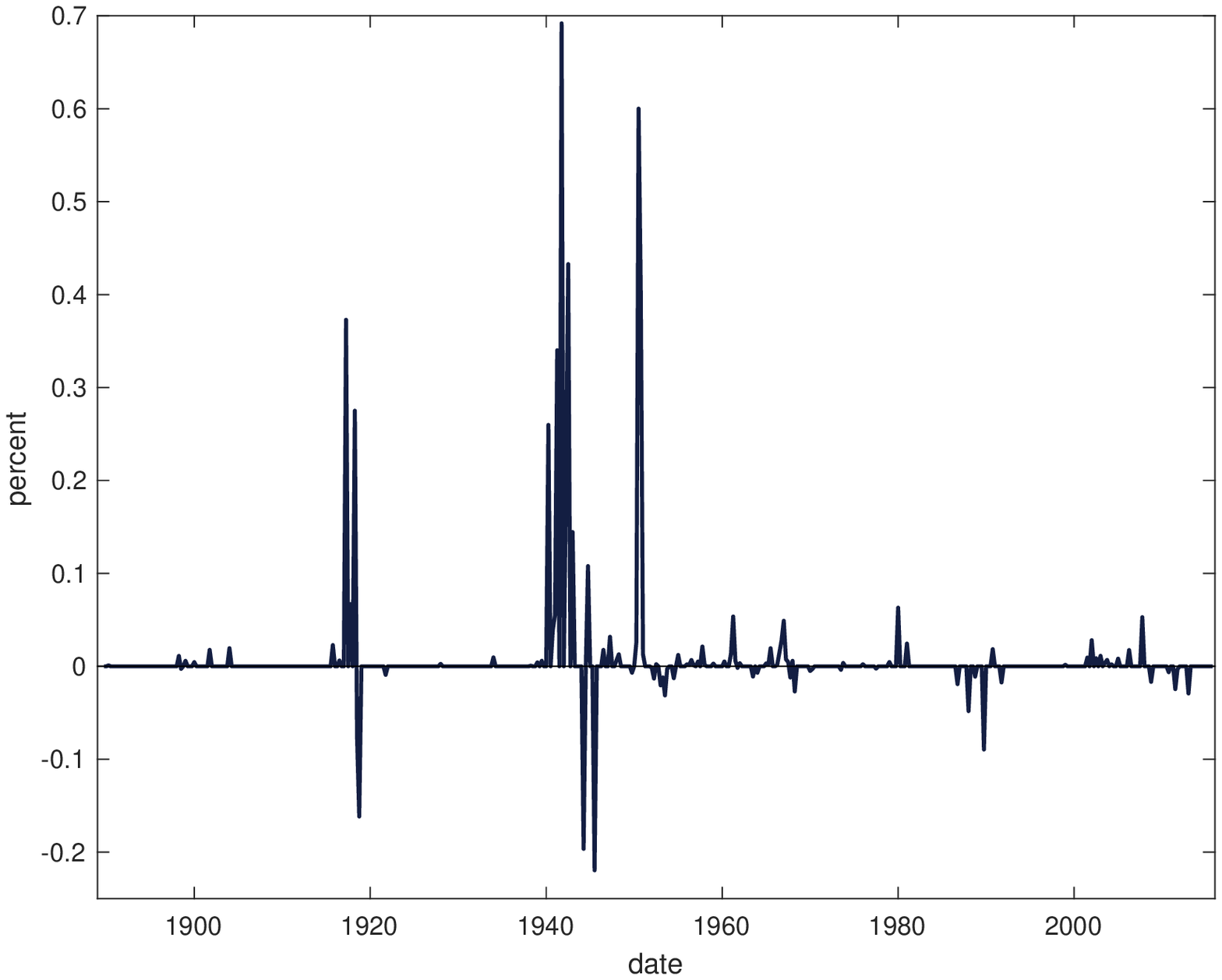}} \\
Panel E: \citet{romer2004measure} & Panel F: \citet{romer2010macroeconomic} \\
{\includegraphics[width=7.5cm]{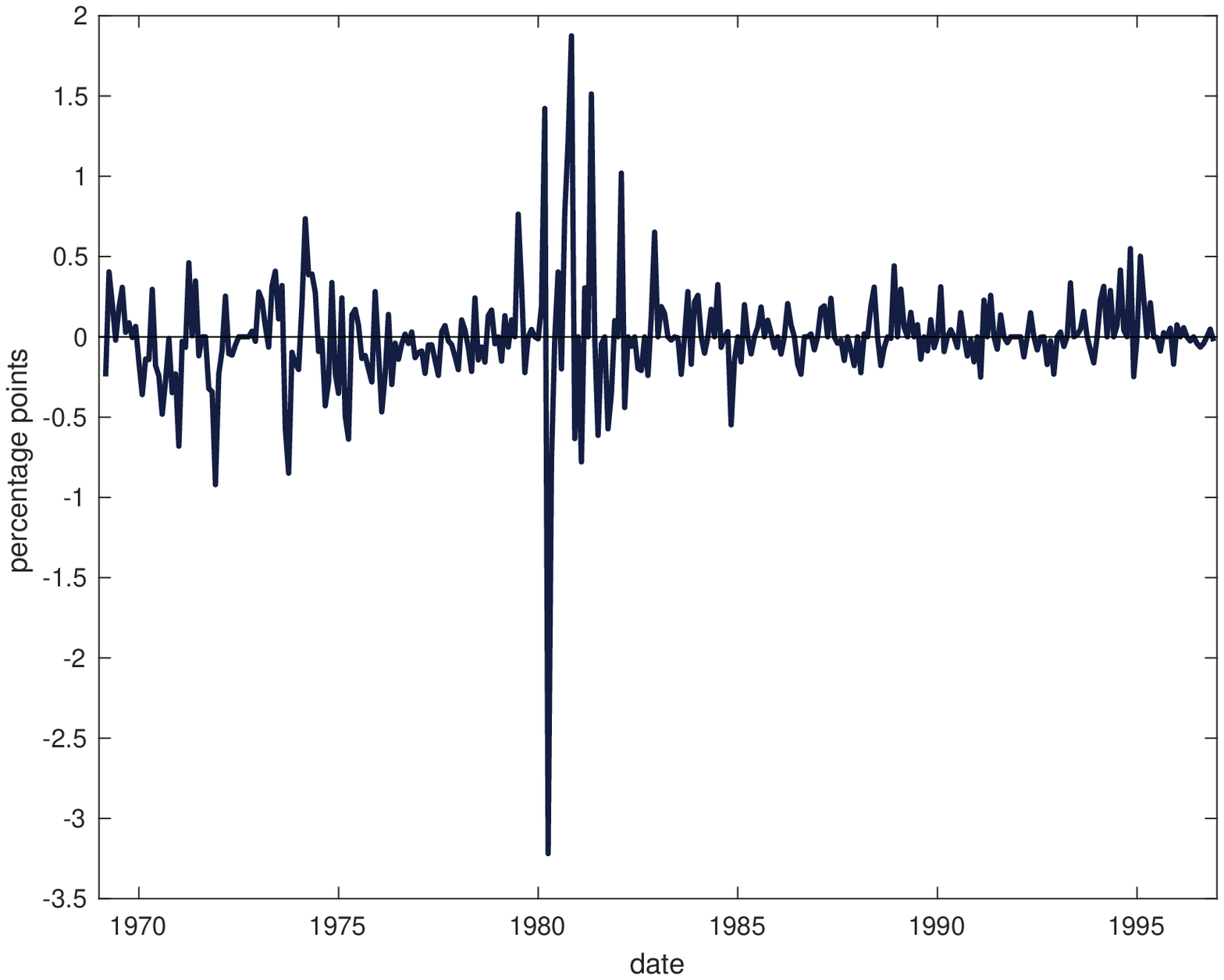}} & {\includegraphics[width=7.5cm]{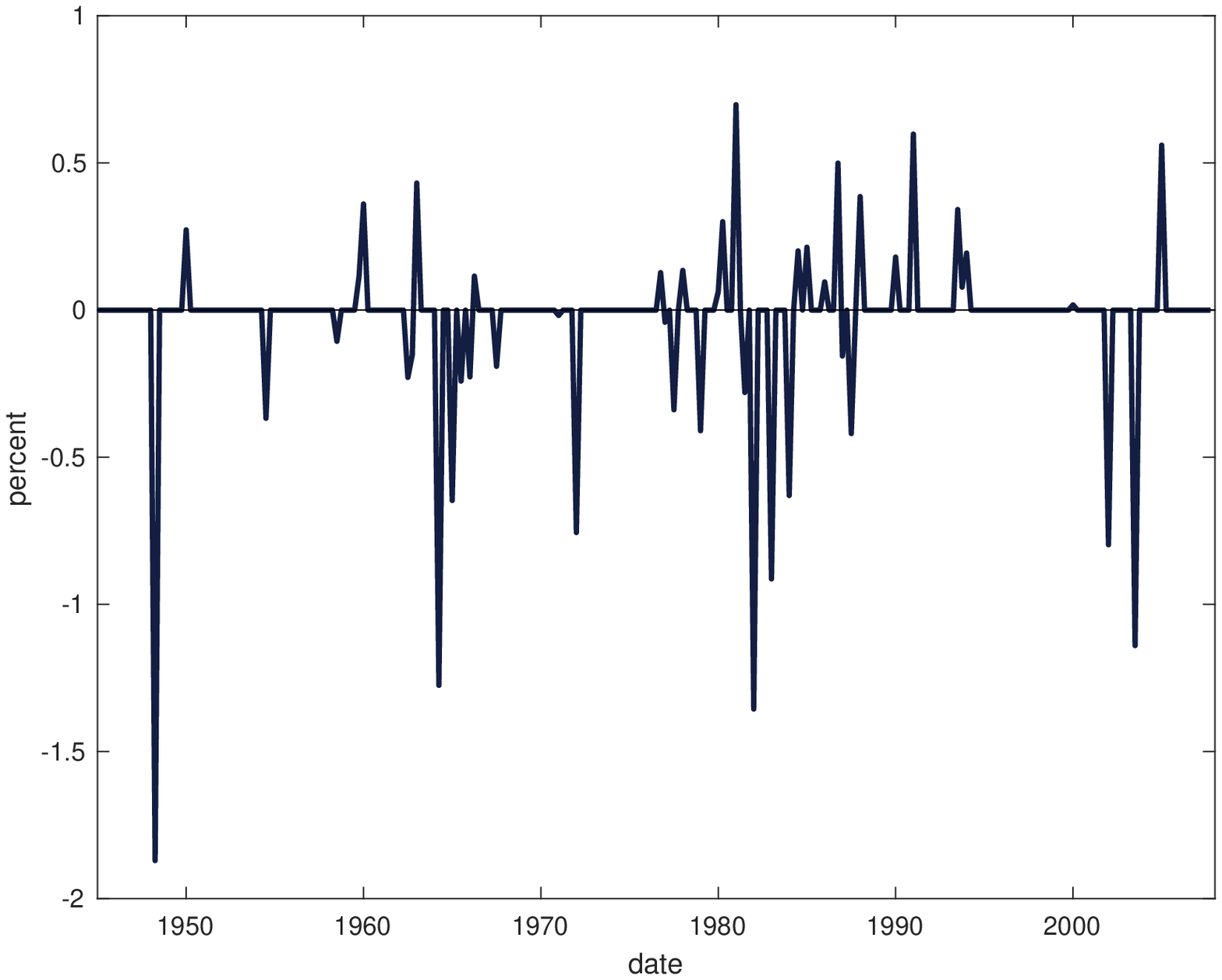}} \\
\end{tabular}
\end{center}
 {\footnotesize  }
\end{figure}

\begin{figure}[!htb]
\caption{Output and government spending responses, with and without leads} \label{fig:Ramey_responses_doubleCI}
\begin{center}
		{\includegraphics[width=10cm]{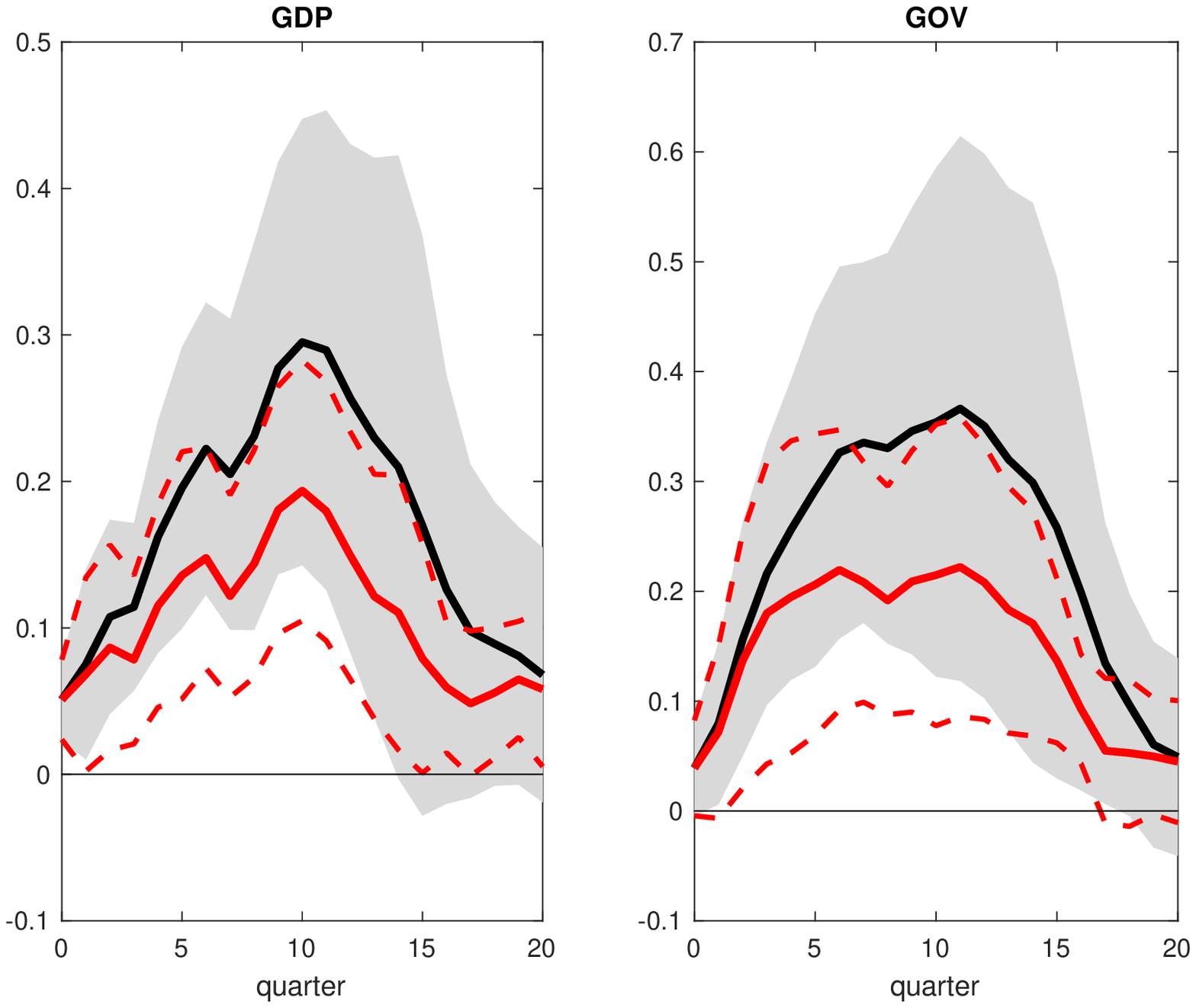}}	
\end{center}
\footnotesize Black lines show the results of estimating the system~\eqref{eq:RameyLP} without including any lead (as in \citet{ramey2018government}). Grey areas represent 68 and 95\% Newey-West confidence intervals for these estimates. Red solid lines represent the results of estimations when including $h$ leads of the \citet{ramey2018government} news variable (with 95\% Newey-West confidence intervals).
\end{figure}

\begin{figure}[!htb]
\caption{Government spending multiplier, with and without leads} \label{fig:Ramey_multiplier}
\begin{center}
		{\includegraphics[width=10cm]{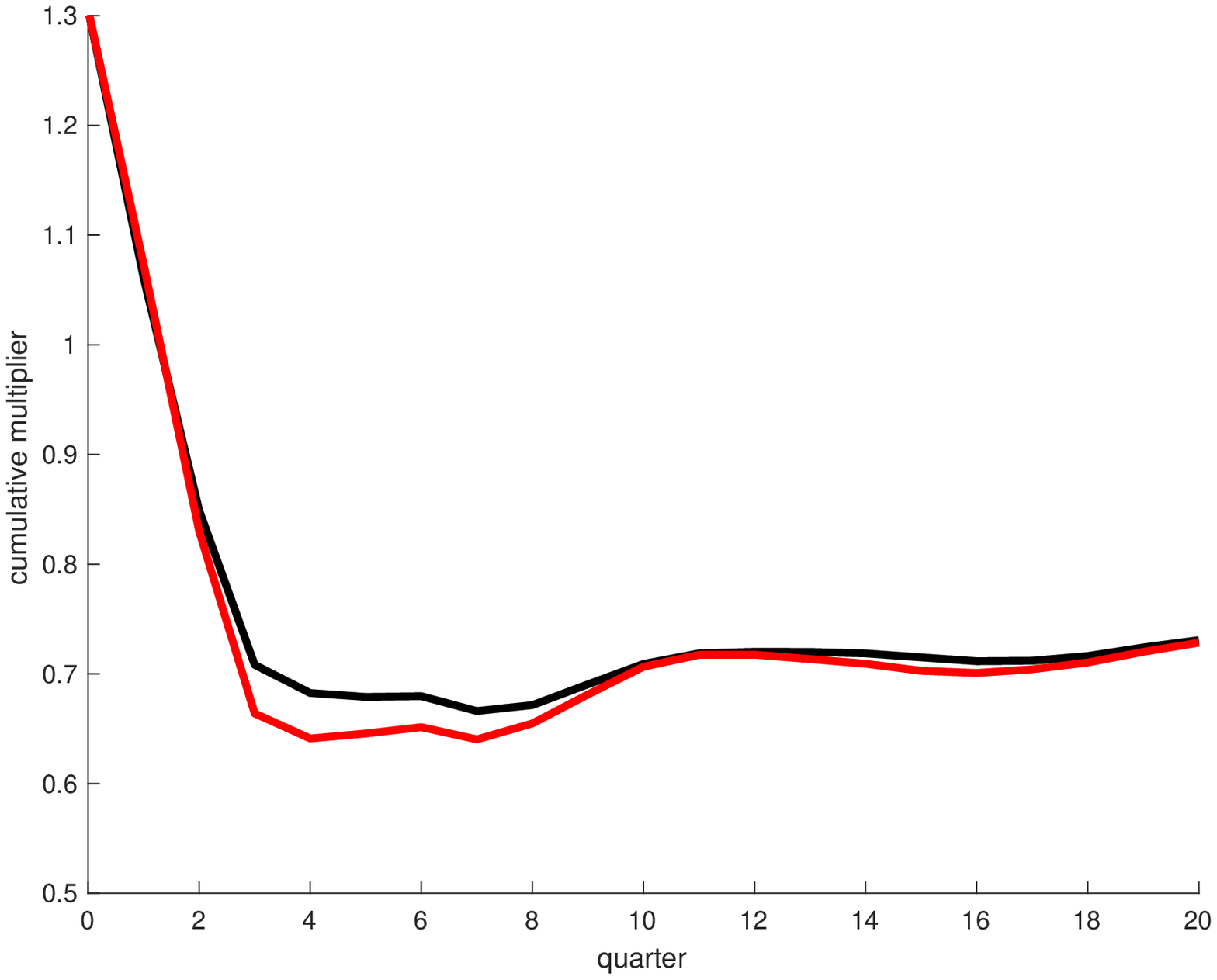}}	
\end{center}
\footnotesize Black lines show the cumulative multiplier without including any lead. Red solid lines represent the estimates of the cumulative multiplier when including a number of leads of the \citet{ramey2018government}) news variable that increase with the response horizon.
\end{figure}

\begin{figure}[!htb] \caption{Government spending multiplier during expansions and recessions, with and without leads}  \label{fig:G_nonli_mults_withSTATE}
\begin{center}
		{\includegraphics[width=10cm]{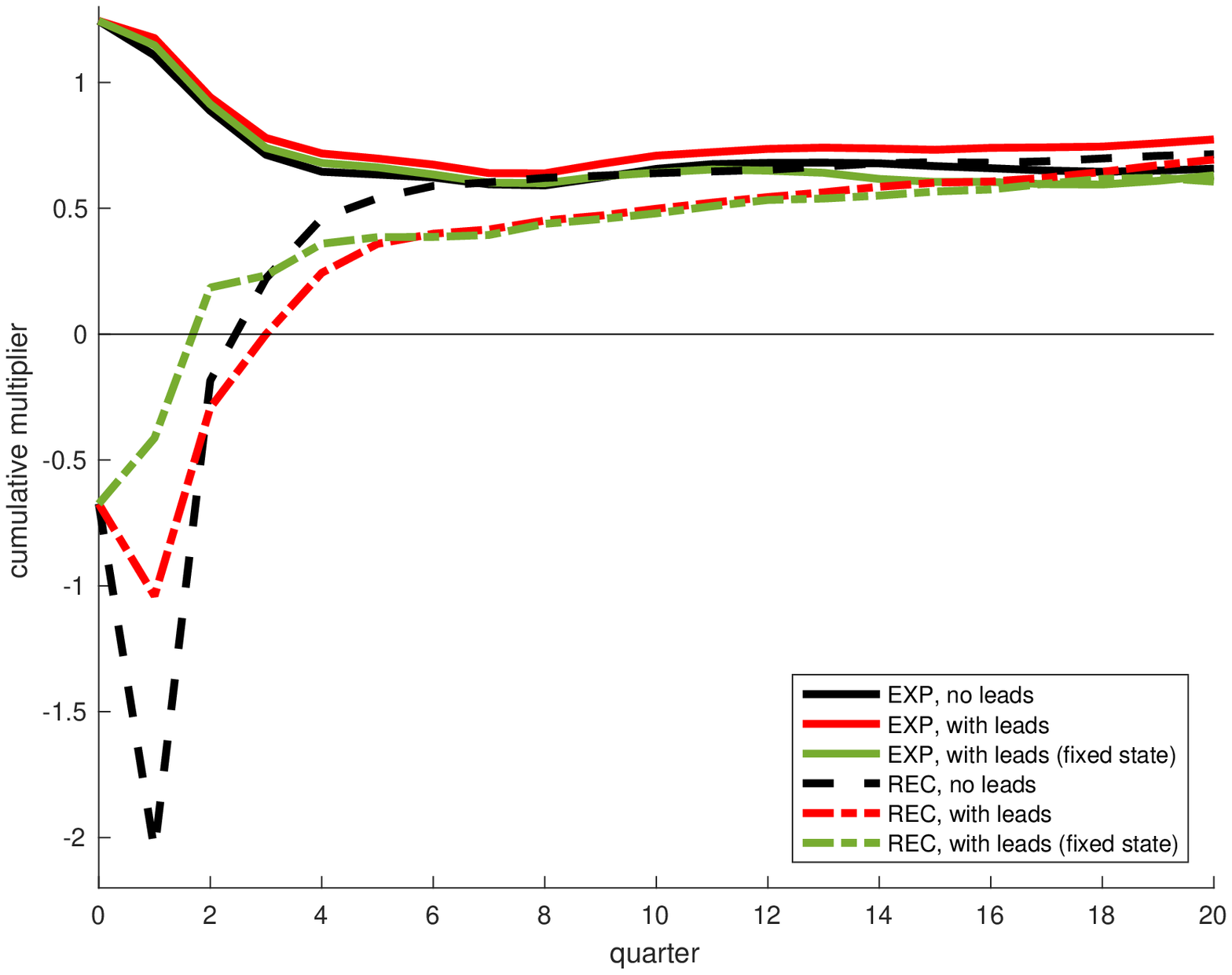}}	
\end{center}
\footnotesize The black solid and dashed lines show the cumulative multiplier during periods of expansion and recession, respectively,  without including any lead (as in \citet{ramey2018government}). The red solid and dashed lines show the cumulative multiplier during periods of expansion and recession, respectively, when including leads of the shocks and the state. Green solid and dashed lines refer to estimates of the expansion and recession multipliers, respectively, when including leads of the shock and the regime.
\end{figure}

\begin{figure}
\caption{Private consumption response to a fiscal consolidation shock, with and without leads} \label{fig:Guajardo_respCONS}
\begin{center}
		{\includegraphics[width=10cm]{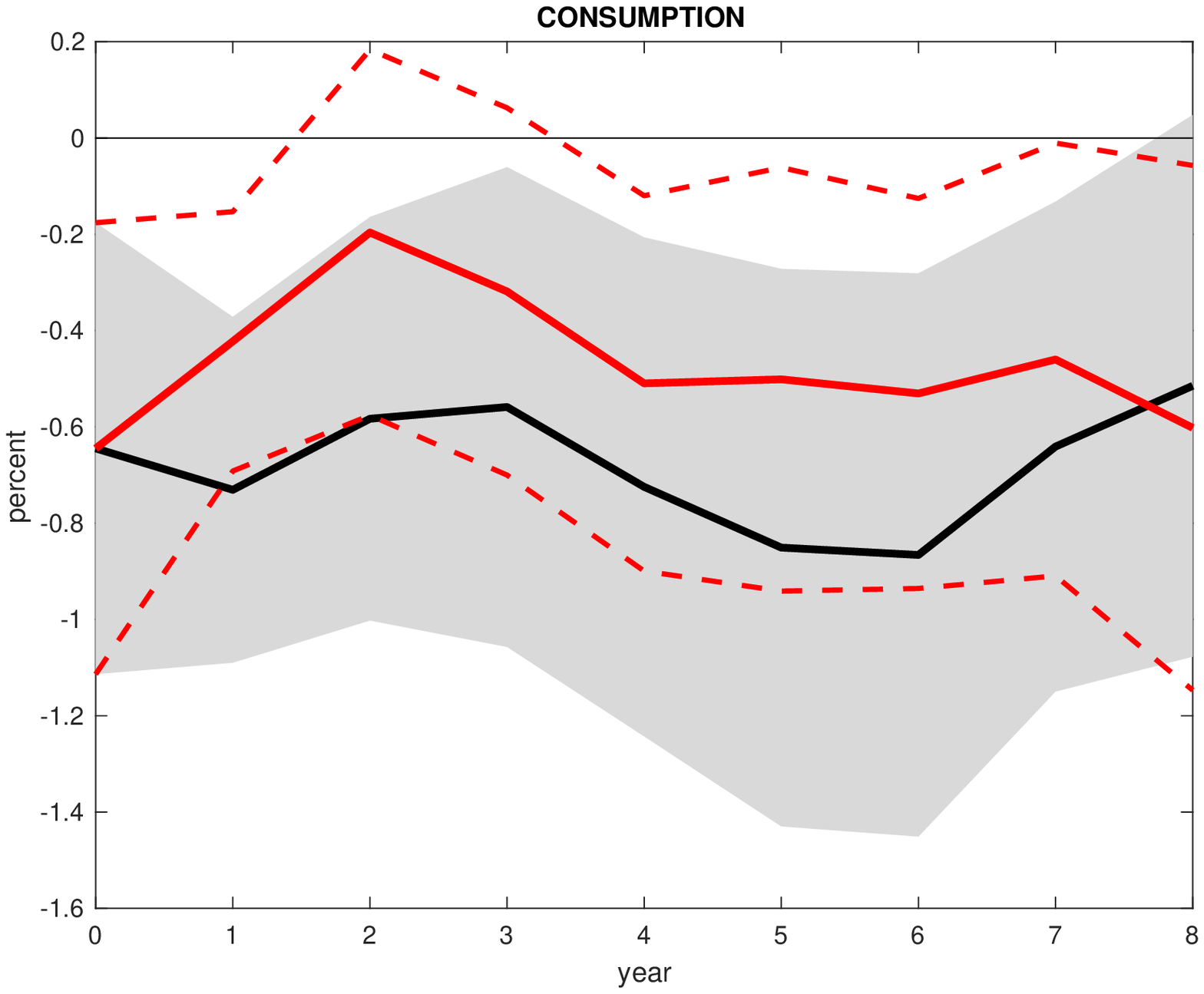}}	
\end{center}
\footnotesize Black lines show the results from equation~\eqref{eq:GLP2016} with private consumption as dependent variable and setting $\beta_{h,f}=0$, i.e., without including any lead  of the shock. Grey areas represent  90\% Newey-West confidence intervals for these estimates (save interval as reported in \citet{guajardo2014expansionary}). Red solid lines represent the results of estimations when allowing  $\beta_{h,f} \neq 0$ and including $h$ leads of the consolidations variable.
\end{figure}

%

\end{appendices}
\end{document}